\documentclass[11pt, letterpaper]{article}

\usepackage[T1]{fontenc}

\usepackage[top=1in, bottom=1in, left=1in, right=1in]{geometry}

\usepackage{amsmath}
\usepackage{amsthm}
\usepackage{amsfonts, mathrsfs}
\usepackage{amssymb}
\usepackage{mathtools}
\usepackage{url}

\usepackage{comment}
\usepackage{enumerate}
\usepackage[shortlabels]{enumitem}
\usepackage{datetime}
\usepackage{xspace}
\usepackage[T1]{fontenc}
\usepackage[utf8]{inputenc}
\usepackage{authblk}
\usepackage{color, hyperref, yfonts}
\hypersetup{
	colorlinks=true,
	linkcolor=blue,
	citecolor= blue
}

    \newtheorem{theorem}{Theorem}[section]
    \newtheorem{lemma}[theorem]{Lemma}
    \newtheorem{proposition}[theorem]{Proposition}
    \newtheorem{corollary}[theorem]{Corollary}
    \newtheorem{claim}[theorem]{Claim}
    \newtheorem{conjecture}[theorem]{Conjecture}
    \newtheorem{problem}[theorem]{Problem}

    \newtheorem{observation}[theorem]{Observation}    
    
    \theoremstyle{definition}
    \newtheorem{definition}[theorem]{Definition}
    \newtheorem{access}[theorem]{Access}
    
	\newtheorem{notation}[theorem]{Notation}
    \newtheorem{remark}[theorem]{Remark}
    \newtheorem{fact}[theorem]{Fact}
    \newtheorem{convention}[theorem]{Convention}

\newcommand{\cA}{\mathcal{A}}

\newcommand{\cC}{\mathcal{C}}

\newcommand{\cF}{\mathcal{F}}

\newcommand{\cK}{\mathcal{K}}
\newcommand{\cL}{\mathcal{L}}

\newcommand{\cS}{\mathcal{S}}

\newcommand{\cW}{\mathcal{W}}
\newcommand{\cWa}{\cW^{\operatorname{a}}}

\newcommand{\scrF}{\mathscr{F}}
\newcommand{\scrC}{\mathscr{C}}
\newcommand{\scrA}{\mathscr{A}}

\newcommand{\NN}{\mathbb{N}}

\newcommand{\ZZ}{\mathbb{Z}}

\newcommand{\Z}{\mathbb{Z}}

\newcommand{\fD}{\mathfrak{D}}

\newcommand{\fG}{\mathfrak{G}}

\newcommand{\boldg}{\mathbf{g}}
\newcommand{\bolds}{\mathbf{s}}

\DeclareMathOperator{\enc}{enc}
\DeclareMathOperator{\abelgroups}{\mathfrak{Abel}}
\DeclareMathOperator{\groups}{\mathfrak{Groups}}
\DeclareMathOperator{\altgroups}{\mathfrak{Alt}}
\newcommand{\grpclass}{\mathfrak{G}}

\newcommand{\highagr}{\mathcal{L}}
\newcommand{\outputlist}{\widetilde{\mathcal{L}}}

\DeclareMathOperator*{\EE}{\mathbb{E}}
\DeclareMathOperator*{\PP}{\Pr}
\newcommand{\mean}[1][]{\textswab{m}\!_{\,#1}\ell}

\newcommand{\LGH}{\Lambda}

\newcommand{\ind}[2]{\lvert #1 : #2 \rvert}
\DeclarePairedDelimiter{\abs}{\lvert}{\rvert}

\DeclarePairedDelimiter{\ang}{\langle}{\rangle}
\DeclarePairedDelimiter{\gen}{\langle}{\rangle}

\DeclarePairedDelimiter{\paren}{(}{)}
\newcommand{\gengroup}[1]{\langle #1 \rangle} 
\newcommand{\gencoset}[1]{\langle #1 \rangle_{\mathsf{aff}}}
\newcommand{\genaff}[1]{\gencoset{#1}}
\newcommand{\defn}[1]{\textbf{#1}}

\DeclareMathOperator{\partialto}{\rightharpoonup}
\DeclareMathOperator{\dist}{dist}

\DeclareMathOperator{\depth}{depth}
\DeclareMathOperator{\Eq}{Eq}
\DeclareMathOperator{\agr}{agr}
\DeclareMathOperator{\Hom}{Hom}
\DeclareMathOperator{\aHom}{aHom}
\DeclareMathOperator{\Var}{Var}
\DeclareMathOperator{\Cov}{Cov}
\DeclareMathOperator{\poly}{poly}
\DeclareMathOperator{\Alt}{Alt}
\DeclareMathOperator{\Sym}{Sym}

\newcommand{\defeq}{:=}
\newcommand{\from}{\colon}

\DeclareMathOperator{\sig}{enl}

\newcommand{\buck}{\cL}

\newcommand{\buckl}[1]{\buck_{#1}}

\newcommand{\abkpm}{\cL_\psi^M}
\newcommand{\sbk}{\cL}
\newcommand{\gph}{X}

\DeclareMathOperator{\dom}{dom}

\newcommand{\exthom}{\widetilde{\varphi}}

\newcommand{\HomExt}{\textsc{HomExt}}

\newcommand{\Decode}{\textsc{Decode}}

\DeclareMathOperator{\HExt}{HExt}

\newcommand{\halfl}{\frac{\frac12 \lambda + \eps}{\lambda + \eps}}
\newcommand{\epsl}{\frac{\eps}{\lambda + \eps}}
\newcommand{\abelsubs}{\mathcal{A}}

\DeclareMathOperator{\BBmult}{\mathsf{mult}}
\DeclareMathOperator{\BBinv}{\mathsf{inv}}
\DeclareMathOperator{\BBid}{\mathsf{id}}

\newcommand{\vf}{\varphi}                       
\newcommand{\ie}{i.\,e.}                        

\newcommand{\mindist}{\textsf{mindist}}

\renewcommand{\phi}{\varphi}

\newcommand{\eps}{\varepsilon}

\begin{document}
\title{\LARGE List-decoding homomorphism codes with arbitrary codomains}
\author{L\'{a}szl\'{o} Babai\thanks{\tt laci@cs.uchicago.edu}}
\author{Timothy Black\thanks{\tt timblack@math.uchicago.edu}}
\author{Angela Wuu\thanks{\tt wu@math.uchicago.edu}}
\affil{University of Chicago}
\maketitle

\vspace{.2in}\begin{abstract}
	The codewords of the \emph{homomorphism code} $\aHom(G,H)$ 
	are the affine homomorphisms between two finite groups, 
	$G$ and $H$, generalizing Hadamard codes.  Following the 
	work of Goldreich--Levin (1989), Grigorescu et al. (2006),
	Dinur et al. (2008), and 
	Guo and Sudan (2014), we further expand the range of groups
	for which local list-decoding is possible up to $\mindist$,
        the minimum distance of the code.  In particular, for the first
	time, we {\bf do not require either $G$ or $H$ to be solvable}.
	Specifically, we demonstrate a $\poly(1/\eps)$ bound on the
	list size, i.\,e., on the number of codewords within distance 
	$(\mindist-\eps)$ from any received word,
	when $G$ is either abelian or an alternating group,
	and {\bf $H$ is an arbitrary (finite or infinite) group}.
        We conjecture that a similar bound holds for all finite
        simple groups as domains; the alternating groups serve
        as the first test case.

        The abelian vs. arbitrary result then permits us to adapt 
        previous techniques to obtain efficient local list-decoding
        for this case.  We also obtain efficient local list-decoding
	for the permutation representations of alternating groups
        (i.\,e., when the codomain is a symmetric group $S_m$) under the
        restriction that the domain $G=A_n$ is paired with codomain
	$H=S_m$ satisfying $m < 2^{n-1}/\sqrt{n}$.

        The limitations on the codomain in the latter case arise
        from severe technical difficulties stemming from
        the need to solve the \emph{homomorphism extension 
        $(\HomExt)$ problem} in certain cases; these are 
        addressed in a separate paper (Wuu 2018).  

        However, we also introduce an intermediate ``semi-algorithmic''
        model we call \textbf{Certificate List-Decoding} that bypasses
        the $\HomExt$ bottleneck and works in the alternating
        vs.~arbitrary setting.

	Our new combinatorial tools allow us to play on the 
        relatively well-understood top layers of the 
        subgroup lattice of the domain, avoiding the dependence 
        on the codomain, a bottleneck in previous work.
\end{abstract}

\tableofcontents

\section{Introduction}
\subsection{Brief history}

\label{section:intro-briefhistory}
Let $G$ and $H$ be finite groups, to be referred to as
the \emph{domain} and the \emph{codomain}, respectively.
A map $\psi : G\to H$ is an
\emph{affine homomorphism} if it is a translate of a 
homomorphism, \ie, if there exists a homomorphism
$\vf : G\to H$ and an element $h\in H$ such that 
$(\forall g\in G)(\psi(g)=\vf(g)\cdot h)$.  We write
$\Hom(G,H)$ and $\aHom(G,H)$ to denote the set of 
homomorphisms and affine homomorphisms, respectively.
Let $H^G$ denote the set of all functions $f : G\to H$.

We view $\aHom(G,H)$ as a (nonlinear) code within the
code space $H^G$ (the space of possible ``received words'')
and refer to this class of codes
as \emph{homomorphism codes}.

Homomorphism codes are candidates for efficient
\emph{local} list-decoding \emph{up to minimum distance} 
($\mindist$) and in many cases it is known that their 
minimum distance is (asymptotically) equal 
to the list-decoding bound. 

This line of work goes back to the celebrated paper 
by Goldreich and Levin (1989) \cite{GL89} who found
local list-decoders for Hadamard codes, \ie, for 
homomorphism codes with domain $G=\Z_2^n$ and
codomain $H=\Z_2$.  This result was extended to 
homomorphism codes of abelian groups (both the domain
and the codomain abelian)
by Grigorescu, Kopparty, and Sudan (2006) \cite{GKS06} 
and Dinur, Grigorescu, Kopparty, and Sudan (2008)
\cite{DGKS08} and to the case of supersolvable
domain and nilpotent codomain by Guo and Sudan (2014)
\cite{GS14}, cf. \cite{BGSW}.  

While homomorphism codes have low (logarithmic) rates,
they tend to have remarkable list-decoding properties.

In particular,
in all cases studied so far (including the present paper), 
for an \emph{arbitrary} received word $f \in H^G$, and any $\eps > 0$, 
the number of codewords within radius $(\mindist - \eps)$
is bounded by $\poly(1/\eps)$ (as opposed to some faster-growing
function of $\eps$, as permitted in the theory of list-decoding).
This is an essential feature for the complexity-theoretic
application (hard-core predicates) by Goldreich and Levin.

We call the $\poly(1/\eps)$ bound \emph{economical},
and a homomorphism code permitting such a bound 
{\bf combinatorially economically list-decodable (CombEcon)}.

By \emph{efficient} decoding we mean
$\poly(\log |G|, 1/\eps)$ queries to the received word and \linebreak
$\poly(\log|G|, \log|H|, 1/\eps)$  additional work.
We call a CombEcon code {\bf AlgEcon} 
{\bf (algorithmically economically list-decodable)}
if it permits efficient decoding in this sense.  So the cited 
results show that homomorphism codes with 
abelian domain and codomain, and more generally with
supersolvable domain and nilpotent codomain, are CombEcon 
and AlgEcon.

In all work on the subject, this efficiency depends on the 
computational representation of the groups used (presentation
in terms of generators and relators, black-box access, 
permutation groups, matrix groups).
We shall make the representation required
explicit in all algorithmic results.

\subsection{Our contribution -- combinatorial bounds}
In this paper we further expand the range of groups
for which efficient local list-decoding is possible up to the 
minimum distance.  In particular, for the first
time, we {\bf do not require either $G$ or $H$ to be solvable}.
In fact, in our combinatorial and semi-algorithmic results
(see below), {\bf the codomain is an arbitrary (finite or infinite)
group}.  We say that a class $\fG$ of finite groups is {\bf universally
CombEcon} if for all $G\in\fG$ and arbitrary (finite or infinite)
$H$, the code $\aHom(G,H)$ is CombEcon.  This paper is the first 
to demonstrate the existence of significant universally CombEcon classes.

\begin{convention} \label{conv:intr-infinite}
When speaking of a homomorphism code $\aHom(G,H)$, the domain $G$ 
will always be a finite group, but the codomain $H$ will, in general,
not be restricted to be finite.
\end{convention}

\begin{theorem}[Main combinatorial result]  \label{thm:maincomb}
Finite abelian and alternating groups are universally CombEcon.
\end{theorem}

We explain this result in detail.  
By ``distance'' in a code we mean normalized Hamming distance.

\medskip\noindent
(Restatement of Theorem~\ref{thm:maincomb}.) \quad
Let the domain $G$ be a finite abelian or alternating group
and $H$ an arbitrary (finite or infinite) group.  
Let $\mindist$ denote the minimum distance of the 
homomorphism code $\aHom(G,H)$ and let $\eps > 0$.
Let $f\in H^G$ be an arbitrary received word.
Then the number of codewords within $(\mindist-\eps)$
of $f$ is at most $\poly(1/\eps)$.

\bigskip
The degree of the polynomial in the $\poly(1/\eps)$
expression for abelian domains $G$ is $C+4$ where $C$ 
is the degree in the corresponding $\{$abelian$\to$abelian$\}$
result (currently $C\approx 105$~\cite{GS14}). 
For alternating domains $G$, we prove a bound of 9 on the degree 
of the polynomial; with additional work, this can be improved 
to 7. 

Our choice of the alternating groups as the domain
is our test case of what we believe is a general
phenomenon valid for all finite simple groups.

\begin{conjecture}
The class of finite simple groups is universally CombEcon.
\end{conjecture}

The following problem is also open.

\begin{problem}
Is the class of finite groups universally CombEcon?
\end{problem}
We suspect the answer is ``no.''

Theorem~\ref{thm:maincomb} also holds for a hierarchy of
wider classes of finite groups
we call \emph{shallow random generation} groups or ``SRG groups''
(see Sec.~\ref{section:intro-SRG}).
This class includes the alternating groups.   The defining feature
of these groups is that a bounded number of random elements
generate, with extremely high probability, a ``shallow'' subgroup,
i.\,e., a subgroup at bounded distance from the top of the subgroup 
lattice. 

Our new combinatorial tools allow us to play on the 
relatively well-understood top layers of the 
subgroup lattice of the domain, avoiding the dependence 
on the codomain, a bottleneck in previous work.

\begin{remark}  \label{rmk:blowup}
Our results list-decode certain classes of codes up
to distance $(\mindist-\eps)$ for positive $\eps$.
In many cases, $\mindist$ is the list-decoding
boundary; 
examples show that the length of the list may blow up
when $\eps$ is set to zero.
Classes of such examples with abelian domain and codomain were found by
Guo and Sudan \cite{GS14}.  We add classes of examples with alternating
domains (Section~\ref{section:alt-blowup}).
\end{remark}

\subsection{Our contribution -- algorithms}

On the algorithmic front, the combinatorial bound in the
$\{$abelian$\to$arbitrary$\}$ case permits us to adapt the 
algorithm of \cite{GKS06} to obtain efficient local 
list-decoding.  We say that a class $\fG$ of finite groups 
is {\bf universally AlgEcon} if for all $G\in\fG$ and 
arbitrary finite $H$, the code $\aHom(G,H)$ is AlgEcon.  
The validity of such a statement depends not only on
the class $\fG$ but also on the representation of the
domain and codomain.

\begin{corollary}
Let $G$ be a finite abelian group and $H$ an arbitrary
finite group.  Under suitable assumptions on the representation
of $G$ and $H$, the homomorphism code $\aHom(G,H)$ is
AlgEcon.  
\end{corollary}

In other words, abelian groups are \emph{universally AlgEcon.}

In fact, the algorithm is so efficient that in the 
\emph{unit-cost black-box-access model} for $H$ 

(elements of $H$ can be named and operations
on them performed at unit cost)
the work required is only $\poly(\log|G|, 1/\eps)$.
(The cost does not depend on $|H|$; indeed, in this case, 
infinite $H$ is also allowed).  

We need to clarify the ``suitable representation.''
It suffices to assume that $G$ is a finite abelian group 
given in any presentation by generators and relators,
assuming in addition that a superset of the prime divisors 
of the order of $G$ is available.  Without the prime divisors,
we need a factoring oracle.  We need black-box access to $H$.

A \emph{permutation representation of degree $m$} of a group $G$ is a
homomorphism $G\to S_m$, where the codomain is the symmetric
group of degree $m$.  We also obtain efficient local list-decoding
for the permutation representations of alternating groups
under a rather generous restriction on the size of the
permutation domain.

\begin{theorem}[Main algorithmic result] \label{thm:mainalg}
Let $G=A_n$ be the alternating group and $H=S_m$ the symmetric
group of degree $m$.  Then $\aHom(G,H)$ is AlgEcon, assuming
$m < 2^{n-1}/\sqrt{n}$.
\end{theorem}

The limitations on the codomain arise from severe technical 
difficulties encountered.

In contrast to all previous work, in the alternating case
the minimum distance does not necessarily
correspond to a subgroup of smallest index
(modulo the ``irrelevant kernel,'' see 
Sec.~\ref{section:contr-extendingdomain}).  
This necessitates the introduction of the 
\emph{homomorphism extension $(\HomExt)$ problem}, a 
problem of interest in its own right, which remains
the principal bottleneck for algorithmic progress.
The problem was solved by Wuu~\cite{HE}
in the special case stated above.  

To bypass the $\HomExt$ bottleneck, we introduce a new
  model we call \textbf{Certificate List-Decoding}.  In this
  model the output is a short ($\poly(1/\eps)$)
  list of partial maps from $G$ to $H$
  that includes, for each affine homomorphism $\vf$ within 
  $(\mindist-\eps)$ of the received word, a 
  \emph{certificate} of $\vf$, i.\,e.,
  a partial affine homomorphism that uniquely extends to $\vf$.

  We say that a homomorphism code is
{\bf economically certificate-list-decodable (CertEcon)}
if such a list can be efficiently generated.

Note that, by definition, AlgEcon $\implies$ CertEcon $\implies$ 
CombEcon.

  We say that a class $\fG$ of finite groups is {\bf universally
  CertEcon} if for all $G\in\fG$ and arbitrary (finite or infinite)
$H$, the code $\aHom(G,H)$ is CertEcon. 

\begin{theorem}[Main semi-algorithmic result]
Alternating groups are universally CertEcon.
\end{theorem}

In fact we show that SRG groups are universally CertEcon.

Finally, we show that certificate list-decoding, combined 
with a $\HomExt$ oracle for the top layers of the subgroup 
lattice of $G$, suffices for list-decoding $\aHom(G,H)$.

This is the route we take to proving Theorem~\ref{thm:mainalg}.

We give more formal statements of these results in 
Section~\ref{sec:formal-contribution}.

\subsection{The structure of the paper}
\label{sec:structure}

Much of our conceptual framework can be interpreted
for codes in general, not just for homomorphism codes.
In Section~\ref{section_generalcodes} 
we develop the general terminology.
This includes the notions of \emph{economy}
in local list-decoding as well as 
the new concepts of \emph{certificate-list decoding}
(Sec.~\ref{section:terminology-cert}), our semi-algorithmic
intermediate concept, and \emph{mean-list decoding},
our main tool for domain relaxation 
(Sec.~\ref{sec:mean-list}), motivated by Guo and Sudan's
use of repeated codes~\cite{GS14}.  We also introduce
\emph{subword extenders}, which constitute the bridge
between certificate-list decoding and algorithmic
list-decoding (Sec.~\ref{section:terminology-subword}).

In Section~\ref{section_background} we present
notation and terminology from group theory
and computational group theory, including
our \emph{access models}, i.\,e.,
computational representations of groups 
(black-box, \linebreak generator-relator presentations, etc.).

Section~\ref{sec:formal-contribution} gives formal
statements of our results and occasional minor proofs
that contribute to the conceptual development.
The section includes
a discussion of shallow-random-generation (SRG) groups
(Section~\ref{section:intro-SRG}).  
Section~\ref{section:contr-cert-homext-to-alg}
explains the role of the Homomorphism Extension
problem in bridging the gap between 
certificate-list decoding and algorithmic list-decoding.

Section~\ref{sec:bipartite} describes a simple combinatorial
lemma (``Bipartite covering lemma'') and applies it in two
separate contexts: connecting mean-list-size to list-size,
from which we infer our domain relaxation principle, and 
the equivalence (both combinatorial and algorithmic) 
of $\Hom$ and $\aHom$.

Section~\ref{section:strategy-for-CombEcon} outlines our
basic strategy for the combinatorial bounds.  It indicates
the differences between the approach to abelian domains
and to alternating (and SRG) domains.  We indicate that
the same strategy also produces certificate-list-decoders.

Section~\ref{section:tools} describes the tools
for the combinatorial bounds.  We compare one of our
tools, a sphere packing argument via a strong
negative correlation inequality, to the Johnson bound.

Section~\ref{section:abelian} gives the full technical development
of our results for abelian domains.  

The rest of the paper, Sections~\ref{section:alternating}
to~\ref{section:srg-consequences}, provides the proofs
for alternating domains and their generalizations,
the SRG groups.

We give two proofs that alternating groups are CombEcon.

The first proof, in Section~\ref{section:alternating},
is based on a sphere packing argument and is non-constructive,
but the method applies under quite general circumstances.
The second, in Section~\ref{section:srg-consequences},
depends on structure specific to the alternating groups 
(or more generally, to SRG groups), that proof 
directly translates to a semi-algorithmic result
(CertEcon), and under restrictions of the codomain, also provides
an algorithmic result (AlgEcon).

\section{Terminology for general codes}
\label{section_generalcodes}
\subsection{List-decoding}
We introduce some terminology that applies to codes in general and 
not just homomorphism codes. 

Let $\Sigma$ be an alphabet and $\Omega$ a set we think of as the set of positions. We view $\Sigma^\Omega$, the set of $\Omega \to \Sigma$ functions, as our code space; we call its elements \emph{words.}
We write $\dist(u,w)$ for the \defn{normalized Hamming distance} between two words $u, w \in \Sigma^\Omega$ (so $0 \leq \dist(u,w) \leq 1$) and refer to it simply as \emph{distance.}  Let $\cC \subseteq \Sigma^\Omega$ be a code; we call its elements \emph{codewords.}

We write $\mindist(\cC)$ (or simply $\mindist$) for the minimum distance between distinct codewords in $\cC$. 

Words we wish to decode are referred to in the literature as \emph{received 
words.} We refer to the set of codewords within a specified distance $\rho$ of a received word $f \in \Sigma^\Omega$ as \defn{``the list''} and denote it by $\cL = \cL(\cC, f, \rho)$. We write $\ell(\cC, \rho) := \max_f \abs{\cL(\cC, f, \rho)}$.

\subsection{Combinatorial list-decoding}

The list-decoding problem
splits into a combinatorial and
an algorithmic part.

The combinatorial problem, to which we refer as
\emph{combinatorial list-decoding}, asks to
bound the size of the list.  Typically, we take $\rho = (\mindist-\eps)$ and 
we wish to obtain a bound $\ell(\cC, \rho) \leq c(\eps)$, that depends only on 
$\eps$ and the class $\scrC$ of codes under discussion ($\cC\in\scrC$).

We say that a class $\scrC$ of codes is \textbf{CombEcon} 
(``combinatorially economically list-decodable'')
if $c(\eps) = \poly(1/\eps)$ for $\cC\in\cF$. 
(With some abuse of terminology, we shall refer to a code $\cC$ 
as a CombEcon code is the class $\scrC$
of codes is understood from the context.)

\subsection{Algorithmic list-decoding}
We shall describe algorithms with certain performance guarantees
typically guaranteeing properties of the output with specified
probability.

A \defn{list-decoder} is an algorithm that, given the received word 
$f\in \Sigma^\Omega$ and the distance $\rho$, lists a superset 
$\outputlist$ of the list $\cL = \cL(\cC, f, \rho)$. 
 Typically, we take $\rho = (\mindist-\eps)$ and we wish to produce 
a list of size $\abs{\outputlist} \leq \widetilde{c}(\eps)$ for some
$\widetilde{c}(\eps)$ that depends only on $\eps$ and 
the class $\scrC$ of codes under discussion ($\cC\in\scrC$).

Adapting the terminology of \cite{GKS06} and \cite{DGKS08}, we say 
that a \defn{local algorithm} is a probabilistic algorithm that 
has oracle access to the received word $f$.  

We say that $\scrC$ is an \defn{AlgEcon} (``algorithmically economically 
list-decodable'') class of codes if there exists a local list-decoder 
with the following features. \\
\textbf{Input:} $\mindist$, $\epsilon > 0$, oracle access to $f \in \Sigma^\Omega$. \\
\textbf{	Notation:} $\cL = \cL(\cC, f, \mindist - \eps)$. \\
\textbf{Output:} A list $\outputlist$ of codewords in $\cC$ of length $\abs{\outputlist} = \poly(1/\eps)$.\\
\textbf{Guarantee:} With probability $\geq 3/4$, we have $\outputlist \supseteq \cL$.\\
\textbf{Cost:} 
	\begin{enumerate}[(i)]
	\item $\poly(\log \abs{\Omega}, 1/\eps)$ queries to the received word $f$.
	\item $\poly(\log \abs{\Omega}, \log\abs{\Sigma}, 1/\eps)$ amount of work.
	\end{enumerate}
\textbf{Access:} The meaning of this definition depends also on the access model to $\Sigma$ and $\Omega$. We shall clarify this in each application.

\paragraph{\underline{Strong AlgEcon}\vspace{.1in}\\}

In the \textbf{unit cost model} for $\Sigma$, we charge unit cost to name an element of $\Sigma$. 

We say that $\cC$ is a \textbf{strong AlgEcon} code if there exists a list-decoder  satisfying the conditions of AlgEcon, except with (ii) replaced by the following. 
\begin{enumerate}
\item[(ii')] $\poly(\log \abs{\Omega}, 1/\eps)$ amount of work in the unit cost model for $\Sigma$.
\end{enumerate}

Typically, elements of $\Sigma$ are encoded by strings of length $\log \abs{\Sigma}$ and therefore (ii') implies (ii) with linear dependence on $\log \abs{\Sigma}$. The AlgEcon results proved in prior work~\cite{DGKS08, GS14, BGSW} are actually strong AlgEcon results for those classes of pairs of groups. Our AlgEcon result for abelian domain is also strong AlgEcon (see Section~\ref{section:contr-abel}). On the other hand, our AlgEcon result for alternating domain does not meet the ``strong'' requirement.

\begin{remark}
The unit cost model can also be used in the case of infinite $\Sigma$. In fact, our AlgEcon result for abelian domains holds even for infinite codomains in the unit cost model, i.e., it satisfies (ii').
\end{remark}

\subsection{Certificate list-decoding}
\label{section:terminology-cert}
In the light of technical difficulties arising from algorithmic
list-decoding, we introduce a new type of list-decoding that is
intermediate between the combinatorial and algorithmic. We call it
``certificate list-decoding.''  We shall refer to results
of this type as ``semi-algorithmic.''

A \defn{partial map} $\gamma$ from $\Omega$ to $\Sigma$, denoted
$\gamma \colon \Omega \partialto \Sigma$, is a map of a subset of
$\Omega$ to $\Sigma$. In particular, $\dom(\gamma) \subseteq \Omega$.

\begin{definition}[Certificate]
We say that a partial map $\gamma \colon \Omega \partialto \Sigma$
is a \defn{certificate for the codeword}
$\vf\in \cC$ if $\gamma = \vf_{|\dom(\gamma)}$ and
$\vf$ is the unique codeword that extends $\gamma$. 
A \defn{certificate} for the code $\cC$ 
is a certificate for some codeword in $\cC$. 
\end{definition}

\begin{definition}[Certificate-list]
We say that a list $\Gamma$ of $\Omega \partialto \Sigma$ partial maps is a 
\defn{certificate-list} for the set $\cK \subseteq \cC$ of codewords if 
$\Gamma$ contains a certificate for each codeword in $ \cK$. 
A \defn{certificate-list} for $\cC$  \defn{up to distance} $\rho$ 
of the received word $f: \Omega \to \Sigma$ is a certificate-list 
for the list $\cL = \cL(\cC, f, \rho)$. 
\end{definition}
\begin{remark}   \label{rem:cert-list}
Note that we permit the certificate-list $\Gamma$ to contain redundancies 
(more than one certificate for the same codeword) and irrelevant items 
(partial functions that are not certificates of any codeword in $\cK$, 
or not even certificates of any codeword at all).
\end{remark}

\begin{definition}
A \defn{certificate-list-decoder} is an algorithm that, 
given the received word $f\in \Sigma^\Omega$ and the distance $\rho$, 

constructs a certificate-list of $\cC$ up to distance $\rho$ of $f$. 
\end{definition}

\begin{definition}    \label{def:certecon}
We say that $\cC$ is a  \textbf{CertEcon} (``certificate-economically list-decodable'') code if there exists a local certificate-list-decoder 
with the following features. \\
 \textbf{Input:} $\eps > 0$, oracle access to $f \in \Sigma^\Omega$. \\
\textbf{Notation:} Again, let $\cL = \cL(\cC, f, \mindist - \eps)$. \\
\textbf{Output:} A list $\Gamma$ of $\Omega \partialto \Sigma$ partial maps of length $\abs{\Gamma} = \poly(1/\eps)$.\\
\textbf{Guarantee:} With probability $\geq 3/4$, we have that $\Gamma$ is a certificate-list for $\cL$. \\
\textbf{Cost:} 
	\begin{enumerate}[(i)]
	\item $\poly(\log \abs{\Omega}, 1/\eps)$ queries to the received word $f$.
	\item $\poly(\log \abs{\Omega}, \log\abs{\Sigma}, 1/\eps)$ amount of work.
	\end{enumerate}
\textbf{Access:} The meaning of this definition depends also on the access 
model to $\Sigma$ and $\Omega$. We shall clarify this in each application. 
\end{definition}
\begin{remark}
 Note that $\mindist$ is not part of the input. In our results,
we are likely to find a certificate of $\cC$ up to distance 
$(\mindist -\eps)$ of the received word $f$, regardless of 
the actual value of $\mindist$.  We note that, depending on
the access model, we may not be able to find $\mindist$.
\end{remark}

\begin{remark}
CertEcon is intermediate between AlgEcon and CombEcon. Indeed, CertEcon 
implies CombEcon, by the length bound of the Output and the Guarantee. 
Moreover, AlgEcon implies CertEcon, as the AlgEcon Output $\outputlist$ 
satisfies the definition of a certificate, under the same Guarantee and 
Cost bound. 
\end{remark}

\paragraph{\underline{Strong CertEcon}\\}

\begin{definition}    \label{def:strong-certecon}
 We say that $\cC$ is a \textbf{strong CertEcon} code if there exists a 
certificate-list-decoder  satisfying the conditions of CertEcon, 
except with (ii) replaced by the following. 

\begin{enumerate}
\item[(ii')] $\poly(\log \abs{\Omega}, 1/\eps)$ amount of work 
in the unit cost model for $\Sigma$.
\end{enumerate}
\end{definition}

All CertEcon results in this paper are actually strong CertEcon results. 

\begin{remark}
Strong CertEcon does not follow from AlgEcon, though it does follow from strong AlgEcon. 
\end{remark}

\begin{remark}
As in the AlgEcon context, the unit cost model can also be used in the case of infinite $\Sigma$. In fact, all our CertEcon results hold  for infinite codomain in the unit cost model, i.e., they satisfy (ii''). 
\end{remark}

\subsection{Mean-list-decoding}   \label{sec:mean-list}

Let $\scrF = \{ f_i : i \in I \}$ be a family of received words $f_i \in \Sigma^\Omega$. By the size $\abs{\scrF}$ we mean the size $\abs{I}$ of the index set $I$. The \defn{average distance} $\dist(w, \scrF)$ of a word $w \in \Sigma^\Omega$ to $\scrF$ is the average distance of $w$ to elements of $\scrF$, given by $\dist(w,\scrF) = \EE_{i \in I}[ \dist(w, f_i)]$. (The expectation $\EE$ is taken with respect to the uniform distribution over $I$.)

\begin{definition}[Mean-lists]  \label{def:mean-list-general}
We define the \defn{mean list} $\cL$ as the set of codewords within a specified average distance $\rho$ of the received words $\scrF$, i.e., 
	\begin{equation}
	\cL = \cL(\cC, \scrF, \rho) := \{ w \in \cC : \dist(w, \scrF) \leq \rho \}.
	\end{equation}
We write $\mean (\cC, \rho) := \max_{\scrF} \cL(\cC, \scrF, \rho)$ for the maximum mean-list size for a given distance $\rho$.
\end{definition}
This concept was inspired by the use of repeated codes by 
Guo and Sudan~\cite{GS14}, see Remark~\ref{rmk:mean-identification}.

As we shall see, the mean list-decoding concept helps expand the scope of our results, without making them more difficult to prove.  We adapt above terminology to the context of mean-list-decoding. 

\paragraph{Combinatorial.} We wish to bound mean-list size by 
$\abs{\cL(\cC, \scrF, \rho)} \leq c'(\eps)$ for some $c'(\eps)$
that depends only on $\eps$ and 
the class $\scrC$ of codes under discussion ($\cC\in\scrC$).
We say that the class $\scrC$ of codes is \textbf{CombEconM} 
(``combinatorially economically mean-list-decodable'') 
if $c'(\eps) = \poly(1/\eps) $ for $\cC\in\scrC$.

\paragraph{Algorithmic.} We say that the class $\scrC$ of codes is 
\defn{AlgEconM} (``algorithmically economically mean-list-decodable'')
if it satisfies the definition of AlgEcon classes of codes, with the 
following modifications.

For each $\cC\in\scrC$, the received word $f$ is replaced by a family 
$\scrF$ of received words and the list $\cL$ becomes 
$\cL = \cL( \cC, \scrF, \rho )$. Oracle access to $\scrF$ means that, 
given $i\in I$ and $\omega \in \Omega$, the oracle returns $f_i(\omega)$. 
Condition (ii) is replaced by the following.
\begin{itemize}[(ii-M)]
\item $\poly(\log \abs{\Omega}, \log \abs{\Sigma}, \log \abs{\scrF}, 1/\eps)$ amount of work.
\end{itemize} 

 Note that the number of queries to the family $\scrF$ remains $\poly(\log \abs{\Omega}, 1/\eps)$.

\paragraph{Certificate.}

 We say that a class $\scrC$ of codes is \defn{CertEconM} 
(``certificate economically mean-list-decodable'') if it 
satisfies the definition of CertEcon codes, with the same 
modifications as AlgEconM. 

\begin{theorem}
\label{thm:main-meanLD-iff-LD}
For a class $\scrC$ of codes, we have the following. 
\begin{enumerate}[(i)]
\item $\scrC$ is CombEconM if and only if it is CombEcon.
\item  $\scrC$ is AlgEconM if and only if it is AlgEcon.
\item  $\scrC$ is CertEconM if and only if it is CertEcon. \end{enumerate}
\end{theorem}
For more detailed statements and proofs see Section~\ref{section:mean}.

\begin{remark}[Significance of mean-list-decoding]
\label{rmk:intro-mean-significance}
Dinur et al. show the CombEcon and AlgEcon list-decodability of 
$\{$abelian$\to$abelian$\}$ homomorphism codes~\cite{DGKS08}.
We shall see that Theorem~\ref{thm:main-meanLD-iff-LD} quickly leads to the conclusion of CombEcon list-decodability of $\{$arbitrary$\to$abelian$\}$ homomorphism codes. The same inference can be made about AlgEcon list-decodability, assuming natural conditions about representation of the domain group. See Section~\ref{section:mean} for details.
\end{remark}

\paragraph{\underline{Strong mean-list-decoding}\vspace{.1in}\\}
We say that $\cC$ is a \defn{strong AlgEconM} code if it satisfies the definition of AlgEconM, except with (ii-M) replaced by (ii'-M) below. Similarly, we say that $\cC$ is a \defn{strong CertEconM} code if it satisfies the definition of CertEconM, except with (ii-M) replaced by (ii'-M) below. 
\begin{enumerate}
\item[(ii'-M)] $\poly(\log \abs{\Omega}, 1/\eps)$ amount of work in the unit cost model for $\Sigma$ and unit sampling cost model for $\scrF$.
\end{enumerate}
In the \defn{unit sampling cost model} for $\scrF = \{ f_i : i \in I \}$, we charge unit cost for naming any $i \in I$ and for generating a  uniform random $i \in I$.

\subsection{Subword extension}
\label{section:terminology-subword}

In this section we introduce terminology to
formalize our strategy to advance from
certificate-list-decoding to algorithmic list-decoding
(Observation~\ref{obs:cert-to-alg} below).

\begin{definition}[Subword extension problem]
Let $\cC$ be a code. The \defn{subword extension problem} asks, 
given a partial map $\gamma : \Omega \partialto \Sigma$, whether 
$\gamma$ extends to a codeword in $\cC$.
\end{definition}

A \defn{subword extender} is an algorithm that answers this question 
and returns a codeword in $\cC$ extending $\gamma$, if one exists.

\begin{observation}
A certificate-list-decoder and a subword extender combine to a list-decoder. 
\end{observation}

\begin{remark}  \label{rmk:terminology-subword}
This observation describes our two-phase plan to prove algorithmic list-decodability results for homomorphism codes with alternating domains. In the case of homomorphism codes, the subword extension problem corresponds to the \emph{homomorphism extension problem} (see Section~\ref{section:intro-contribution-altalg}). The algorithmic difficulty of the homomorphism extension problem is a major bottleneck to further progress. 
\end{remark}

In fact, the plan suggested by this observation is too ambitious. We have 
no hope to solve the subword extension problem in cases of interest for 
\emph{all} subwords.

Therefore, we relax the subword extender concept; 
correspondingly, we strengthen the notion of certificates required. 

Let $\cW$ be a set of $\Omega \partialto \Sigma$ partial maps. 

\begin{definition}[$\cW$-subword extender]
\label{def:intro-subword-extender}
The \defn{$\cW$-subword extension problem} asks to solve the subword extension problem on inputs from $\cW$. A \defn{$\cW$-subword extender} is 
an algorithm $\scrA$ that takes as input any partial map 
$\gamma : \Omega \partialto \Sigma$ and returns a yes/no answer; and
in the case of a ``yes'' answer, it also returns a codeword
$\scrA(\gamma)\in \cC$, such that
\begin{itemize}
\item if $\gamma\in\cW$ then the answer is ``yes'' if and only 
if $\gamma$ extends to a codeword, and in this case,
$\scrA(\gamma)$ is a codeword that extends $\gamma$.
\end{itemize}
\end{definition}
\begin{remark}
Note that $\scrA$ is not required to decide whether $\gamma\in\cW$.
$\cA$ must correctly decide extendability of $\gamma$ for all
$\gamma\in\cW$; in case $\gamma\notin\cW$, the algorithm may
return an arbitrary answer.
\end{remark}

\begin{definition}[$\cW$-certificate]
A \defn{$\cW$-certificate} is a certificate that belongs to $\cW$.
\end{definition}

\begin{definition}[$\cW$-certificate-list]
We say that a list $\Gamma$ of $\Omega \partialto \Sigma$ partial maps is a 
\defn{$\cW$--certificate-list} for the set $\cK \subseteq \cC$ of 
codewords if $\Gamma$ contains a $\cW$-certificate for each codeword 
in $ \cK$. 
A \defn{$\cW$-certificate-list} for $\cC$  \defn{up to distance} $\rho$ 
of the received word $f: \Omega \to \Sigma$ is a $\cW$-certificate-list 
for the list $\cL = \cL(\cC, f, \rho)$. 
\end{definition}

\begin{remark}
Note that, as mentioned in Remark~\ref{rem:cert-list},
we permit the $\cW$-certificate-list $\Gamma$ to contain redundancies 
and irrelevant items, including partial functions $\gamma$ that
do not belong to $\cW$.
\end{remark}

\begin{definition}
A $\cW$-certificate-list-decoder is an algorithm that, 
given the received word $f\in \Sigma^\Omega$ and the distance $\rho$, 
constructs a $\cW$-certificate-list of $\cC$ up to distance $\rho$ of $f$. 
\end{definition}

Our overall strategy for the case when $G$ is ``far from abelian''
is summarized in the following observation.

\begin{observation}
\label{obs:cert-to-alg}
For any set $\cW$ of $\Omega \partialto \Sigma$ partial maps, 
a $\cW$-certificate-list-decoder and a $\cW$-subword extender 
combine to a list-decoder. 
\end{observation}

\begin{definition}    \label{def:wcertecon}
We say that $\cC$ is a \textbf{$\cW$-CertEcon} 
(``$\cW$-certificate-economically list-decodable'')  code
if there exists a local $\cW$-certificate-list-decoder with
the features listed in Definition~\ref{def:certecon}.
\end{definition}

\begin{definition}    \label{def:strong-wcertecon}
We say that $\cC$ is a \textbf{strong $\cW$-CertEcon} code if there exists a 
strong $\cW$-certificate-list-decoder, \ie,
a $\cW$-certificate-list-decoder that is 
a strong certificate-list-decoder 
(see Definition~\ref{def:strong-certecon}).
\end{definition}

\subsection{Minimum distance versus maximum agreement}

Recall that our code space is $\Sigma^\Omega$, 
the set of $\Omega \to \Sigma$ functions.
In the theory of error-correcting codes, the usual measure of distance
between two functions (strings) is the (normalized) Hamming distance, 
the fraction of symbols on which they differ.
Following~\cite{GKS06}, we find it convenient 
to consider the measure complementary to normalized Hamming
distance, the (normalized) \defn{agreement},
\begin{equation}  \label{eq:agr-generalcodes}
 \agr(f,g) \defeq \frac{1}{
  \abs{\Omega}}\abs{\{\omega \in \Omega \mid f(\omega) = g(\omega)\}},
\end{equation}
the fraction of positions
on which the two functions $f, g: \Omega \to \Sigma$ agree.
\begin{definition}  \label{def:maxagreement-generalcodes}
The \defn{maximum agreement} of the code $\cC$ is given by
	\begin{equation*}
	\Lambda_{\cC} \defeq \max_{\substack{\phi,\psi \in \cC \\ 
          \phi \ne \psi}} \agr(\phi,\psi). 
	\end{equation*}
\end{definition}

\begin{fact}
The minimum distance is the complement of the maximum agreement, i.\,e., 
	\begin{equation*}
	\mindist = 1 - \Lambda_{\cC}.
	\end{equation*}
\end{fact} 
So, the codewords within distance $(\mindist - \eps)$ of a received
word $f$ are the same as the codewords with agreement at least
$\Lambda_{\cC} + \eps$ with $f$.

Classes of examples
for the infeasibility of list-decoding outside this range were provided by Guo and Sudan \cite{GS14} for abelian domain and codomain, and we provide such classes for alternating domain (see Section~\ref{section:alt-blowup}), 
so the list-decoding radius is $\mindist$ for those classes.

\section{Preliminaries}
\label{section_background}

Let $G$ be a set. For any subset $S \subseteq G$, define 
the \defn{density} of $S$ in $G$ by $\mu_G(S) = \frac{\abs{S}}{\abs{G}}.$
We call $G$ the ``ambient set'' and write $\mu(S) = \mu_G(S)$ when $G$ is 
understood. The ambient set will generally be a group $G$.

\subsection{Groups}

In this paper we will denote the class of all groups (finite or infinite) 
by $\groups$.  We write $\abelgroups$ to denote the class of finite
 abelian groups and $\altgroups$ for the class of (finite) alternating 
groups.

Our group theory reference is \cite{Rob82}. We review some definitions 
and facts. 

Let $G$ be a group. We write $H \leq G$ to express that $H$ is a
subgroup; we write $H \trianglelefteq G$ if $H$ is a normal
subgroup. We refer to cosets of subgroups of $G$ as
\defn{subcosets}. For the subcoset $aH$ of $G$ (where $H \leq G$), let
$\ind{G}{aH} := \ind{G}{H}$ denote the index of $H$ in $G$. For a
subset $S$ of a group $G$, the \defn{subgroup $\langle S \rangle$
  generated by $S$} is the smallest subgroup of $G$ containing $S$. If
$\langle S \rangle = G$, then $S$ \defn{generates} $G$. 
A subset $K\subseteq G$ is \defn{affine-closed} if 
$(\forall a,b,c\in K)(ab^{-1}c\in K)$.  An affine-closed subset
is either empty or it is a subcoset.  The intersection of 
affine-closed subsets is affine-closed.  The
\defn{affine closure $\gencoset{S}$}, affinely generated by
$S$, is the the smallest affine-closed subset containing $S$. 
Note that the affine closure of the empty set is empty.
The affine closure of a nonempty set is a subcoset; indeed, for
any $q \in S$, we have that $\gencoset{S}= q\cdot\gengroup{q^{-1} r
  \mid r \in S }$.

\subsection{Homomorphism codes}

\subsubsection{Affine homomorphisms as codewords}

Let $G$ be a finite group and $H$ a group. Denote the set of
homomorphisms from $G$ to $H$ by $\Hom(G,H)$.  

\begin{definition}
Let $G_1$ and $H_1$ be affine-closed subsets of $G$ and $H$, resp.
A function $\vf\colon G_1 \to H_1$ is an \defn{affine homomorphism} if
$$(\forall a, b, c \in G_1)(\vf(a)\vf(b)^{-1}\vf(c) = \vf(ab^{-1}c))\,.$$
\end{definition}
We write $\aHom(G_1,H_1)$ to denote the set of affine homomorphisms from 
$G_1$ to $H_1$.
\begin{fact}
Let $G_1\le G$ and $H_1\le H$.  Let $a\in G$ and $b\in H$.
A function $\vf\colon aG_1 \to bH_1$ is an affine homomorphism
if and only if there exists $h \in H$ and $\vf_0 \in \Hom(G_1,H_1)$ 
such that 
\begin{equation}  \label{eq:affine}
\vf(g) = h \cdot \vf_0(g)
\end{equation}
for every $g \in G_1$. 
The element $h$ and the homomorphism $\vf_0$ are unique.
\end{fact}
The analogous statement also holds with $h$ on the right of $\vf_0$.

\begin{definition}
For sets $G, H$ and functions $f,g \colon  G \to H$,
the \defn{equalizer} $\Eq(f,g)$
is the subset of $G$ on which $f$ and $g$ agree, i.\,e.,
	\[
	\Eq(f,g) := \{x \in G \mid f(x) = g(x)\}.
	\]
More generally, if $\Phi$ is a collection of functions from $G$ to $H$,
then the \defn{equalizer} $\Eq(\Phi)$ is the set
	\[
	\Eq(\Phi) := \{x \in G \mid (\forall f,g \in \Phi)(f(x) = g(x))\}.
	\]
\end{definition}
\begin{fact}  	\label{fact:equalizer-coset}
\begin{itemize}
 \item[(a)] If $\vf,\psi\in\Hom(G,H)$ then $\Eq(\vf,\psi)\le G$.
 \item[(b)] If $\vf,\psi\in\aHom(G,H)$ then $\Eq(\vf,\psi)$ is
affine-closed.  Moreover, if $\vf_0,\psi_0\in\Hom(G,H)$
are the corresponding homomorphisms (see~\eqref{eq:affine}) then
either $\Eq(\vf,\psi)$ is empty or $\Eq(\vf,\psi)=g\cdot\Eq(\vf_0,\psi_0)$
for any $g\in\Eq(\vf,\psi)$. 

\end{itemize}
\end{fact}

Recall that the (normalized) \defn{agreement} $\agr(f,g)$ between 
two functions $f,g\colon G \to H$ is given by
\[
\agr(f,g) := \frac{\abs{\Eq(f,g)}}{\abs{G}}.
\]

Specializing Def.~\ref{def:maxagreement-generalcodes} to homomorphism
codes, we write 
\[
\Lambda_{G,H} := \Lambda_{\aHom(G,H)}
\]
for the \defn{maximum agreement} of $\aHom(G,H)$.  In other words,
	\[
	\Lambda_{G,H} :=
	\max_{\substack{\vf,\psi \in \aHom(G,H) \\ \vf \ne \psi}}
	\agr(\vf,\psi)
	\]

If the groups $G$ and $H$ are understood, we often write $\Lambda$ 
in place of $\Lambda_{G,H}$.  
Using this terminology, the min distance of the homomorphism code
$\aHom(G,H)$ is $(1- \Lambda_{G,H})$.

The following statement appears in \cite[Prop. 3.5]{Guo15}.
We include the proof for completeness.
\begin{proposition}[Guo]
\label{prop:prelim-lambda-Hom-vs-aHom}
Let $G,H$ be groups.
The maximum agreement $\Lambda_{G, H}$ can equivalently be defined with 
$\aHom$ replaced by $\Hom$, i.\,e.,  
\[ \Lambda_{\Hom(G,H)} = \Lambda_{\aHom(G,H)}.
\]
\end{proposition}
Here we use the convention that the maximum of the empty set 
(of nonnegative numbers) is zero.  Otherwise we would need
to make the additional assumption $\abs{\Hom(G, H)} > 1$. 
\begin{proof}
Let $\Lambda'_{G, H}=\Lambda_{\Hom(G,H)}$.

Obviously $\Lambda_{G, H}\ge \Lambda'_{G, H}$.
Now let $\vf,\psi\in \aHom(G,H)$.  So, by Eq.~\eqref{eq:affine},
there exist $h_1,h_2\in H$ and $\vf_0,\psi_0\in\Hom(G,H)$ be such 
that for all $g\in G$ we have $\vf(g)=h_1\vf_0(g)$ and 
$\psi(g)=h_2\psi(g)$.  It follows that if 
$g\in\Eq(\vf,\psi)$ then $\Eq(\vf,\psi)=g\Eq(\vf_0,\psi_0)$.
Hence $\agr(\vf,\psi)$ is either zero or equal to
$\agr(\vf_0,\psi_0)$, proving that $\Lambda_{G, H}\le \Lambda'_{G, H}$.
\end{proof}

\begin{corollary}
	\label{cor:lambda-subgroup}
Let $G$ be a finite group and $H$ a group. Then, 
$\Lambda \leq \max\{\mu(K) \mid K \lneq G\}$, 
the largest density of a proper subgroup of $G$. 
\end{corollary}

\begin{fact}
Let $G$ and $H$ be groups and $S \subseteq G$ a subset. 
If $\vf, \psi \in \aHom(G, H)$ and $\vf(x) = \psi(x)$ 
for all $x \in S$, then $\gencoset{S} \subseteq \Eq(\vf, \psi)$. \qed
\end{fact}

\begin{corollary}  \label{cor:agree-on-coset-generators}
Let $G$ be a finite group, $H$ a group, and $S \subseteq G$, 
such that $\mu(\gencoset{S}) > \Lambda_{G, H}$. 
If $\vf, \psi \in \aHom(G, H)$ are such that 
$\vf(x) = \psi(x)$ for all $x \in S$, then $\vf = \psi$.
\end{corollary}
\begin{remark}[Why affine?]
The reader may ask, why we (and all prior work) consider affine 
homomorphisms rather than homomorphisms.  The reason is that
affine homomorphisms are simply the more natural objects in this 
context.  To begin with, this object is more homogeneous.
For instance, for finite $H$, under random affine homomorphisms, 
the images of any element $g\in G$ are uniformly distributed over $H$. 

This uniformity also
serves as an inductive tool: when extending the domain from
a subgroup $G_0$ to a group $G$, the action of any homomorphism
$\vf\in\Hom(G,H)$ can be split into actions on the cosets
of $G_0$ in $G$.  Those actions are affine homomorphisms.
On the other hand we also note that list-decoding of 
$\Hom(G,H)$ and $\aHom(G,H)$ are essentially equivalent
tasks; see Section~\ref{sec:Hom-versus-aHom}.
\end{remark}

\subsection{Computational representations of groups and homomorphisms}
\label{section_algorithmicgrouptheory}

In this section we discuss the models of access to groups 
required by our algorithms.  The choice of the model
significantly impacts the running time and even the 
feasibility of an algorithm. 
 
The models include oracle models (black-box access, black-box groups),
generator-relator presentations, and various explicit models
such as permutation groups, matrix groups, direct products of
cyclic groups of known orders.

Recall that our domain groups are always finite but the codomain
may be infinite (Convention~\ref{conv:intr-infinite}).  

Recall also that {\bf homomorphisms} will be represented by
the list of their values on a set of generators.

\subsubsection{Black-box models}
\label{section:prelim-blackbox}
If the codomain is infinite, and even if it is finite but very large, 
the black-box-group model with its fixed-length encoding~\cite{BS84} 
is not appropriate (see ``encoded groups'' below).  
We start with an extension of that model.

\begin{definition}[Black-box access]
	\label{def:black-box access}
	An \defn{unencoded black-box representation} of a 
 (finite or infinite) group $K$ is an ordered $5$-tuple
	\begin{align*}
	(U, r, \BBmult, \BBinv, \BBid)
	\end{align*}
	where 
	\begin{itemize}
		\item $U$ is a (possibly infinite) set;
		\item $r\colon  U \to K \cup \{*\}$ with $r(U) \supseteq K$;
		\item $\BBmult\colon  r^{-1}(K) \times r^{-1}(K) \to r^{-1}(K)$ with $r(\BBmult(x, y)) = r(x)r(y)$ for all $x, y \in r^{-1}(K)$;
		\item $\BBinv\colon  r^{-1}(K) \to r^{-1}(K)$ with $r(\BBinv(x)) = r(x)^{-1}$ for $x \in r^{-1}(K)$; and
		\item $\BBid\colon  r^{-1}(K) \to \{\mathsf{yes}, \mathsf{no}\}$ with $\BBid(x) = \mathsf{yes}$ if and only if $r(x)$ is the identity in $K$.
	\end{itemize}
 We say that an algorithm has {\bf black-box access} to the group $K$ if 
the algorithm can store elements of $U$ and query the functions (oracles)
$\BBmult, \BBinv, \BBid$.  We say that $K$ is given as an
{\bf (unencoded) black-box group} if in addition a list of generators
of $K$ is given.
\end{definition}
\begin{remark}
We emphasize that the difference between \emph{black-box access} 
to a group
$G$ and the group $G$ being given as a \emph{black-box group} is that
in the latter model, a list of generators of $G$ is given, whereas
no elements of $G$ may be a priori known in the former.
\end{remark}

If $U=\{0,1\}^n$ then we talk about an {\bf encoded group}, of
{\bf encoding length} $n$.  This of course implies that $K$ is
finite, namely, $|K|\le 2^n$.  (This is the model introduced
in \cite{blackbox}.)

\medskip
In an abuse of notation, when black-box access to a group $K$ is given, 
we may
refer to elements of $r^{-1}(K)$ by their images under $r$, we may
write $gh$ in place of $\BBmult(g, h)$, we may write $g^{-1}$ in place
of $\BBinv(g)$, and we may write $g = 1$ in place of $\BBid(g) =
\mathsf{yes}$.

\medskip\noindent
{\bf Access to domain and codomain.}
In general we shall not need generators of the codomain, $H$,
just black-box access.  On the other hand, we do need generators
of the domain, $G$; homomorphisms will be defined by their
values on a set of generators.  So our access to the 
domain will be assumed to be at least as strong as 
an (encoded) black-box group.

\medskip\noindent
{\bf The black-box unit cost model.}
The (unencoded) black-box access model is particularly well suited
to the {\bf unit-cost model} where we assume that we can copy and store an 
element of $U$ and query an oracle at unit cost.  We shall analyze our
algorithms in the unit-cost model for the codomain $H$.  This essentially
counts the operations performed in $H$, so its bit-cost will incur an
additional factor of $O(\log |H|)$ (if $H$ is finite and nearly optimally 
encoded).

\medskip\noindent
{\bf Random generation.}
In encoded black-box groups, independent nearly uniform random elements
can be generated in time, polynomial in the 
encoding length~\cite{Bab91BBpolygen}. 

\begin{remark}
Black-box groups have been
studied in a substantial body of literature, both
in the theory of computing and in computational group theory (see the
references in~\cite{BabaiBealsSeress}).  It is common to make
additional access assumptions to a black-box group (assume additional
oracles) such as an oracle for the order of the elements.

Given a black-box group $H$, we cannot determine the order
$\abs{H}$ or the order of a given element $h \in H$. In fact, 
even with an oracle for the order of elements, $\Z_p$
and $\Z_p \times \Z_p$ cannot be distinguished in fewer than 

$\Omega(\sqrt{p})$ randomized black-box identity queries.  

To avoid
such obstacles, it is common to assume additional information beyond
black-box access. In finding $\Lambda_{G,H}$ for abelian domain $G$

one needs to decide if a given prime divides $\abs{H}$. To accomplish 
this, we assume additional information about the group $H$

such as 
the order $\abs{H}$ or 
the list of primes dividing $\abs{H}$.

\end{remark}

\subsubsection{Generator-relator presentation, homomorphism
checking}  
By ``presentation'' of a group we mean \emph{generator-relator
presentation}.   

For a group given by a presentation,
basic questions, such as whether the group has order 1, are undecidable.
However, special types of presentations, such polycyclic presentations
of finite solvable groups, are often helpful.  Note, however,
that it is not known, how to efficiently perform group operations
in a finite solvable group given by a polycyclic presentation,
so such presentations cannot answer basic black-box queries.

Any presentation, however, can be used for homomorphism
checking, a critical operation in decoding homomorphism
codes.

\begin{proposition}[Homomorphism checking]  \label{prop:prelim-homcheck}
Let $S\subseteq G$ be a list of generators of $G$.
Assume a presentation of $G$ is given in terms of $S$.   
Let $\vf : S\to H$ be a function.  
Then $\vf$ extends to a homomorphism $\exthom : G\to H$
if and only if the list $(\vf(s) \mid s\in S)$ satisfies
the relations.
\end{proposition}

Note that this gives an efficient way to check whether
$\vf$ extends to a homomorphism if the relators are 
short or are given as short \emph{straight-line programs},
assuming black-box access to the codomain.

\begin{definition}
Let $G$ be a group and $S=\{s_1,\dots,s_k\}$ a list of elements
of $G$.  A \emph{straight-line program} in $G$ from $S$ to
$g\in G$ is a sequence $P=(x_1,\dots,x_m)$ of elements of $G$ such 
that each $x_i$ is either a member of $S$ or a product of the form
$x_jx_k$ for some $j,k < i$ or $x_i^{-1}$ for some $i < k$.
We say that the element $x_m$ is given in terms of $S$ by the 
straight-line program $P$.
\end{definition}

The following is well known.
\begin{proposition}  \label{pres-of-perm-groups}
Let $G\le S_n$ be a permutation group and $S$ a
set of generators of $G$.  Given $S$, a presentation of $G$
in terms of $S$ can be computed in polynomial time,
where the relators returned are represented as straight-line
programs.
\end{proposition}

\subsubsection{Abelian groups}
\label{section:prelim-smith}
The \emph{invariant factor decomposition} of a finite abelian 
group $G$ is a decomposition as a direct product of cyclic
groups, $G \cong \ZZ_{n_1}\times\cdots\times\ZZ_{n_k}$,
where 
for each $i$, the integer $n_i$ divides $n_{i+1}$.
Such a decomposition
can be further split into a direct product of cyclic groups
of prime power order; the result is a \emph{primary decomposition}.

Any abelian presentation of a finitely generated abelian group
can be converted into an invariant factor decomposition in 
polynomial time, using the Smith normal form.  However,
moving to a primary decomposition requires factoring
the order of the group; to this end, knowing a superset
of the prime divisors of the order suffices.  All prior
algorithmic results as well as those of the present paper on
homomorphism codes with abelian domain require the primary 
decomposition.

\section{Formal statements} \label{sec:formal-contribution}

\subsection{List-decoding homomorphism codes}

Let $\fD$ be a class of pairs $(G,H)$ of groups. We say that $\fD$ is CombEcon
if the class $\{ \aHom(G,H) \mid (G,H) \in \fD \}$ of codes is CombEcon. 
We define CertEcon and AlgEcon classes of pairs of groups
analogously. 

Denote by $\groups$ the class of all groups, finite or infinite. 
Recall that we say that a class $\fG$ of finite groups is 
\defn{universally CombEcon} if $\fG \times \groups$ is CombEcon.  
We define universally CertEcon and universally AlgEcon analogously, 
under access models to be specified.
 
A common feature of the prior work reviewed in 
Section~\ref{section:intro-briefhistory} is that
each class of pairs of groups considered was
CombEcon and AlgEcon.

The present work continues to maintain this feature.

All previously existing results put structural restrictions both on the 
domain and the codomain. 
In particular, they were restricted to subclasses of the solvable groups. In this paper we extend the economical list-decodability (both combinatorial and algorithmic) in the following three directions. 
\begin{enumerate}
\item We give a general principle for removing certain types of constraints on the domain 
(see Section~\ref{section:contr-extendingdomain}). It will follow that the previously known results extend to arbitrary domains. 
\item We find universally economically list-decodable classes of groups 

Specifically, abelian and alternating groups are universally CombEcon. 
Moreover, abelian groups are universally AlgEcon, and alternating groups 
are universally CertEcon, under modest access assumptions. 
\item We exhibit the first (nontrivial) classes of 
examples where the domain is not solvable. 
\end{enumerate}

We note that no CombEcon bounds appear to be
known for the much-studied classical linear codes
(Reed--Solomon, Reed--Muller, BCH) (cf., e.g., \cite{BL15}).

The $\poly(1/\eps)$ CombEcon bound for Hadamard codes is 
quadratic~\cite{GL89}. For abelian and nilpotent 
groups, it currently has degree 105~\cite{DGKS08, GS14}.

\subsection{Extending the domain: the irrelevant kernel}
\label{section:contr-extendingdomain}
In the prior work reviewed, both the domain 
and the codomain was abelian or close to abelian 
(nilpotent or supersolvable).
It is natural to ask how to further relax the
structural constraints on the groups involved.

We point out that structural constraints such as
nilpotence or solvability (or any other hereditary
property) play a very different role if imposed
on the domain as on the

codomain.  For instance, a combinatorial list-decoding bound on
$\{$abelian~$\to$~abelian$\}$ homomorphism codes
implies the same bound for $\{$arbitrary~$\to$~abelian$\}$
homomorphism codes. This is shown by reducing 
list-decoding $\aHom(G,H)$ for arbitrary $G$ and abelian $H$
to mean-list-decoding

$\aHom(G/G',H)$, where $G'$ is the commutator 
subgroup of $G$, so $G/G'$ is the largest 
abelian quotient of $G$.  A similar argument 
extends the bounds for $\{$nilpotent~$\to$~nilpotent$\}$ 
homomorphism codes to $\{$arbitrary~$\to$~nilpotent$\}$,
working through the largest nilpotent quotient
of $G$. 

Similar results hold for certificate and algorithmic list-decoding. 

In general, we can replace $G$ by its \emph{relevant quotient} $G/N$, 
where $N$ is the irrelevant kernel (intersection of the kernels
of all $G\to H$ homomorphisms),

see Sec.~\ref{section:mean-irrelevant}.

While this observation extends the reach of the
results of Dinur et al.~\cite{DGKS08} and
Guo and Sudan~\cite{GS14}, it also shows that,
in a sense, the gains by extending the 
class of groups serving as the domains, without
relaxing the structural constraints on the
codomains, are \emph{virtual}, and the
main impediment to extending these results
to wider classes of pairs of groups is 
the structural constraints on the \emph{codomain}.

Our main contribution is the {\bf elimination of
	all constraints on the codomain.}

This also opens up the question of meaningfully
(as opposed to ``virtually'') {\bf removing 
	structural constraints on the domain side.}
Of particular interest becomes the case
where the domain is a \emph{finite simple group} and 
the codomain is arbitrary.  We initiate this direction
by studying the class of {\bf alternating groups
	as domains.}

\begin{definition}[Irrelevant kernel]
\label{def:irrelevant-kernel}
Let $G$ and $H$ be groups. The \defn{$(G,H)$-irrelevant kernel} (or ``irrelevant kernel'' if $G$ and $H$ are clear) is the intersection of the kernels of all $G \to H$ homomorphisms, i.e.,
	\begin{equation}
	\bigcap_{\varphi \in \Hom(G,H)} \ker(\varphi).
	\end{equation}
We call elements and subgroups of the irrelevant kernel \defn{irrelevant}.
\end{definition}

For instance, if $H$ is abelian, then the commutator subgroup $G'$ is irrelevant. 

\begin{theorem}
\label{thm:main-Econ-quotient-implies-Econ-fullgroup}
Let $N$ be an irrelevant normal subgroup of $G$. Then, $\Lambda_{G/N, H} = \Lambda_{G,H}$. Moreover, 
\begin{enumerate}[(i)]
\item if $\aHom(G/N,H)$ is CombEcon then $\aHom(G,H)$ is CombEcon;  
\item if $\aHom(G/N,H)$ is CertEcon then $\aHom(G,H)$ is CertEcon; 
\item if $\aHom(G/N,H)$ is AlgEcon then $\aHom(G,H)$ is AlgEcon. 
\end{enumerate}
For items (ii) and (iii) we need to make suitable assumptions on access 
to the domain. 
\end{theorem}
For the proofs and discussion, see Section~\ref{section:mean}. The proofs 
rely on mean-list-decoding (Theorem~\ref{thm:main-meanLD-iff-LD}). 

\begin{corollary}
The code $\aHom(G,H)$ is AlgEcon for any finite group $G$ and any 
finite nilpotent group $H$. 
\end{corollary}
\begin{proof}
Combine Theorem~\ref{thm:main-Econ-quotient-implies-Econ-fullgroup} and 
the main result of~\cite{GS14}. For abelian $H$, use~\cite{DGKS08} instead.
\end{proof}

\subsection{List-decoding: abelian $\to$ arbitrary}
\label{section:contr-abel}

We state our main result about abelian domains.

\begin{theorem} 
\label{thm:intro-main-abel}
If $G$ is a finite abelian group, then $G$ is universally CombEcon and strong AlgEcon list-decodable. 

\end{theorem}

The degree of the $\poly(1/\eps)$ list-size bound is $C+4$ where $C$ 
is the bound
for $\{$abelian~$\to$~abelian$\}$ homomorphism codes
(currently $C\approx 105$~\cite{GS14}). 

The proof of the CombEcon bound is based on the following structural result that says the range of all relevant homomorphisms is covered by a small number of finite abelian subgroups of $H$.

\begin{theorem}[Structure of range]
	\label{thm:intro-abel-split}
	Let $G$ be a finite abelian group, $H$ an arbitrary group (finite or infinite), $f \in H^G$ a function, and $\eps > 0$. Then there exists a set $\abelsubs$ of finite abelian subgroups of $H$ with $\displaystyle{\abs{\abelsubs} < \frac{1}{4 (\Lambda + \eps) \eps^2} + \frac{1}{\eps}}$ such that for all $\phi \in \highagr(\Hom(G, H), f, \Lambda + \eps)$, there is $M \in \abelsubs$ such that $\phi(G) \leq M$.
\end{theorem}

\paragraph{Access model.} We need to clarify how the algorithm accesses the domain and codomain. Following~\cite{DGKS08, GS14, BGSW}, we assume the domain is given explicitly as a direct product of cyclic groups of prime power order. We remark that representing the domain in terms of a presentation by generators and abelian relations would suffice, if we are also given a superset of the prime divisors of the order of the domain. Without that additional information, factoring would be required (see Section~\ref{section:prelim-smith}). --- We only require black-box access to the codomain (see Definition~\ref{def:black-box access}).

\paragraph{Pointer.} We prove Theorems~\ref{thm:intro-main-abel} and~\ref{thm:intro-abel-split} in Section~\ref{section:abelian}. The essential new result is the CombEcon bound, proved in Section~\ref{section:abelian-comb}. The algorithm is an adaptation of the algorithm of~\cite{DGKS08, GS14}, based on our CombEcon bound. This adaptation will be discussed in Section~\ref{sec:abelian-to-anything-algorithm}.

\subsection{Shallow random generation and list-decodability}
\label{section:intro-SRG}

We shall consider groups with the property that a bounded number of 
random elements tend to generate a subgroup of bounded depth (see Definitions~\ref{def:srg} and~\ref{def:SRG} below). This class includes the alternating groups. We show that groups in this class are CombEcon, and under minimal assumptions on access they are also CertEcon.

It will be useful to consider an $H$-independent lower bound on the quantity $\Lambda_{G,H}$. 

\begin{definition}
We define $\Lambda_G^* = \min \{ \Lambda_{G,H} :  \Lambda_{G,H} \neq 0,  H \in \groups \}$. 
\end{definition}

\begin{observation}
For simple groups the following three quantities are equal: (a) $\Lambda_G^*$\,, (b) $\Lambda_{G,G}$\,, and (c) the largest fraction of elements of $G$ fixed by an automorphism. 
\end{observation}
\begin{observation}
For $G = A_n$, $n \geq 5$, we have $\Lambda_G^* = 1/\binom{n}{2}$. 
\end{observation}

The \defn{depth} of a subgroup $M$ in a group $G$ is the length $d$ of the longest subgroup chain $M = M_0 < M_1 < \cdots < M_d = G$. We say that a subgroup is ``shallow'' if its depth is bounded. It follows from a result of~cite{Babai1989} that already a pair of elements in $A_n$ generates a subgroup of depth at most 6. This is the property that we generalize.

\begin{definition}[Shallow random generation] 
\label{def:srg}
Let $k, d \in \NN$. We say that a finite group $G$ is \defn{$(k,d)$-shallow generating}
if 
	\begin{equation}
	\label{eq:SRG-suff-condition}
	\Pr_{g_1, \ldots, g_k \in G}[ \depth(\langle g_1, \ldots, g_k \rangle) > d] < (\Lambda_G^*)^k.
	\end{equation}
\end{definition}

\begin{definition}[SRG groups]
\label{def:SRG}
	We say that a class $\grpclass$ of finite groups has
        \defn{shallow random generation} ($\grpclass$ is SRG) if there
        exist $k, d \in \NN$ such that all $G \in \grpclass$ are
        $(k,d)$-shallow generating.
\end{definition}

\begin{lemma}
\label{lemma:intro-alt-SRG}
The alternating groups are SRG groups. In particular, for sufficiently 
large $n$, the alternating group $A_n$ is $(2,6)$-shallow generating. 
\end{lemma}
We prove this lemma in Section~\ref{section:SRG An}. We note that certain classes of Lie type simple groups are also SRG. We shall elaborate on this in a separate paper. 

Now we can state one of the main results of this paper. 
\begin{theorem} 
\label{thm:contr-SRG-comb}
If $G$ is an SRG group, then $G$ is universally CombEcon list-decodable. 

\end{theorem}

For the case of alternating groups, we
show that the degree of the $\poly(1/\epsilon)$ list-size bound is at most 9; with further work this can
be reduced to 7.
\begin{theorem} 
\label{thm:contr-SRG-cert}
If $G$ is an SRG group, then $G$ is universally strong CertEcon list-decodable. 
\end{theorem}

In fact, SRG groups are universally strong $\cWa_{\Lambda_{G,H}}$-CertEcon list-decodable (see Section~\ref{section:terminology-subword} for the definition of $\cW$-certificates and section \ref{section:contr-cert-discussion} for the definition of $\cWa_{\Lambda_{G, H}}$). This restriction on the type of certificates we obtain is necessary for extensions to AlgEcon results (cf. comment before Definition~\ref{def:intro-subword-extender}). Section~\ref{section:contr-cert-discussion} discusses $\cW$-certificates in the context of homomorphism codes. A formal statement of the $\cWa_{\Lambda_{G,H}}$-CertEcon result is given in Section~\ref{section:contr-cert-results}.

\paragraph{Access model.}
For the CertEcon results, we assume access to (nearly) uniform random elements
of the domain. We do not multiply elements of the domain, so we do not need 
black-box access to the domain.  However, representing the domain as an
encoded black-box group suffices for random generation (see 
Sec.~\ref{section:prelim-blackbox}).

We need no access to the codomain. 

\paragraph{Pointers.}
We prove the CombEcon result in Section~\ref{section:SRG implies CombEcon} and the  CertEcon result in Section~\ref{section:SRG implies CertEcon}. For alternating groups we also give another, non-algorithmic, proof of the CombEcon result in Section~\ref{section:alternating}. That proof relies on a generic sphere packing argument to split the sphere into more tractable bins (see Lemma~\ref{lemma:sphere-packing} and Section~\ref{section:alt-tools-2}).

\subsection{Certificate list-decoding for homomorphism codes}  
\label{section:contr-cert-discussion}

First we translate the concepts associated with certificate list-decoding (Section~\ref{section:terminology-cert}) to the context of homomorphism codes. A \defn{certificate} $\gamma$ is a $G \partialto H$ partial map  that extends uniquely to an affine homomorphism $\varphi \in \aHom(G,H)$. 

A \defn{subword extender} is an algorithm that extends a $G \partialto H$ partial map to a full homomorphism if possible. 

Recall that for a subset $S \subseteq G$, we denote by $\mu_G(S) := \abs{S}/\abs{G}$ the \defn{density} of $S$ in $G$. For notational simplicity, we write $\Lambda$ for $\Lambda_{G,H}$. 

\begin{notation}
Let $\cW_\lambda$ (resp.~$\cWa_\lambda$) be the set of $G \partialto H$ partial maps $\gamma$ such that $\mu(\langle \dom(\gamma) \rangle) > \lambda$ (resp.~$\mu(\gencoset{\dom(\gamma)}) > \lambda$). 
\end{notation}

Recall that we have introduced certificate list-decoding as an intermediate step towards algorithmic list-decoding, to address technical difficulties that arise in algorithmic list-decoding in the alternating case. 
Our plan is to apply Observation~\ref{obs:cert-to-alg} on subword extension
with $\cW = \cWa_{\Lambda}$.  

\begin{observation}
If a partial map $\gamma \colon G \partialto H$ belongs to $\cWa_{\Lambda}$, then $\gamma$ extends to at most one affine homomorphism in $\aHom(G,H)$. 
\end{observation}
We will find $\cWa_{\Lambda}$-certificate-list-decoders for a large class of homomorphism codes, and we wish to find corresponding $\cWa_{\Lambda}$-subword-extenders.

Let $\gamma$ be a $G \partialto H$ partial map. We present three conditions on $\gamma$, then discuss their relationships to each other as well as to list-decoding. 
\begin{enumerate}[(1)]
\item If $\gamma$ extends to an affine homomorphism in $\aHom(G,H)$, then the extension is unique, i.e., $\gamma$ is a certificate for some affine homomorphism. 
\item $\mu \left( \gencoset{\dom(\gamma)} \right) > \Lambda $. 
\item The affine closure of $\dom(\gamma)$ is $G$.  
\end{enumerate}
Clearly, Condition (3) implies Condition (2), which implies Condition (1). Implications in the other direction do not hold in general. In particular, neither reverse implication holds for the alternating groups. 

Algorithmic list-decoding requires a list of full affine homomorphisms.
(Recall that affine homomorphisms are represented as partial maps 
satisfying Condition (3).)

Certificate list-decoding requires the list of partial maps to satisfy 
Condition (1).   
Our CertEcon algorithms actually return certificates satisfying Condition (2),
i.\,e., they are $\cWa_{\Lambda}$-certificate-list-decoders.

In the case of abelian $G$, Condition (3) is equivalent to Condition (1) if the irrelevant kernel is trivial (see Definition~\ref{def:irrelevant-kernel}). So, in this case certificate list-decoding and algorithmic list-decoding are equivalent. We introduced the mean-list-decoding machinery to address the case of nontrivial irrelevant kernel (see Theorems~\ref{thm:main-meanLD-iff-LD} and~\ref{thm:mean-quotient-implies-wholegroup}).

\subsection{Certificate list-decoding: SRG $\to$ arbitrary }
\label{section:contr-cert-results}

Recall that, in the context of list-decoding $\aHom(G,H)$, $\cWa_\Lambda$ denotes the set of $G \partialto H$ partial maps $\gamma$ such that $\mu( \gencoset{ \dom(\gamma) }) > \Lambda$, where $\Lambda = \Lambda_{G,H}$. 
We state the promised strengthening of Theorem~\ref{thm:contr-SRG-cert}.
\begin{theorem}[SRG certificate, abridged]
\label{thm:contr-SRG-cert-stronger}
If $G$ is an SRG group, then $G$ is universally strong 
$\cWa_\Lambda$-CertEcon list-decodable.
\end{theorem}

\paragraph{Access model.} We assume access to (nearly) uniform random elements of the domain. We do not multiply elements of the domain.  We remark that representing the domain as a black-box group would suffice for random generation~cite{Bab91BBpolygen}. 

We need no access to the codomain. We get ahold of elements of the codomain by querying the received word. We shall not perform any group operations in the codomain. 

Actually our conclusion is much stronger than what would be implied by 

Theorem~\ref{thm:contr-SRG-cert-stronger}.
\begin{theorem}[SRG certificate, unabridged]
\label{thm:contr-srg-cert-unabr}
Let $G$ be a $(k,d)$-shallow generating group and $H$ an arbitrary group.  We have a local algorithm with the following features. \\

 \textbf{Input:} 

Values $\eps, \eta > 0$. \\
  \textbf{Output:} A set $\Pi \subseteq G^{k+d+1}$ of $({k+d+1})$-tuples in $G$, where 
 $$\abs{\Pi} = \left\lceil\frac{1}{\eps^{k+d+1}} 
   \ln\left( \frac{1}{\eta \eps^{k+d+1}} \right) \right\rceil.$$
 \textbf{Cost:}  $\poly(1/\eps, \ln(1/\eta))$ amount of work. \\ 
 \textbf{Performance guarantee:} For every received word $f \in H^G$, with probability at least $(1-\eta)$, the set $\Gamma := \{ f|_R : R \in \Pi \}$ is $\cWa_{\Lambda}$-certificate-list for $\aHom(G,H)$ up to distance $(\mindist - \eps)$ of $f$. 

\end{theorem}
\paragraph{Access model.}  Same as in 
Theorem~\ref{thm:contr-SRG-cert-stronger}.

\paragraph{Pointer.}  The proof of Theorems~\ref{thm:contr-SRG-cert-stronger} 
and~\ref{thm:contr-srg-cert-unabr}
can be found in Section~\ref{section:SRG implies CertEcon}.

\begin{remark}
Given that $A_n$ is $(2,6)$-shallow generating 
(Lemma~\ref{lemma:intro-alt-SRG}), Theorem~\ref{thm:contr-srg-cert-unabr}
applies to $A_n$ with $k+d+1=9$.  We think of $A_n$ being given in its
natural permutation representation.  We note that a representation
of $A_n$ as a black-box group would suffice, because the natural permutation 
representation of an alternating group can be 
efficiently extracted from a black-box group representation~cite{BBtoAlt}.  
\end{remark}

\subsection{Algorithmic list-decoding, assuming certificate list-decoding and homomorphism extension}
\label{section:contr-cert-homext-to-alg}

In the light of Observation~\ref{obs:cert-to-alg}

(a certificate-list-decoder and a subword extender combine to a list-decoder) 
and the CertEcon results stated above, we need subword 
extenders for homomorphism codes.  

The \emph{homomorphism extension problem} is the same as the subword 
extension problem for \linebreak $\Hom(G,H)$.  We shall see below that it
can also be used to solve the subword extension problem
for $\aHom(G,H)$.

The Homomorphism Extension Problem asks whether a $G\partialto H$
partial map extends 
to a homomorphism on the whole group. 
As before, let $\Lambda = \Lambda_{G,H}$. 
\begin{definition}(Homomorphism Extension, $\HomExt(G,H)$)\label{def:homext-problem}\\
\indent \textbf{Instance:} A partial map $\gamma: G \partialto H$. \\
\indent \textbf{Solution:} A homomorphism $\varphi \in \Hom(G,H)$ that 
extends $\gamma$, i.\,e., $\varphi|_{\dom\gamma} = \gamma$. 
\end{definition}
The Homomorphism Extension Decision Problem asks whether a solution exists. 
The Homomorphism Extension Search Problem asks to determine
whether a solution exists and, if so, to find one.

Let $M$ denote the subgroup of $G$ generated by the domain of $\gamma$.
The Homomorphism Extension problem splits into the following two questions.
\begin{itemize}
\item[(a)] Does $\gamma$ extend to an $M\to H$ homomorphism?  (If such an 
  extension exists, it is clearly unique.)
\item[(b)] Given an $M\to H$ homomorphism, does it extend to a $G\to H$
homomorphism?
\end{itemize}
Question (a) can be solved efficiently if a presentation of $M$
is available in terms of the set $\dom(\gamma)$ of generators
and we have black-box access to $H$ (see Prop.~\ref{prop:prelim-homcheck}).
Such presentation can always
be found efficiently if $G$ is given as a permutation group.
(Prop.~\ref{pres-of-perm-groups}).

The difficult problem is to extend a homomorphism from $M$ to $G$.
For $G=A_n$ we are only able to do this when $M$ has polynomial index
in $G$.  Therefore we consider the threshold version of the problem.

\begin{definition}(Homomorphism Extension with Threshold, $\HomExt_\lambda(G,H)$) \label{def:homext-problem-lambda} \\
\indent \textbf{Instance:} A number $\lambda > 0$ and a
partial map $\gamma: G \partialto H$ satisfying 
$\mu(\langle \dom \gamma \rangle) > \lambda$.\\
\indent \textbf{Solution:} A homomorphism $\varphi \in \Hom(G,H)$ that 
extends $\gamma$, i.\,e., $\varphi|_{\dom\gamma} = \gamma$. 
\end{definition}
Note that, if $\lambda_1 \leq \lambda_2$, then an oracle for 
$\HomExt_{\lambda_1}(G,H)$ can answer the $\HomExt_{\lambda_2}(G,H)$ queries. 

Next we reduce the extension problem for affine homomorphisms
to the $\HomExt$ problem, i.\,e., the extension problem for homomorphisms.
\begin{proposition}
\label{prop:contr-cert-extenders-to-alg}
Let $G$ and $H$ be groups to which we are given black-box access. Then,
a subword extender for $\aHom(G,H)$ can be implemented in 
$\poly(\enc(G))$-time in the unit-cost model for $H$, assuming we are 
given an oracle for the $\HomExt(G,H)$ search problem --- that is, 
we have a subword extender for $\Hom(G, H)$.
\end{proposition}

\begin{proof}
Let $\gamma \from G \partialto H$ be a partial map. 
If $\dom(\gamma)=\emptyset$ then the map $G\to\{1_H\}$ to the identity
element of $H$ extends $\gamma$.  Otherwise,
fix $a \in \dom(\gamma)$. Let $\gamma_0 \from a^{-1} \cdot \dom(\gamma) \partialto H$ by $\gamma_0(g) = \gamma(a)^{-1} \gamma(a g)$. Then, $\gamma_0$ extends to a homomorphism $\phi_0$ if and only if $\gamma$ extends to an affine homomorphism $\phi$, with $\phi(g) = \gamma(a) \phi_0(a^{-1} g)$ for all $g \in G$.
\end{proof}

Since Theorem~\ref{thm:contr-SRG-cert} guarantees $\cWa_\Lambda$-certificate-lists, we need only provide a $\cWa_\Lambda$-subword extender 

(see Observation~\ref{obs:cert-to-alg}).
In this case, the $\textsc{HomExt}$ oracle may be relaxed to account for this restriction on certificates. Further elaborating on the comment after 
Remark~\ref{rmk:terminology-subword}, we note that this relaxation is critical to our application to the alternating group. 
While we are able to solve $\HomExt(A_n,S_m)$ for partial maps whose 
domain generates subgroups of polynomial index, we see little
hope to solving it for all partial maps on $A_n$.

The next result is the $\cWa_\Lambda$-relaxation of 
Proposition~\ref{prop:contr-cert-extenders-to-alg}. 
\begin{proposition}
\label{prop:contr-cert-extenders-to-alg-W}
Let $G$ and $H$ be groups to which we are given black-box access. 
Suppose we are given an oracle for the $\textsc{HomExt}_\Lambda(G,H)$
search problem.
Then, a $\cWa_\Lambda$-subword extender for $\aHom(G,H)$ 
can be implemented in $\poly(\enc(G))$ time in the unit-cost model for $H$.  
 \end{proposition}
The proof is the same as that of 
Proposition~\ref{prop:contr-cert-extenders-to-alg}.
\begin{remark}
In practice we may not be able to determine the value of $\Lambda$, while we may be able to determine a rather large lower bound $\lambda \leq \Lambda$. So, we instead ask for an oracle for $\HomExt_\lambda(G,H)$. This is the procedure we follow in this paper for the alternating group.
\end{remark}

\begin{corollary}
\label{cor:contr-SRG-AlgEcon}
Let $G$ be an SRG group and $H$ be an arbitrary group. Under the 
assumptions of Proposition~\ref{prop:contr-cert-extenders-to-alg-W}, 
$\aHom(G,H)$ is AlgEcon. 
\end{corollary}
\begin{proof}
Combine Proposition~\ref{prop:contr-cert-extenders-to-alg-W} 
and Theorem~\ref{thm:contr-SRG-cert-stronger}.
\end{proof}

\subsection{Homomorphism extension from alternating groups}
\label{section:intro-homext-result} 

The following theorem addresses the $\HomExt$ Search Problem for the permutation representations of the alternating groups. This is the main result of~\cite{HE}.
\begin{theorem}[Wuu]  \label{thm:homext}
Let $G=A_n$, $H = S_m$ and $\lambda = 1/\poly(n)$. 
If $m < 2^{n-1}/\sqrt{n}$, then the $\HomExt_\lambda(G,H)$ 
search problem can be solved in $\poly(n,m)$ time. 
\end{theorem}

\begin{remark}
In fact, under the assumptions of Theorem~\ref{thm:homext}, 
the number of extensions can be counted in $\poly(n,m)$ time. 
\end{remark}

This result is proved by looking at the orbits in $[m]$ of the group $M$
generated by the domain of the partial function, then deciding how they 
may combine to form orbits of $G$. We reformulate $\HomExt$ with 
symmetric codomain as an exponentially large instance of a generalized
Subset Sum Problem to which we have oracle access. 
The technical assumption $m < 2^{n-1}/\sqrt{n}$ guarantees 
that the arising instance of generalized Subset Sum is tractable. 
Answering oracle queries amounts to solving certain problems of 
computational group theory such as the conjugacy problem for
permutation groups.

\subsection{Algorithmic list-decoding: alternating $\to$ symmetric, 
restricted cases }
\label{section:intro-contribution-altalg}

We need one more ingredient before we can prove our main algorithmic result.

\begin{lemma}
	\label{lem:lambda-alternating}
	Let $n \geq 10$. Let $G = A_n$ and let $H$ be a group. If 
	$\Lambda_{G, H} \neq 0$, 
	then either $\Lambda_{G, H} = 1/\binom{n}{2}$ or $\Lambda_{G, H} = 1/n$. 
\end{lemma}
\begin{proof}

	We note that $\Lambda_{G, G} \geq 1/\binom{n}{2}$, since the identity automorphism of $G$ and the automorphism that sends $g$ to its conjugation by the transposition $(1 \, 2)$ agree on $G_{\{1, 2\}}$ which has index $\binom{n}{2}$ (In fact, $\Lambda_{G, G} = 1/\binom{n}{2}$).
	
	Suppose $\Lambda_{G, H} \neq 0$, so $\Hom(G,H)$ is nontrivial. Since $A_n$ is simple, $H$ contains an isomorphic copy of $A_n$ (The image of a nontrivial homomorphism is isomorphic to $A_n$). So, $\Lambda_{G, H} \geq \Lambda_{G,G} \geq 1/\binom{n}{2}$. By Fact~\ref{fact:equalizer-coset} and the Jordan-Liebeck Theorem (see Section~\ref{section:alt-background}), $\Lambda_{G, H} = 1/\binom{n}{2}$ or $1/n$.

\end{proof}

We remark that Guo \cite[Proposition~6.1]{Guo15} proved that $1/\binom{n}{2} \leq \Lambda_{A_n, A_n} \leq 1/n$ for $n \geq 5$.

We have now stated all the ingredients needed for the AlgEcon result for alternating domains. 
\begin{theorem}[AlgEcon for alternating domains]
\label{thm:main-alt-alg}
If $G=A_n$ is an alternating group and $H=S_m$ is a symmetric group, 
then $\aHom(G,H)$ is AlgEcon, assuming $m < 2^{n-1}/\sqrt{n}$.
\end{theorem}

\paragraph{Access model.} 
We assume both $A_n$ and $S_m$ are given in their natural permutation representations.

\begin{proof}The proof follows  the ``CertEcon with HomExt implies AlgEcon'' approach discussed in Section~\ref{section:contr-cert-homext-to-alg}.

Lemma~\ref{lemma:intro-alt-SRG} shows that alternating groups 
are SRG groups, which are universally $\cW_\Lambda$-CertEcon by 
Theorem~\ref{thm:contr-SRG-cert-stronger}. 
Now $\Lambda_{G,H} \ge 1/\binom{n}{2}$ 
by Lemma~\ref{lem:lambda-alternating}. 
Theorem~\ref{thm:homext} shows that \linebreak $\HomExt_{1/\binom{n}{2}}(G,H)$ can 
be solved in $\poly(n,m)$ time, under the stated restrictions on the 
codomain $H$. 

This $\poly(n,m)$-time subword extender combines with 
the $\cW_\Lambda$-CertEcon claim to justify the AlgEcon claim
via Observation~\ref{obs:cert-to-alg}.

\end{proof}

\section{Bipartite covering arguments}
\label{sec:bipartite}

\subsection{Bipartite covering lemma}

In this section we describe a simple combinatorial lemma that will be
used in two separate contexts throughout this section
(mean-list decoding with application to the domain-relaxation
principle and the equivalence of efficiency of list-decoding
$\Hom$ and $\aHom$.  The applications are both combinatorial and
semi-algorithmic.

We write 
 $X = (V, W; E)$ to denote a bipartite graph
with given vertex partition $(V,W)$ (all edges go between $V$ and $W$).
We denote the set of neighbors of vertex $u$ by $N(u)$.

\begin{lemma}   
\label{lemma:mean-comb-bipartite}
\label{lemma:bipartite}
Let $\delta, \eta, L > 0$. 
Let $X = (V, W; E)$ be a bipartite graph. Suppose 
$\deg(v) \leq L$ for all $v \in V$ and $\deg(w) \geq \delta \abs{V}$ 
for all $w \in W$.  Then the following hold.
\begin{itemize}
\item[\emph{(a)}]\emph{(double counting lemma)}
     \quad       $\abs{W} \leq L/\delta$. 
\item[\emph{(b)}]\emph{(bipartite covering lemma)}
\quad Set 

$s=\left\lceil\frac{4}{3\delta}(\ln(L/(\eta\delta))\right\rceil$.
Choose a sequence $(u_1,\dots,u_s)\in V^s$ uniformly at random.
Create a set $U \subseteq V$ by 
independently including each $u_i$ with probability $3/4$.
Then with probability $\geq (1-\eta)$, we have 
$W = \bigcup_{u \in U} N(u)$.
\end{itemize}
\end{lemma}
\begin{proof}
(a) Count edges two ways. 
$$ L \cdot \abs{V} \geq \sum_{v \in V} \deg(v)  = \sum_{w \in W} \deg(w) \geq \delta \abs{W} \abs{V}.$$
	So, $\abs{W} \leq L /\delta$. 
	
(b)

Given $u_1,\dots,u_s$, choose 
$\hat{u}_1, \ldots, \hat{u}_s \in V \cup \{\star\}$ independently as 
follows. 
For each $i = 1, \ldots, s$, let $\hat{u}_i$ be $u_i$ with 
probability $3/4$ and $\star$ otherwise. 
Define the neighbor set of $\star$ by $N(\star) = \emptyset$. 

Fix $w \in W$. We have $\Pr_{v \in V}( w \in N(v))\geq \delta$ by 
assumption. So, for each $i$, $\Pr_{\hat{u}_i} ( w \in N(\hat{u}_i)) = 
\frac34 \cdot \Pr_{u_i} (w \in N(u_i)) \geq \frac34 \delta$. Since the $\hat{u}_i$ were chosen independently, 
\begin{equation*}
\Pr\left( w \notin \bigcup_{i =1}^s N(\hat{u}_i) \right) \leq \left( 1- \frac34\delta \right)^s.
\end{equation*}

Taking the union bound over $w \in W$, we find that 
$$\Pr\left( W \nsubseteq \bigcup_{i =1}^s N(\hat{u}_i) \right) \leq \abs{W} \left( 1- \frac34\delta \right)^s \leq \frac{L}{\delta}  \left( 1- \frac34\delta \right)^s \leq \eta.$$

\end{proof}

\subsection{Mean-list-decoding}
\label{section:mean}

This is achieved using the concept of mean-list-decoding,
introduced in Section~\ref{sec:mean-list}.

\subsection{List size versus mean-list size}

In this section we discuss results that apply to all codes,
not just to homomorphism codes.

The main result of this section is Lemma~\ref{lemma:mean-central-meanlist-relationship}, which shows that mean-lists are contained in a small number of random lists, with a slight degradation of the parameters. That mean-list size is 
bounded via list-size is shown by item (i) of Lemma~\ref{lemma:mean-central-meanlist-relationship}). 
It follows immediately that the concepts of CombEconM and CombEcon are equivalent (Corollary~\ref{cor:mean-CombEcon}). 
Lemma~\ref{lemma:mean-central-meanlist-relationship} item (ii) shows the equivalence of AlgEconM with AlgEcon and CertEconM with CertEcon, completing the proof of Theorem~\ref{thm:main-meanLD-iff-LD}. Further consequences of Lemma~\ref{lemma:mean-central-meanlist-relationship} will follow in Section~\ref{section:mean-irrelevant}, leading to the constraint-relaxation principle on the domain.

Recall that $\ell(\cC, \lambda)$ denotes the maximum list size for $\cC$ 
with agreement $\lambda$.
Now we define the analogous quantities for mean-lists, slightly
refining Def.~\ref{def:mean-list-general}.
Let $\cC$ be a code, $r$ and $s$ natural numbers, and $\lambda, \delta > 0$.
\begin{definition}[Mean-list-size]
The \defn{maximum $r$-mean-list size} for $\cC$ with agreement $\lambda$, denoted $\mean[r](\cC,\lambda)$, is the maximum size of the mean-lists $\cL(\cC, \scrF, \lambda)$ over all families $\scrF$ of $r$ received words, i.e.,
	\begin{equation*}
	\mean[r](\cC, \lambda) = \max \{ \abs{\cL(\cC, \scrF, \lambda) } : \abs{\scrF} = r \}.
	\end{equation*}

The \defn{maximum mean-list size} for $\cC$ with agreement $\lambda$ is the maximum over the $r$-mean-list sizes for $\cC$ with agreement $\lambda$, i.e., 
\begin{equation*}
\mean(\cC, \lambda)  = \max_r \mean[r](\cC, \lambda).
\end{equation*}
\end{definition}

Note that $\mean[1](\cC,\lambda) = \ell(\cC, \lambda)$. 

{}From the definitions it follows that $\aHom(G,H)$ is CombEconM if and only if $$\mean(\aHom(G,H), \Lambda_{G,H} + \eps) = \poly(1/\eps)\,.$$

\begin{notation}
For a word $w$, we denote by $w * r = (\overbrace{w \ldots w}^r)$ the word found by concatenating $r$ copies of $w$. For a set $\cS$ of words, we write $\cS*r := \{ w* r : w \in \cS\}$. 
\end{notation}
\begin{remark}[Mean-list-decoding versus repeated codes]
\label{rmk:mean-identification}
Let $\scrF = \{ f_i: i \in [r] \}$ be a family of $r$ received words. Notice that $\cL( \cC * r, (f_1, \ldots, f_r), \lambda)$ is the $r$-fold repetition of $\cL(\cC, \scrF, \lambda)$, i.\,e.,
	$$\cL(\cC, \scrF, \lambda) * r = \cL( \cC * r, (f_1, \ldots, f_r), \lambda). $$ 
It follows that $\mean[r](\cC, \lambda) = \ell(\cC * r, \lambda)$.
In this way, mean-list-decoding can be viewed as list-decoding repeated codes. 
\end{remark}

Next we state the central result of this section: every mean-list is covered by a small number of lists. 

\begin{lemma}[Concentration of mean-lists]
\label{lemma:mean-central-meanlist-relationship}
Let $\cC$ be a code and $\lambda, \delta, \eta > 0$. 
Let $\scrF = \{ f_i: i \in I \}$ be a family of received words. 
Let $\cL = \cL(\cC, \scrF, \lambda + \delta)$. Then the following hold.
\begin{itemize}
\item[\emph{(i)}] $\abs{\cL} \leq \ell(\cC, \lambda) / \delta$.

\item[\emph{(ii)}] 
Set
$s=\left\lceil \frac{4}{3\delta}(\ln  \ell(\cC, \lambda)  + 
\ln(1/\eta\delta)\right\rceil$.
Choose a sequence $(j_1,\dots,j_s)\in I^s$ uniformly at random.
For each $i$ $(1\le i\le s)$ independently apply the
list-decoder to the received word $f_{j_i}$ with
agreement threshold $\lambda$.  Let $\cL_i$ denote the output list.

Then, with probability $\geq 1- \eta$, we have 
$\cL \subseteq \bigcup_{i=1}^s \cL_i.$
\end{itemize}
\end{lemma}

Not only does this lemma allow us to
give combinatorial bounds for mean-lists in terms 
of lists, it will also be used to construct a 
(certificate-)mean-list-decoder 
from a (certificate-)list-decoder. 

The proof will follow from the Bipartite covering lemma 
(Lemma~\ref{lemma:bipartite}) together with the following
observation.

\begin{lemma}[Markov degredation] 
\label{lemma:mean-markov-degredation}
Fix a codeword $\vf$.  Let $\scrF=\{f_i : i\in I\}$
be a family of received words in the codespace.  Assume 
$\EE_i(\agr(\vf,f_i))\ge \lambda+\delta$.
Then $\PP_i(\agr(\vf,f_i)> \lambda) > \delta.$
\end{lemma}

\begin{proof}
Let $x_i=\dist(\vf,f_i)=1-\agr(\vf,f_i)$.
Then $\EE_i(x_i)\le 1-\lambda-\delta$.  Therefore,
by Markov's inequality,
$\PP(\agr(\vf,f_i)\le \lambda) =
\PP(\dist(\vf,f_i)\ge 1- \lambda) \le
\frac{1-\lambda-\delta}{1-\lambda} =
1 - \frac{\delta}{1-\lambda} < 1-\delta$.
\end{proof}

\begin{proof}[Proof of Lemma~\ref{lemma:mean-central-meanlist-relationship}]
We apply Lemma~\ref{lemma:mean-comb-bipartite} to the bipartite graph 
$X=(I, \cL; E)$ where the edge set $E$ consist of the pairs
$(i, \varphi) \in I \times \cL$ satisfying $\agr(f_i, \varphi) > \lambda$. 
Then, $\deg(f) \leq \ell(\cC, \lambda)$ by the definition of max list size 
$\ell$ and $\deg(\varphi) \geq \delta \abs{I}$ by 
Lemma~\ref{lemma:mean-markov-degredation}.

The decoder succeeds with probability at least $3/4$, so
$\cL_i \supseteq \cL(\cC, f_{j_i}, \lambda)$
happens with probability $\geq 3/4$ independently over $i=1,\dots,s$. 
The lemma follows from Lemma~\ref{lemma:mean-comb-bipartite}. 
\end{proof}

\begin{corollary}
\label{cor:mean-general}
For $\cC$ a code, $r$ a natural number, and $\lambda, \delta> 0$, we have 
	\begin{equation}
	\mean(\cC, \lambda + \delta) \leq\frac1\delta \mean[r](\cC, \lambda).
	\end{equation}
\end{corollary}
\begin{proof} 
Let $s$ be a natural number. By the definition of $\mean$, it suffices to show that $\mean[s](\cC, \lambda + \delta) \leq \frac1\delta \mean[r](\cC, \lambda)$. 

But, we find that $\mean[s](\cC, \lambda+\delta) \leq \mean[sr](\cC, \lambda+\delta)$ by~\cite[Lemma 3.3]{GS14} (though their lemma is stated in terms of repeated codes). By Lemma~\ref{lemma:mean-central-meanlist-relationship}, we find that $\mean[sr](\cC, \lambda+\epsilon) \leq \frac1\delta \mean[r](\cC, \lambda)$. 
\end{proof}

The following is now immediate.
\begin{corollary}
\label{cor:mean-CombEcon}
For $\cC$ a code and $\eps>0$, we have 
\begin{equation*}
\mean(\cC, 1-\mindist + \eps) \leq \frac2\eps \ell(\cC, 1 - \mindist+\eps/2).
\end{equation*}
Consequently, if a class of codes is CombEcon with degree $c$, 
then it is CombEcon with degree $c+1$. 
\end{corollary}
Next, we derive the algorithmic versions of this result. We shall make the following assumption on access to our family $\scrF = \{f_i : i \in I \}$ of received words. 
\begin{access}
\label{access:mean-oracle}
An oracle provides uniform random elements of the index set $I$ of $\scrF$. 
\end{access}
\begin{theorem}  \label{thm:algecon-iff-algeconm}
Under Access~\ref{access:mean-oracle}, if a class $\scrC$ of codes is AlgEcon then it is AlgEconM. Under the same assumptions, if $\scrC$ is CertEcon then it is CertEconM. 
\end{theorem}
\begin{remark}
\label{rmk:constants-deterioration}
The bounds on cost in the result above deteriorate as follows. 
\begin{itemize}
\item A $2/\eps$ multiplicative factor in list size. 
\item An $O(\frac1\eps \ln(1/\eps))$ multiplicative factor in queries to the received word 	$f$. 
\item An $O(\frac1\eps \ln(1/\eps))$ multiplicative factor in amount of work.
\end{itemize}
\end{remark}
We show that, to (certificate-)mean-list-decode a family $\scrF$ of functions, \linebreak (certificate-)list-decoding a small random subset of the functions in $\scrF$ suffices. Lemma~\ref{lemma:mean-comb-bipartite} already contains the machinery to guarantee the necessary probability of success. 
\begin{proof}[Proof of Theorem~\ref{thm:algecon-iff-algeconm}]
We first prove the claim for AlgEcon. Let $\Decode$ be a list-decoder for the class $\scrC$ satisfying AlgEcon assumptions. We denote by $\Decode(\cC, f, 1-\mindist+\eps)$ the output of $\Decode$ on the input $f$ and $\eps > 0$, where $f$ is a received word in the code space of a code $\cC \in \scrC$.

We describe here a  mean-list-decoder that satisfies AlgEconM. It takes as input $\scrF$ and $\eps > 0$, where $\scrF$ is the family of received words in the code space of the code $\cC \in \scrC$. 

Denote by $I$ the index set of $\scrF$. Via the provided oracle, generate a subset $S \subseteq I$ by picking $s$ elements of $I$ independently and uniformly. The value of $s$ will be determined later. Return the list given by 
\begin{equation}
\bigcup_{i \in S} \Decode(\cC, f_i, 1-\mindist+\epsilon/2). 
\end{equation}

We show that this is a list-decoder satisfying the conditions of AlgEconM through direct application of Lemma~\ref{lemma:mean-central-meanlist-relationship}.

The output $\Decode(\cC, f, 1-\mindist+\eps/2)$ contains  $\cL(\cC, f, 1-\mindist+\eps/2)$ with probability $3/4$ by the definition of list-decoder. If $s $ is set as in Lemma~\ref{lemma:mean-central-meanlist-relationship} (ii) with $\eta = 1/4$ and $\delta = \eps/2$, we find that the the desired mean-list is returned with probability at least $3/4$.

Notice that the list-decoder $\Decode$ is called $s = \left\lceil \frac{8}{3\eps} \left(\ln( \ell(\cC, 1-\mindist+\eps/2) + \ln( 8/ \eps) \right)\right\rceil$ times as a subroutine. Very little processing is done outside of these calls. Moreover, since $\cC$ is AlgEcon, it is CombEcon, so $\ell(\cC, 1-\mindist+\eps/2) = \poly(\eps/2)$ and $\Decode$ is called $O( \frac1\eps \ln (1/\eps))$ times.\\

The proof for CertEcon is similar, found mainly by replacing the  occurrences of ``list-decoder'' with ``certificate-list-decoder.'' The mean-certificate-list-decoder returns the union of $s$ output lists by  $\Decode$, which would denote the assumed certificate-list-decoder. 
\end{proof}
The same conclusions follow for the strong versions of these concepts. 

\begin{remark}
Knowledge of $\mindist$ is not needed in the conversion from CertEcon to CertEconM. Even in the AlgEcon case, $\mindist$ is only needed if required by the list-decoder $\Decode$. The crucial knowledge for this conversion is $\eps$, so that the deterioration factor (denoted $\delta$ above) can be controlled. This deterioration factor is set to $\delta = \eps/2$ in our proofs.
\end{remark}

\subsection{Irrelevant normal subgroups and the domain relaxation
 principle}
\label{section:mean-irrelevant}

In this section we present the principle of lifting constraints on the domain.
An example is the automatic extension of $\{$abelian$\to$abelian$\}$ 
results to the $\{$arbitrary$\to$abelian$\}$ context
(Theorem~\ref{thm:mean-quasivariety}.
(See the discussion in Section~\ref{section:contr-extendingdomain} and  
Remark~\ref{rmk:intro-mean-significance}.)

The key concept is the \emph{irrelevant kernel} $N$ for a pair $(G,H)$
of groups, defined as the intersection of the kernels of 
all $G\to H$ homomorphisms.  We shall find that extending $G/N$
to $G$ retains economical list-decodability.

We first identify the code $\aHom(G,H)$ with a repeated code found from $\aHom(G/N,H)$. This hinges on $N$ being an irrelevant normal subgroup.
Recall that $N$ is $(G,H)$-irrelevant if $N$ is contained in the kernel of every $G \to H$ homomorphism (see Definition~\ref{def:irrelevant-kernel}).

For groups $K$ and $H$, an enumeration $K = \{k_1, \dots, k_{\abs{K}}\}$ induces a bijection between the set of functions $H^K$ and the set of words $H^{\abs{K}}$ by $f \mapsto (f(k_1), \dots, f(k_{\abs{K}}))$.

\begin{observation}[Identification of $\aHom(G,H)$ lists with $\aHom(G/N, H)$ mean-lists]
\label{obs:mean-identification}

Let $G, H$ and $N$ be groups such that $N$ is a $(G, H)$-irrelevant normal subgroup. Let $f \from G \to H$. There are enumerations of $G$ and $G/N$, and a family $\scrF$ of functions $G/N \to H$ such that 
\begin{align*}
	\cL( \aHom(G,H), f, \lambda) = \cL(\aHom(G/N, H), \scrF, \lambda) * \abs{N}.
\end{align*}
\end{observation}

\begin{proof}
Let $S = \{s_1, \dots, s_{\ind{G}{N}}\}$ be a set of coset representatives of $N$ in $G$. We write $N = \{g_1, \ldots, g_{\abs{N}}\}$. We enumerate the elements of $G/N$ as $(s_1 N, \dots, s_{\ind{G}{N}} N)$, and we enumerate the elements of $G$ by concatenating $(s_1 g_i, \dots, s_{\ind{G}{N}} g_i)$ for $i = 1, \dots, \abs{N}$. We thereby think of any function $G/N \to H$ as a word in $H^{\ind{G}{N}}$, and any function $G \to H$ as a word in $H^{\abs{G}}$.

For $i = 1, \dots, \abs{N}$, let $f_i \from G/N \to H$ by $f_i(s N) = f(s g_i)$ for all $s \in S$. Let $\scrF = \{f_i : i = 1, \dots, \abs{N}\}$.

Let $\pi \from G \to G/N$ be the projection onto cosets. We note that $\aHom(G,H) = \{\varphi \circ \pi : \varphi \in \aHom(G/N,H) \}$, since $N$ is a $(G,H)$-irrelevant subgroup. Then, interpreting functions as codewords as above, $\aHom(G, H) = \aHom(G/N, H) * \abs{N}$. In other words, the codeword $\varphi \circ \pi $ is a concatenation $(\varphi, \ldots, \varphi)$ of $\abs{N}$ copies of the codeword $\varphi$. 

Furthermore, for any $\varphi \in \aHom(G/N, H)$, we have that
\begin{align*}
\agr(f, \varphi\circ \pi) & = \frac{1}{\abs{G}}\abs{ \{ g : f(g) = (\varphi \circ \pi) (g) \}} \\ 
& = \frac{1}{\abs{N}} \sum_{i = 1 }^{\abs{N}} \frac{1}{\abs{G/N}} \abs{ \{ s \in S : f(sg_i) = (\varphi \circ \pi) (sg_i) \}} \\
& = \frac{1}{\abs{N}} \sum_{i = 1 }^{\abs{N}}  \agr(f_i, \varphi) \\ 
& = \EE_{i} [ \agr(f_i, \varphi)].
\end{align*}

This shows that $\varphi \circ \pi \in \cL(\aHom(G,H), f, \lambda)$ exactly if $\varphi \in \cL(\aHom(G/N, H), \scrF, \lambda)$. So,  
\begin{align*}
\cL(\aHom(G,H), f, \lambda) 
& = \{ \varphi \circ \pi: \varphi \in \cL(\aHom(G/N, H), \scrF, \lambda) \} \\
& =  \cL(\aHom(G/N, H), \scrF, \lambda) * \abs{N}.\nonumber \qedhere
\end{align*}
\end{proof}

\begin{remark}
	The proof of Observation~\ref{obs:mean-identification} shows more: There is a bijection between functions $f \in H^G$ and families $\mathscr{F}$ of $\abs{N}$ functions so that the equation holds.
\end{remark}

\begin{corollary}
\label{cor:mean-identification-irrelevantversion}
If $G, H$ and $N$ are groups such that $N$ is a $(G,H)$-irrelevant normal subgroup of $G$, then  
$$\ell(\aHom(G,H), \lambda) = \mean[\abs{N}]( \aHom(G/N, H), \lambda).$$
\end{corollary}

  We first illustrate the relaxation principle combinatorially, through bounds on list-size. The goal would be to conclude that, if $\grpclass \times \grpclass$ is CombEconM for a class $\grpclass$ of groups, then $\groups \times \grpclass$ is CombEcon.  This principle will generalize to CertEcon and AlgEcon as well.

\begin{remark}
If $N$ is a $(G,H)$-irrelevant normal subgroup, then $\Lambda_{G/N, H} = \Lambda_{G,H}$. 
\end{remark}

\begin{lemma}[Irrelevant normal subgroup lemma]
\label{lemma:irr-normal-subgroup}
Let $G, H$ and $N$ be groups such that $N$ is a $(G,H)$-irrelevant normal subgroup. Then, 
	\begin{equation*}
		\ell(\aHom(G, H), \Lambda_{G,H} + \eps) \leq \frac{2}{\eps} \cdot \ell(\aHom(G/N, H), \Lambda_{G,H} + \eps/2).
	\end{equation*}
\end{lemma}
\begin{proof}
Calculate
\begin{align*}
\ell(\aHom(G,H), \Lambda+\eps) & = \mean[\abs{N}]( \aHom(G/N, H), \Lambda+\eps) & \text{Corollary~\ref{cor:mean-identification-irrelevantversion}}\\ 
&  \leq \mean(\aHom(G/N,H), \lambda) & \text{ definition of $\mean$} \\ 
& \leq \frac2\eps \cdot \mean[1](\aHom(G/N,H) &\text{Corollary~\ref{cor:mean-general} with $r=1$}\\
& = \frac2\eps \cdot \ell(\aHom(G/N,H), \Lambda + \eps/2).
\end{align*}
\end{proof}

This implies that, if $\aHom(G/N, H)$ is CombEconM, then $\aHom(G,H)$ is CombEcon. This principle holds for CertEcon and AlgEcon as well. 

\begin{theorem}
\label{thm:mean-quotient-implies-wholegroup}
Let $G, H$ and $N$ be groups such that $N$ is a $(G,H)$-irrelevant normal subgroup of $G$.
\begin{enumerate}[(i)]
\item If $\aHom(G/N,H)$ is CombEcon, then $\aHom(G,H)$ is CombEcon. 
\item Under suitable access assumptions (Access~\ref{access:mean-quotient-implies-wholegroup} (ii)), if $\aHom(G/N, H)$ is CertEcon, then $\aHom(G,H)$ is CertEcon.
\item Under suitable access assumptions (Access~\ref{access:mean-quotient-implies-wholegroup} (iii)), if $\aHom(G/N,H)$ is AlgEcon, then $\aHom(G,H)$ is AlgEcon. 
\end{enumerate} 
The deterioration in cost is as described in Remark~\ref{rmk:constants-deterioration}.
\end{theorem}

\begin{access}
\label{access:mean-quotient-implies-wholegroup}
\begin{enumerate}[(i)]
\setcounter{enumi}{1}
\item (a) Elements of $N$ can be generated uniformly. (b) A transversal, i.e., an injection $G/N \to G$ that assigns a representative element to each coset, is given. (c) $G/N$ is known well enough to satisfy the CertEcon access assumptions on $\aHom(G/N,H)$. 
\item (a') Elements of $N$ can be generated uniformly and generators for $N$ are given. (b') Same as (b). (c') Same as (c), for AlgEcon. 
\end{enumerate}
\end{access}

\begin{remark}
If the access assumptions on $N$ are at least as strong as having $N$ as a black-box group, then generating (nearly) uniform elements and being given a set of generators are equivalent. If $N$ is a black-box group, generators are given by definition. Nearly uniform random elements in black-box groups can be generated in polynomial time (polynomial in the encoding length of groups elements). 
\end{remark}

\begin{remark}
For the proof of this theorem, we will actually need the --EconM versions of the assumptions, which we may assume as a consequence of Theorem~\ref{thm:main-meanLD-iff-LD}.
\end{remark}

\begin{proof}[Proof of Theorem~\ref{thm:mean-quotient-implies-wholegroup}]
Set $\Lambda = \Lambda_{G,H} = \Lambda_{G/N,H}$. Let $\pi: G \to G/N$ be the projection onto cosets. 

(i) $\aHom(G/N,H)$ is CombEconM, so $\aHom(G,H)$ is CombEcon by Corollary~\ref{cor:mean-identification-irrelevantversion}.

(ii) A certificate-list-decoder satisfying the conditions of CertEconM exists for $\aHom(G/N,H)$. Its output list $\Gamma$ is a certificate list for $\cL(\aHom(G/N,H), \scrF, \Lambda+\eps)$, where $\eps>0$ and $\scrF= \{f_1, \ldots, f_{\abs{N}} \}$ is constructed from $f$ as in Observation~\ref{obs:mean-identification}. We construct a certificate list  $\widetilde{\Gamma}$ for $\cL(\aHom(G,H), f, \Lambda+\eps)$, by replacing each $G/N \partialto H$ partial map $\gamma \in \Gamma$ with a $G \partialto H$ partial map $\widetilde{\gamma}$ defined as follows. 

Denote by $\tau: G/N \to G$ the injection guaranteed by the assumption. Let $\widetilde{\gamma}$ have domain $\dom(\widetilde{\gamma}) = \tau(\dom(\gamma))$, and define $\widetilde{\gamma}(g) = \gamma(\pi(g))$ for each $g \in \dom(\widetilde{\gamma})$. If $\gamma$ is a certificate for $\varphi \in \aHom(G/N, H)$, then $\widetilde{\gamma}$ is a certificate for $\varphi \circ \pi \in \aHom(G,H)$. 

(iii) A list-decoder satisfying the conditions of AlgEconM exists for $\aHom(G/N,H)$. By Observation~\ref{obs:mean-identification}, it suffices to, given the list $\cL = \cL(\aHom(G/N, H), \scrF, \Lambda+\eps)$, return the list $\widetilde{\cL} = \{ \varphi \circ \pi: \varphi \in \cL \}$. We address algorithmic issues of defining $\varphi \circ \pi$ from $\varphi$. 

Denote by $X$ the given set of generators of $N$ and by $\tau : G/N \to G$ the given injective map.
Each homomorphism $\varphi$ is represented by its values on a set $Y$ of generators of $G/N$. The set $X \cup \tau(Y)$ is a set of generators for $G$. Define $\widetilde{\varphi}$ on this set by the following.
\begin{equation*} \widetilde{\varphi}(g) = 
\begin{cases}
1 & g  \in X \\ 
(\varphi \circ \pi)(g) & g \in \tau(Y)\end{cases}.
\end{equation*}
\end{proof}

A class $\grpclass$ of finite groups is a \defn{quasivariety} if it is closed under subgroups and direct products. The classes of abelian, nilpotent, and solvable  groups are examples of quasivarieties. 

\begin{theorem}
\label{thm:mean-quasivariety}
Let $\grpclass$ be a quasivariety of finite groups. 
\begin{enumerate}[(i)]
\item Suppose $\aHom(G,H)$ is CombEcon for every $G,H \in \grpclass$. Then, $\aHom(G,H)$ is CombEcon for $H \in \grpclass$ and arbitrary $G$.
\item Under suitable access assumptions (Access~\ref{access:mean-quasivariety} (ii)), if $\aHom(G,H)$ is CertEcon for every $G,H \in \grpclass$, then $\aHom(G,H)$ is CertEcon for $H \in \grpclass$ and arbitrary $G$.
\item Under suitable access assumptions (Access~\ref{access:mean-quasivariety} (ii)), if $\aHom(G,H)$ is AlgEcon for every $G,H \in \grpclass$, then $\aHom(G,H)$ is AlgEcon for $H \in \grpclass$ and arbitrary $G$.
\end{enumerate}
The deterioration in cost is as described in Remark~\ref{rmk:constants-deterioration}. 
\end{theorem}

\begin{remark}
In fact, the class $\grpclass$ need only be closed under subdirect products. 
\end{remark}

The access assumptions mirror those of Access~\ref{access:mean-quotient-implies-wholegroup} of Theorem~\ref{thm:mean-quotient-implies-wholegroup}.
\begin{access} 
\label{access:mean-quasivariety}
For every $G \in \groups$ and $H \in \grpclass$, denote by $N$ the $(G,H)$-irrelevant kernel and assume we have access to $N$ and $G/N$ as follows. 
\begin{enumerate}[(i)]
\setcounter{enumi}{1}
\item (a) Random elements of $N$ can be generated uniformly. (b) A transversal of $G/N$ in $G$ can be found. (c) $G/N$ can be found well enough to satisfy the access model of the assumed CertEcon list-decodability of pairs in $\grpclass \times \grpclass$. 
\item (a')   Random elements of $N$ can be generated uniformly and a set of generators for $N$ can be found. (b') Same as (b). (c') Same as (c) but for AlgEcon list-decodability.  
\end{enumerate}
\end{access}

\begin{proof}
Fix $H \in \grpclass$ and $G \in \groups$. Let $N$ be the $(G,H)$-irrelevant kernel. By Theorem~\ref{thm:mean-quotient-implies-wholegroup}, all desired conclusions will follow if we show that $G/N \in \grpclass$.

Let 
$$\widetilde{H} = \prod_{\varphi \in \aHom(G,H)} H.$$ Define the map $\tau:G \to \widetilde{H}$ given by $\tau(g) = (\varphi(g))_\varphi$. Notice that $\tau(G)$ is subgroup of $\widetilde{H}$ and is thus a subdirect product of copies of the group $H \in \grpclass$. Since $\grpclass$ is closed under subdirect products, it follows that $\tau(G) \in \grpclass$.

Since $\ker(\tau) = N$, we have $\tau(G) \cong G/N$, so $G/N \in \grpclass$.

\end{proof}

The AlgEcon list-decodability of $\{$abelian$\to$abelian$\}$ and $\{$nilpotent$\to$nilpotent$\}$ homomorphism codes is shown in~\cite{DGKS08} and~\cite{GS14}, respectively. As the class of abelian groups and the class of nilpotent groups both form quasivarieties, we conclude the following using the mentioned results and Theorem~\ref{thm:mean-quasivariety}.
\begin{corollary}
If $G$ is a group and $H$ an abelian group (or, more generally, nilpotent), then we have AlgEcon (and therefore CombEcon)  list-decoding of $\aHom(G,H)$.
\end{corollary}

\subsection{$\Hom$ versus $\aHom$}
\label{sec:Hom-versus-aHom}

We show that the code $\Hom(G,H)$ is CombEcon
if and only if $\aHom(G,H)$ is CombEcon, and 
similarly for CertEcon, and AlgEcon under modest
assumptions of the representation of the groups.
Therefore we can use these two types of codes
interchangeably.
This reflects a phenomenon similar to
our results on mean-list-decoding.  

We fix terminology for this section. For an affine homomorphism $\varphi \in \Hom(G,H)$, we denote its \defn{base homomorphism} by $\varphi^0 \in \Hom(G,H)$ (the unique homomorphism satisfying $\varphi = h \varphi^0$ for $h \in H$). For an element $a \in G$ and function $f: G \to H$, we denote by $f^a:G \to H$ the function $f^a(g) = f(a)^{-1} f(ag)$.

We state the central result of this section, that every $\aHom$ list is contained within a small number of translated $\Hom$ lists. It is similar in spirit to Lemma~\ref{lemma:mean-central-meanlist-relationship}, and it is similarly proved by the Bipartite Covering Lemma (Lemma~\ref{lemma:mean-comb-bipartite}).  The deterioration in list size and cost will be addressed in Remark~\ref{rmk:constants-deterioration-aHom} below.

\begin{lemma}[Concentration of $\aHom$ lists]
\label{lemma:mean-central-aHom-relationship}
Let $G$ and $H$ be groups, $f: G \to H$ a received word, and $0 < \lambda \leq 1$. Let $\cL = \highagr(\aHom(G,H), f, \lambda)$. We conclude the following. 
\begin{enumerate}[(i)]
\item $\abs{\cL} \leq \frac1\lambda \ell(\Hom(G,H), \lambda)$. 
\item Let $S$ be a subset of $G$ formed by choosing $\left\lceil \frac{4}{3\lambda}(\ln  \abs{\cL}  + \ln(1/\eta\lambda)\right\rceil
$ elements in $G$ independently and uniformly. Suppose that, independently for each $a \in G$, the subset $D_a$ of $\cL$ contains $\cL(\Hom(G,H), f^a, \lambda)$ with probability $\geq 3/4$. We denote $\widetilde{D_a} = \{ f(a)\psi(a^{-1})\psi\in \aHom(G,H) :  \psi \in D_a \}$. Then, with probability $\geq 1- \eta$, we have 
$$\cL \subseteq \bigcup_{a \in S} \widetilde{D_a}. $$
\end{enumerate}
\end{lemma}
We remark that $\widetilde{D_a} \subseteq \aHom(G,H)$ is found by translating elements of $D_a \subseteq \Hom(G,H)$, but not all by the same element. 

We defer the proof to state the main result of this section, which is an immediate consequence.

\begin{access}\label{access:mean-aHom}
An oracle provides uniform random elements of $G$. 
\end{access}
\begin{corollary}[$\Hom$ versus $\aHom$]
 Let $G$ and $H$ be groups. If $\Hom(G,H)$ is CombEcon, then $\aHom(G,H)$ is CombEcon. Under Access~\ref{access:mean-aHom}, if $\Hom(G,H)$ is CertEcon, then $\aHom(G,H)$ is CertEcon. Under Access~\ref{access:mean-aHom}, if $\Hom(G,H)$ is AlgEcon, then $\aHom(G,H)$ is AlgEcon.
\end{corollary}
\begin{remark}[Deterioration]
\label{rmk:constants-deterioration-aHom}
The bounds on cost in the result above deteriorate as follows. 
\begin{itemize}
\item A $\frac1\lambda$ multiplicative factor in list size. 
\item An $O(\frac1\lambda \ln (1/\lambda))$ multiplicative factor in queries to the received word $f$. 
\item An $O(\frac1\lambda \ln (1/\lambda))$ multiplicative factor in amount of work. 
\end{itemize}
\end{remark}

Towards proving Lemma~\ref{lemma:mean-central-aHom-relationship}, 
 we state a few facts relating affine homomorphisms to their base homomorphisms.
\begin{observation}
\label{obs:mean-aHom-basehom}
Let $G$ and $H$ be groups and $\varphi \in \aHom(G,H)$. Then,
\begin{equation*}
\varphi(a)^{-1}\varphi(ag) = \varphi^0(g) \; \;\; \forall a,g\in G.
\end{equation*}
\end{observation}
\begin{corollary}
\label{cor:aHomvsHom-translatefunction}
Let $G$ and $H$ be groups, $f: G \to H$, and $\varphi \in \aHom(G,H)$. If $f(a) = \varphi(a)$, then 
\begin{equation*}
f(ag) = \varphi(ag) \iff f(a)^{-1} f(ag) = \varphi^0(g) \iff f^a(g) = \varphi^0(g). 
\end{equation*}
It follows that $\agr(f, \varphi) = \agr(f^a, \varphi^0)$. 
\end{corollary}

We can now prove our concentration lemma. 
\begin{proof}[Proof of Lemma~\ref{lemma:mean-central-aHom-relationship}]
Fix the function $f: G \to H$. Consider the bipartite graph with vertices $V = G$ and $W = \highagr(\aHom(G,H), f, \lambda)$, where the edge set contains $(a, \varphi) \in G \times \aHom(G,H)$ if $f(a)= \varphi(a)$. We wish to apply Lemma~\ref{lemma:mean-comb-bipartite}, so we first check the conditions are satisfied.

For every $\varphi \in W$, we have $\agr(\varphi, f) > \lambda$ by definition, so $\deg(\varphi) > \lambda$. 

We show that $\deg(a) \leq \ell(\Hom(G,H), \lambda)$, by showing that the map $\varphi \mapsto \varphi^0$ is an injection from $N(a)= \{\varphi \in W : f(a) = \varphi(a)\}$ to $\highagr(\Hom(G,H), f^a, \lambda)$. This map is well-defined since $f(a) = \varphi(a)$ implies $\agr(f, \varphi) = \agr(f^a, \varphi^0)$, by Corollary~\ref{cor:aHomvsHom-translatefunction}. We show this map is injective. If $\varphi_1, \varphi_2 \in \aHom(G,H)$ satisfy $\varphi_1(a) = f(a) = \varphi_2(a)$ and $\varphi^0_1 = \varphi^0_2$, then $\varphi_1(g) = \varphi_1(a)\varphi^0(a^{-1} g) = \varphi_2(a) \varphi^0(a^{-1}g) = \varphi_2(g)$ for all $g \in G$, by Observation~\ref{obs:mean-aHom-basehom}. 

Apply the two parts of Lemma~\ref{lemma:mean-comb-bipartite} to find the following. 
\begin{enumerate}[(i)]
\item We conclude that $\abs{\highagr(\aHom(G,H), f, \lambda)} \leq \frac1\lambda \ell(\Hom(G,H), \lambda)$.
\item The subset $U$ is the chosen subset $S$ of $G$. The subset $\hat{U}$ contains the element $a \in U \subseteq G$ if list-decoding for $\highagr(\Hom(G,H), f^a, \lambda)$ succeeds, which happens with probability $\geq 3/4$ independently over $a$. 

It remains to show that if $D_a = \highagr(\Hom(G,H), f^a, \lambda)$, then $\widetilde{D_a} \supseteq N(a)$. But, if $\varphi \in N(a)$, we have already established that $\varphi^0 \in \highagr(\Hom(G,H), f^a, \lambda)$. Moreover, since $f(a) = \varphi(a)$, we have $\varphi(g) = f(a)\varphi^0(a^{-1}g)$ for all $g \in G$, by Observation~\ref{obs:mean-aHom-basehom}. So, $\widetilde{D_a} \supseteq N(a)$ by the definition of $\widetilde{D_a}$. 
\end{enumerate}

\end{proof}

\section{Strategy}
\label{section:strategy-for-CombEcon}

In this section, we outline our strategy for proving CombEcon results.

Let $f\in \aHom(G,H)$ be a received word and 
let $\cL$ be the set of codewords within
distance $(\mindist-\eps)$ of $f$.
The combinatorial problem is to bound
$\abs{\cL} \leq \poly(1/\eps)$.  First, we partition $\cL$ into more manageable subsets (which we call \emph{buckets}). We bound the number of buckets using a 
sphere-packing argument, and we bound the maximum size of the buckets.

\subsection{The sphere packing argument:
strong negative correlation}
\label{sec:strategy-part-one}

We recall that the \emph{agreement} $\agr(f,g)$ of two
functions $f,g$ in the code space $H^G$
is the proportion of inputs on which 
$f$ and $g$ agree; i.e., 
$\agr(f,g)=(1/\abs{G})\abs{\{x\in G\mid f(x)=g(x)\}}$.
So, the distance between $f$ and $g$ is
$1-\agr(f,g)$. And, $\Lambda=\Lambda_{G,H}$ is the maximum agreement between elements of $\aHom(G,H)$; so the minimum distance of the code $\aHom(G,H)$ is $1-\Lambda$.

Let $\Psi$ be a maximal set of elements of $\aHom(G,H)$ such that their pairwise distance is at least $1 - \Lambda^2$. For each $\psi \in \Psi$, we let the bucket $\cL_{\psi}$ consist of all homomorphisms in $\cL$ that are within a ball of radius $1 - \Lambda^2$ around $\psi$. That is, every pair of homomorphisms in $\Psi$ has agreement at most $\Lambda^2$, and every homomorphism in $\cL_{\psi}$ has agreement greater than $\Lambda^2$ with $\psi$.

We note that every homomorphism in $\cL$ is in at least one bucket.

To bound the number of buckets, $\abs{\Psi}$, we use a sphere-packing argument based on strong negative correlation (the sets $\Eq(f,\psi)$ for $\psi\in\Psi$ are
strongly negatively correlated; see Section~\ref{section:tools-SNC-spherepacking}). We find that $\abs{\Psi}$ is at most $O(1/\eps^2)$.

We discuss the division into buckets in detail in Sections~\ref{sec:bucket-splitting} and~\ref{section:generation-graphs}.

Our strategy to bound the sizes of the buckets is different depending on the type of the domain $G$. In all cases, we further divide each bucket into smaller units (which we call \emph{sub-buckets}), but the method of division will differ. We discuss abelian $G$ in Section~\ref{sec:strategy-abelian} (in detail in Section~\ref{section:abelian}), and alternating and SRG groups $G$ in Section~\ref{sec:strategy-alternating-klc} (in detail in Sections~\ref{section:alternating}, \ref{section:SRG}, and~\ref{section:srg-consequences}).

\subsection{Bounding the list size for abelian groups}
\label{sec:strategy-abelian}

To prove that abelian groups are universally CombEcon, we prove the following structure theorem.
The theorem asserts that the codomain has a small number of abelian
subgroups so that each homomorphism in the list $\cL$ maps
domain $G$ into one of those abelian subgroups.

\begin{theorem}
	\label{7thm:codomain-few-abelian-subgroups}

	There exists a set $\abelsubs$ of finite abelian subgroups of $H$ with $\displaystyle{\abs{\abelsubs} \leq \frac{1}{4 (\Lambda + \eps) \eps^2} + \frac{1}{\eps}}$ such that for all $\phi \in \highagr$, 

	there is $M \in \abelsubs$ such that $\phi(G) \leq M$.
\end{theorem}

This is restated in Section~\ref{section:abelian} as Theorem~\ref{thm:codomain-few-abelian-subgroups}, and proved there.

This result reduces the problem to showing CombEcon of $\{\text{abelian} \to \text{abelian}\}$, which was done by Dinur, Grigorescu, Kopparty, and Sudan \cite{DGKS08}.
 
For find the set $\abelsubs$ of subgroups, we introduce the concept of an \emph{abelian enlargement}. The abelian enlargement of a subset $T \subseteq H$ by a group $B \leq H$ is the group generated by $T$ and the elements of $B$ that are in the centralizer of $T$; that is,
	\[
	\sig_B(T) = \gen{T, C_H(T) \cap B}.
	\]
Let $\abelsubs_\psi$ be the set of all groups $M$ that occur as the abelian enlargement of $f(g)$ by $\psi(G)$ for at least an $\eps$ proportion of $g \in G$.

We shall show

that every homomorphism $\phi$ in the bucket $\cL_{\psi}$ has its image contained in one of these subgroups. The idea is that since $\phi$ and $\psi$ have large agreement, most of $\phi(G)$ is contained in $\psi(G)$. So even if we take a single random element of $\phi(G)$, it is likely that its enlargement by $\psi(G)$ already contains all of $\phi(G)$. Specifically,

\begin{proposition}
	\label{78prop:sig-generate}

	Let $\phi, \psi \in \Hom(G, H)$ and $g \in G$ such that $\gen{g, \Eq(\psi, \phi)} = G$.
	
	Then $\phi(G) \leq \sig_{\psi(G)}(\phi(G)) = \sig_{\psi(G)}(\phi(g))$.
\end{proposition}

This is restated in Section~\ref{section:abelian} as Proposition~\ref{prop:sig-generate}, and proved there.

We can then use the CombEcon bound to adapt the Dinur et.\ al.\ algorithm to this more general class of codes.

\subsection{Bucket estimation for alternating and SRG} 
\label{sec:strategy-alternating-klc}

We describe our strategy to bound the size of a bucket in the case that $G$ is alternating. We carry out this strategy in Section~\ref{section:alt-proof}.

All homomorphisms in one bucket $\cL_{\psi}$ have high agreement with one representative homomorphism $\psi$. However, we have no control over where in $G$ this agreement occurs. We split each bucket $\cL_{\psi}$ further into \emph{sub-buckets} $\cL_{\psi, K}$, this time so that the homomorphisms in each sub-bucket agree on a specified large 

(low-depth) subgroup $K$ of $G$.

\textbf{Bound on the number of sub-buckets per bucket.} Fortunately, low-index subgroups of alternating groups are well understood, and there are few of them. There are at most $\poly(1/\Lambda)$ large subgroups of $G$. 

So, each bucket is divided into at most $\poly(1/\Lambda)$ sub-buckets.

It is a consequence of the strong negative correlation principle that $\poly(1/\eps) = \poly(1/\Lambda, 1/\eps)$. We will find $\Lambda = \poly(n)$ for alternating groups. 

\textbf{Bound on the size of each sub-bucket.} 

We will bound the size of a sub-bucket $\cL_{\psi, K}$. We describe a process for choosing a random homomorphism in the sub-bucket. 

For a positive integer $d$, we choose $d$ random elements of $G$. If there is a unique homomorphism $\phi$ that agrees with $f$ on the $d$ random inputs, and agrees with $\psi$ on $K$, we choose this homomorphism.

If $d$ is at least the depth of the subgroup $K$ in $G$, then each homomorphism in the sub-bucket gets chosen with probability at least $\eps^d$. So, the size of the sub-bucket is at most $1/\eps^d$. Conveniently, the depth of large subgroups of the alternating group $A_n$ is bounded (see Section~\ref{section:alternating}).

\textbf{Generalization.} This strategy for list decoding $\{\text{alternating} \to \text{anything}\}$ works more broadly; it still works when the group in the domain comes from a much larger class of groups, which we call \emph{SRG groups} (Section~\ref{section:SRG}). A group $G$ is SRG if, roughly, a small number of random elements chosen from $G$ are likely to generate a low depth subgroup. See Section~\ref{section:intro-SRG} for a precise definition.

This view allows us not only to combinatorially list-decode $\{\text{SRG} \to \text{anything}\}$, but also to certificate list-decode: Certificates are given by $f$ restricted to a small number of random elements of~$G$.

\section{Tools for CombEcon}
\label{section:tools}

In this section, we introduce the tools we need for proving CombEcon results. In Section~\ref{section:tools-SNC-spherepacking}, we introduce the strong negative correlation bound for bounding the number of sets that have small pairwise intersection. In Section~\ref{section:comparison-SNC-Johnson}, we compare this bound to the classical Johnson Bound. In Section~\ref{section:tools-large-epsilon}, we use this bound to show that we may assume that $\epsilon$ is small in our CombEcon proofs. In Sections~\ref{sec:bucket-splitting} and~\ref{section:generation-graphs}, we start carrying out the strategy outlined in Section~\ref{section:strategy-for-CombEcon} for proving CombEcon results by the method of bucket splitting.

\subsection{Strong negative correlation and sphere packing}
\label{section:tools-SNC-spherepacking}

We define strongly negatively correlated families of subsets and give a simple proof for a bound on their sizes. We apply this bound via a sphere-packing argument to divide lists into a small number of ``buckets.'' Another consequence is that we may without loss of generality assume that $\eps < \sqrt{2\LGH}$, where we again let $\Lambda = \Lambda_{G, H}$.

\begin{definition}[Strong negative correlation]
Let $0 < \rho \leq 1$ and $\tau > 0$.	Let $X$ be a finite set and let $S_1, \ldots, S_r \subseteq X$. We say that $S_1, \ldots, S_r$ are \textbf{$(\rho, \tau)$-strongly negatively correlated} if 
	\begin{enumerate}[(1)]
	\item  $\mu(S_i)  \geq \rho$ for all $i$, and
	\item $\mu(S_i \cap S_j) \leq \rho^2 - \tau$ for all $i \neq j$. 
	\end{enumerate}
\end{definition}

\begin{lemma}[Strong negative correlation bound]
	\label{lemma:strong-negative-correlation}
	Let $0 < \rho < 1$ and $\tau >0$. Let $X$ be a finite set and let $S_1, \ldots, S_r \subseteq X$ be $(\rho, \tau)$-strongly negatively correlated  subsets. Then, $r \leq \frac{1}{4\tau} + 1$. 
\end{lemma}

\begin{proof}
	Choose $x$ uniformly from $X$. For $1 \leq i \leq r$, let $Z_i(x) = \chi[x \in S_i]$ be the indicator random variable for the event that $x \in S_i$. 
	Notice that $\Var(Z_i) = \mu(S_i) ( 1- \mu(S_i) ) \leq \frac{1}{4}$. 
	
	For $i \neq j$, 
	\begin{align*}
	\Cov(Z_i, Z_j) = \EE[Z_i Z_j] - \EE[Z_i] \EE[Z_j] \leq (\rho^2 - \tau) - \rho^2 = - \tau.
	\end{align*}
	So,
	\begin{align*}
	0 \leq \Var\left(\sum_i Z_i\right) 
	= \sum_i \Var(Z_i) + \sum_{i \neq j} \Cov(Z_i, Z_j) 
	\leq \frac{r}{4} + r(r-1) (-\tau).
	\end{align*}
Solving for $r$ gives the bound as claimed. 
\end{proof}

When applied to equalizers, Lemma~\ref{lemma:strong-negative-correlation} will bound the size of a set of homomorphisms, each with high agreement with the received word but also with low pairwise agreement. This will serve as our base tool to split $\highagr(\aHom(G,H), f, \Lambda+\eps)$ into buckets.

\begin{lemma}[Sphere packing bound]
	\label{lemma:sphere-packing}
	
	Let $G$ be a finite group, $H$ a group, and $\eps > 0$. Let $f \from G \to H$ be a received word. Let $\Psi \subseteq \highagr(\aHom(G, H), f, \Lambda + \eps)$ be a subset of the list. Suppose that $\Psi$ is maximal under the condition that its members have low pairwise agreement, specifically, $\agr(\psi_1, \psi_2) \leq \Lambda^2$ for all distinct $\psi_1, \psi_2 \in \Psi$. Then, the size of $\Psi$ is bounded by 	
	\begin{equation}
	\abs{\Psi} \leq \frac{1}{4 (\Lambda + \eps) \eps} + 1.
	\end{equation}
\end{lemma}
Notice that the result also holds with $\Hom$ in place of $\aHom$, as $\Hom \subseteq \aHom$. 

\begin{proof}
	The sets $\Eq(f, \psi)$ for $\psi \in \Psi$ are $(\Lambda + \eps, (\Lambda + \eps) \eps)$-strongly negatively correlated, so the result follows by Lemma~\ref{lemma:strong-negative-correlation}.
\end{proof}

\begin{remark}
While Lemma~\ref{lemma:sphere-packing} applies to all groups, it is an existential result. We cannot algorithmically find the homomorphisms chosen in $\Psi$. So, although we use this lemma in proving combinatorial results, we cannot directly translate those proofs into algorithms.

The way we overcome this varies by setting. The list-decoder for $\{\text{abelian} \to \text{arbitrary}\}$ is indifferent to how CombEcon is proved; the AlgEcon proof relies on already having proved that abelian groups are universally CombEcon, but the algorithm is is unconnected to the method of proof. In other cases, including $\{\text{alternating} \to \text{arbitrary}\}$, we overcome this difficulty by using ``shallow random generation'' (see Section~\ref{section:SRG}).

\end{remark}

\subsection{Comparison between strong negative correlation and the Johnson Bound}
\label{section:comparison-SNC-Johnson}

The sphere packing bound (Lemma~\ref{lemma:sphere-packing}) can be rephrased as a statement about codes, and is essentially equivalent to a version of the classical Johnson bound.

	Consider a code $\mathcal{C} \subseteq \Sigma^n$ of length $n$ over an alphabet of size $\abs{\Sigma} = q$. Suppose the maximum agreement at most $\Lambda$, so that the distance of the code is $1 - \Lambda$; let $d = (1 - \Lambda) n$. Suppose that $0 \in \Sigma$, and that in every codeword in $\mathcal{C}$, at exactly $\rho n$ of the $n$ symbols are $0$; one says that a codeword has \emph{weight} $w = (1 - \rho) n$. Let $A_q(n, d, w)$ be the maximum size of such a code $\mathcal{C}$. 
	
	The Restricted Johnson Bound (see \cite[Section~2.3.1]{huffman}) says that 
	\begin{align*}
		A_q(n, d, w) \leq \frac{n d (q - 1)}{q w^2 - 2 (q - 1) n w + n d (q - 1)},
	\end{align*}
	provided that the denominator is greater than $0$.
	
	Suppose that instead of requiring each codeword to have exactly $\rho n$ zeros, we instead require it to have at least $\rho n$ zeros (that is, weight at most $(1 - \rho) n$). Let $A'_q(n, d, w)$ be the maximum size of such a code $\mathcal{C}$. The Restricted Johnson Bound still holds with $A'$ in place of $A$.

	Equivalently,
	\begin{align*}
		A'_q(n, (1 - \Lambda) n, (1 - \rho) n) \leq \frac{ 1 - \Lambda }{\rho^2 - \Lambda + \frac{(\rho - 1)^2}{q - 1}},
	\end{align*}
	again provided that the denominator is greater than $0$.
	
	So, for all $q \geq 2$, for all $0 < \rho \leq 1$, and for all $0 \leq \Lambda < \rho^2$,
	\begin{align*}
		A'_q(n, (1 - \Lambda) n, (1 - \rho) n) \leq \frac{ 1 - \Lambda }{\rho^2 - \Lambda}.
	\end{align*}
	
	The condition that each codeword has at least $\rho n$ zeros can be rephrased as a condition that each codeword has at agreement at least $\rho$ with the all-zero word. The all-zero word can be replaced with any word, and the bound still holds. Thus, this bound can be rephrased as saying that for any code $\mathcal{C}$ with maximum agreement $\Lambda$, and any $0 < \rho \leq 1$ such that $\Lambda < \rho^2$, we have that
	\begin{align*}
		\ell(\mathcal{C}, \rho) \leq \frac{1 - \Lambda}{\rho^2 - \Lambda}.
	\end{align*}

We compare this to the following generalization of the sphere packing bound, which is proved the same way, using the strong negative correlation bound (Lemma~\ref{lemma:strong-negative-correlation}).

\begin{lemma}
	Let $\cC$ be a code with maximum agreement $\Lambda$. Let $\sqrt{\Lambda} < \rho \leq 1$. Then, $\ell(\cC, \rho) \leq \frac{1}{4(\rho^2 - \Lambda)} + 1$.
\end{lemma}

\begin{proof}
	Consider any word $f$. To bound the size of $\cL = \cL(\cC, f, \rho)$, we note that the sets $\Eq(f, \phi)$ for $\phi \in \cL$ are $(\rho, \rho^2 - \Lambda)$-strongly negatively correlated, and we apply the strong negative correlation bound (Lemma~\ref{lemma:strong-negative-correlation}).
\end{proof}

When $\rho > \frac{1}{2}$, this can be further improved to $\ell(\cC, \rho) \leq \frac{\rho - \rho^2}{\rho^2 - \Lambda} + 1$, by being more careful in the proof of the strong negative correlation bound.

\subsection{Large $\eps$}
\label{section:tools-large-epsilon}

As a first consequence of strong negative correlation (Lemma~\ref{lemma:strong-negative-correlation}), we will see that we may assume $\eps$ is ``small'' --- specifically, $\eps < \sqrt{2\LGH}$ --- in our CombEcon proofs. So, to show CombEcon it suffices to show a list-size bound of $\poly(1/\eps, 1/\Lambda)$ rather than $\poly(1/\eps)$. 

\begin{lemma}[Large $\eps$ lemma] 
	\label{lemma:small-epsilon-vs-lambda}
	Let $G$ be a finite group and $H$ a group. Suppose that $\LGH \leq \frac12 \eps^2$. Then, $\ell(\aHom(G, H), \Lambda + \eps) \leq \frac{1}{2 \eps^2} + 1$. In particular, $\aHom(G,H)$ is combinatorially $(\LGH + \eps, \poly(1/\eps))$-list-decodable.  
\end{lemma}

\begin{proof} 
	The sets $\Eq(f, \phi)$ for $\phi \in \highagr(\aHom(G,H), \Lambda + \eps)$ are $(\Lambda + \eps, \eps^2/2)$-strongly negatively correlated, so the result follows from Lemma~\ref{lemma:strong-negative-correlation}.

\end{proof}

\begin{corollary}

	Let $G$ be a finite group and $H$ a group. If $\ell(\aHom(G, H), \Lambda + \eps) \leq \poly(1/\eps, 1/\Lambda)$, then $\aHom(G, H)$ is CombEcon.
\end{corollary}
	The result also holds with $\Hom$ in place of $\aHom$.

\subsection{Bucket splitting}
\label{sec:bucket-splitting}

In Section~\ref{section:strategy-for-CombEcon}, we outlined a strategy for showing that certain classes of homomorphism codes are CombEcon. In this section, we begin carrying out this strategy, and introduce our tools.

Let $G$ be a finite group, $H$ a group (finite or infinite), $\Lambda = \Lambda_{G, H}$ the maximum agreement, $f \from G \to H$ a received word, $\eps > 0$ a real, and $\highagr = \highagr(\aHom(G, H), f, \Lambda + \eps)$ the list.

Our goal is to bound the size $\abs{\highagr}$ of the list.

To do this, we split $\highagr(\aHom(G,H), f, \Lambda+\eps)$ into sets called \emph{buckets}, which we label with elements of the set $\Psi \subseteq \highagr$ introduced in Lemma~\ref{lemma:sphere-packing}. Each bucket, denoted $\buckl\psi$, will contain the sphere centered at the homomorphism $\psi \in \Psi$ with radius 
$(1 - \Lambda^2)$.

\begin{definition}[Bucket $\buckl\psi$]
Let $G$ be a finite group, $H$ a group, $\psi \in \aHom(G,H)$, $f: G \to H$, and $\eps > 0$.  

The \defn{bucket} $\buckl\psi$ is 

\begin{align*}
	\buckl\psi := \{ \varphi \in \highagr \mid \agr(\varphi, \psi) > \Lambda^2 \}.
\end{align*}
\end{definition}

\begin{lemma}[Bucket-splitting lemma]
	\label{lemma:bucket-splitting}
	Let $G$ be a finite group, $H$ a group, $f: G \to H$, $\psi \in \aHom(G,H)$, and $\eps > 0$. Then, there exists a subset $\Psi \subseteq \highagr(\aHom(G, H), f, \Lambda + \eps)$, with size $\abs{\Psi} \leq \frac{1}{4 (\Lambda + \eps) \eps} + 1$, such that
	\begin{equation*}
	\highagr(\aHom(G,H), f, \Lambda + \eps) \subseteq \bigcup_{\psi \in \Psi} \buckl\psi.
	\end{equation*}	
\end{lemma}
\begin{proof}
	Let $\Psi$ be as in Lemma~\ref{lemma:sphere-packing}, that is, a subset of $\highagr(\aHom(G, H), f, \Lambda + \eps)$ that is maximal under the conditions that distinct $\psi_1, \psi_2 \in \Psi$ have small agreement $\agr(\psi_1, \psi_2) \leq \Lambda^2$. By the maximality of $\Psi$, every $\varphi \in \highagr(\aHom(G,H), f, \Lambda+\eps)$ has high agreement $\agr(\varphi, \psi) > \Lambda^2$ with some homomorphism $\psi \in \Psi$.
\end{proof}

We bound the size of each bucket $\buckl\psi$ by further subdividing it into smaller sets, which we refer to as \emph{sub-buckets}. We then bound both the number of sub-buckets per bucket, and the size of each sub-bucket. The method of subdivision will differ depending on the type of group.

For abelian groups, we label each sub-bucket with an abelian subgroup $M$ of the codomain, $H$; the sub-bucket $\abkpm$ consists of all homomorphisms $\phi \in \buckl\psi$ whose image is contained in $M$.

\begin{definition}[Sub-bucket $\abkpm$ for abelian groups]
	Suppose $G$ is an abelian group. Let $\psi \in \aHom(G, H)$ and $M \leq H$. The \defn{sub-bucket} $\abkpm$ is
	\begin{align*}
		\abkpm = \{\phi \in \buckl\psi \mid \phi(G) \leq M\}.
	\end{align*}
\end{definition}

For alternating groups,

we label each sub-bucket with a subgroup $K$ of the domain, $G$; the sub-bucket $\sbk_{\psi, K}$ consists of all homomorphisms $\phi \in \buckl\psi$ whose equalizer with $\psi$ contains $K$.

\begin{definition}[Sub-bucket $\sbk_{\psi, K}$ for alternating groups]
	\label{defn:subbucket}

	Suppose $G$ is an alternating group. Let $\psi \in \highagr$ and $K \leq H$.	The \defn{sub-bucket} $\sbk_{\psi, K}$ is 
	\begin{align*}
		\sbk_{\psi, K} = \{ \varphi \in \highagr \mid K \leq \Eq(\varphi, \psi) \}.
	\end{align*}
\end{definition}

To bound the size of the buckets, we in both cases use the fact that elements of $\buckl\psi$ agree with $\psi$ on a subgroup of large density. Moreover, in abelian groups and alternating groups, subgroups with large density have small depth. We leverage this in the next lemma. The lemma will help us bound the number of sub-buckets $\abkpm$ per bucket in the abelian case (but we will need to bound the size of each sub-bucket another way). And, it will help us bound the size of each sub-bucket $\sbk_{\psi, K}$ in the alternating case (but we will need to bound the number of sub-buckets per bucket another way).

The set $S$ in the lemma should be thought of as $\Eq(f, \phi)$ for some homomorphism $\phi$ of interest. 

\begin{lemma}
	\label{lemma:random-el-depth}
	Let $0 \leq \lambda < 1$. 
	Let $G$ be a finite group, $K \leq G$ a subgroup, and $S \subseteq G$ a subset. Suppose that $\mu_G(S) > \lambda$. Let $\eps = \mu(S) - \lambda$ and $d = \depth_G(K)$. Then, 
	\[
	\Pr_{s_1, \ldots, s_d \in S}[\mu(\gen{K, s_1, \dots, s_d}) > \lambda] \geq \paren*{ \frac{\eps}{\lambda+\eps} }^d.
	\]
	
	It follows that 
	\[
	\Pr_{g_1, \ldots, g_d \in G} [ g_1, \dots, g_d \in S
	\text{ and }
	\mu(\langle K , g_1, \ldots, g_d \rangle) > \lambda  ] \geq \eps^d.
	\]
\end{lemma}
This is proved by repeated application of Bayes' rule. 
\begin{proof}
	
	Pick $s_1, s_2, s_3, \dotsc$ independently and uniformly from $S$.

	We proceed by induction on $\ind{G}{K}$. 
	
	Suppose $\mu(K) > \lambda$. Then, $\Pr[\mu(\gen{K, s_1, \dots, s_d}) > \lambda] = 1$.
	
	Suppose $\mu(K) \leq \lambda$. Then, with probability $\displaystyle{\frac{\mu(S \smallsetminus K)}{\mu(S)} \geq \epsl}$, we have that $s_1 \notin K$, so $\gen{K, s_1} > K$, and $\depth_G \gen{K, s_1} \leq d - 1$. Then, by the induction hypothesis,
	\begin{align}
		\Pr[\mu(\gen{K, s_1, \dots, s_d})] 
		&\geq \Pr[\mu(\gen{K, s_1, \dots, s_d}) \mid s_1 \notin K] \cdot \Pr[s_1 \notin K] \\
		&\geq \paren*{ \epsl }^{d - 1} \cdot \paren*{ \epsl }.
	\end{align}
	This completes the inductive step.
\end{proof}

\begin{remark}
	With a little more care in the proof (separating out the case $\lambda/2 < \mu(K) \leq \lambda$), one can prove the stronger conclusion	
	\begin{align}
		\Pr[\mu(\gen{K, s_1, \dots, s_d}) > \lambda] \geq \paren*{ \halfl }^{d - 1} \cdot \epsl.
	\end{align}
\end{remark}

\subsection{Bipartite generation-graphs}
\label{section:generation-graphs}

We retain the notation $G, H, f, \eps, \highagr$ from the previous subsection. Let $\Psi \subseteq \aHom(G, H)$ be as defined in Lemma~\ref{lemma:sphere-packing}. 

In the case that $G$ is abelian or alternating, we divided the list $\highagr$ into buckets $\buckl\psi$. To bound the size $\abs{\highagr}$ of the list, we just need to bound the size $\abs{\buckl\psi}$ of each bucket. If we fix $\psi \in \Psi$ and consider some homomorphism $\phi \in \buckl\psi$ that is hidden to us, we can almost recover $\phi$ with some decent probability if we have access to $\psi$ and a few random elements of $\Eq(f, \phi)$. By ``almost recover,'' we mean there is a small list of homomorphisms which contains $\phi$.

To make this more precise, we define the following bipartite graph. 

\begin{definition}
	Let $d$ be a positive integer. We define a bipartite graph $\gph_{\psi, d}$. The left vertex set is $V = G^d$ and the right vertex set is $W = \buckl\psi$. The vertices $(g_1, \dots, g_d) \in V$ and $\phi \in W$ are adjacent if $g_1, \dots, g_d \in \Eq(f, \phi)$ and $\mu(\gen{\Eq(\psi, \phi), g_1, \dots, g_d}) > \Lambda$.
\end{definition}

We will bound $\buckl\psi$ by applying part~(a) of
Lemma~\ref{lemma:mean-comb-bipartite} (double counting)
to $\gph_{\psi, d}$. So, we would like to bound the degree of a left vertex from above, and the degree of a right vertex from below.

The next lemma bounds below the degree of a right vertex in certain cases, and will be useful when $G$ is abelian or alternating.

\begin{lemma}
	\label{lem:right-degree-big-flex}
	Let $d$ be a positive integer, and $\phi \in \buckl\psi$ be a right vertex of $\gph_{\psi, d}$ such that $\depth_G \Eq(\psi, \phi) \leq d$. Then $\phi$ is adjacent to at least an $\eps^d$ fraction of the left vertices.
\end{lemma}

\begin{proof}

	By Lemma~\ref{lemma:random-el-depth} with $\lambda = \Lambda$, and $S = \Eq(f, \phi)$, and $K = \Eq(\psi, \phi)$, we have that
	\[
		\Pr_{g_1, \ldots, g_d \in G} [ g_1, \dots, g_d \in \Eq(f, \psi)
		\text{ and }
		\mu(\langle \Eq(\psi, \phi) , g_1, \ldots, g_d \rangle) > \lambda  ] \geq \eps^d.
	\]
	
	So, $\phi$ is adjacent to at least $\eps^d$ fraction of the tuples $(g_1, \dots, g_d)$.
\end{proof}

We would also like to bound the degree of a left vertex from above. For abelian groups, we will do this in Section~\ref{section:abelian-comb}, in Corollary~\ref{cor:abelian-left-vertex-degree}. For alternating groups, we do this by splitting the graph $\gph_{\psi, d}$ into subgraphs based on the 

sub-buckets $\sbk_{\psi, K}$ defined in the previous subsection. We will use the following lemma.

\begin{lemma}
	\label{lem:unique-homomorphism-K-g}
	Let $K \leq G$ be a subgroup, $\psi \in \Hom(G, H)$ a homomorphism, $d$ a nonnegative integer, and $g_1, \dots, g_d \in G$. If $\mu(\gen{K, g_1, \dots, g_d}) > \Lambda$, then there is at most one homomorphism $\phi \in \Hom(G, H)$ such that $K \leq \Eq(\psi, \phi)$ and $g_1, \dots, g_d \in \Eq(f, \phi)$.
\end{lemma}
\begin{proof}

	If there were two such homomorphisms $\phi_1$ and $\phi_2$, we would have that $K \leq \Eq(\psi, \phi_1) \cap \Eq(\psi, \phi_2) \leq \Eq(\phi_1, \phi_2)$, and $g_1, \dots, g_d \in \Eq(f, \phi_1) \cap \Eq(f, \phi_2) \leq \Eq(\phi_1, \phi_2)$ so $\agr(\phi_1, \phi_2) \geq \mu(\gen{K, g_1, \dots, g_d}) > \Lambda$, so $\phi_1 = \phi_2$.
\end{proof}

As an application of Lemma~\ref{lem:unique-homomorphism-K-g}, we can bound the size of the sub-buckets $\sbk_{\psi, K}$ for $K$ of low depth, which will be useful in proving that alternating groups are universally CombEcon. We do this in the next corollary. This lemma will also be used when $G$ is a shallow random generation group (See Section~\ref{section:SRG}).

\begin{corollary}[Sub-bucket bound for low-depth label subgroups] 
	\label{cor:bucket-bound}
	Let $G$ be a finite group, $H$ a group, $K \leq G$ a subgroup, $f \from G \to H$, and $\eps > 0$.

	Then, 
	$$\abs{\sbk_{\psi, K}} \leq 1/\eps^{\depth_G(K)}.$$
\end{corollary}

\begin{proof}
	Let $d = \depth_G(K)$. Consider the induced subgraph of $\gph_{\psi, d}$ with the same left vertex set, and with right vertex set $\sbk_{\psi, K}$. By Lemma~\ref{lem:unique-homomorphism-K-g}, each left vertex has degree at most $1$. By Lemma~\ref{lem:right-degree-big-flex}, each right vertex is adjacent to at least an $\eps^d$ fraction of the left vertices. Apply part~(a) of
Lemma~\ref{lemma:mean-comb-bipartite} (double counting).
\end{proof}

We will use this corollary in Section~\ref{section:alt-tools-2}.

\section{Homomorphism Codes with finite abelian domain and arbitrary codomain}
\label{section:abelian}

In this section we show that finite abelian groups are universally combinatorially and algorithmically economically list-decodable. The key technical result is Theorem~\ref{thm:codomain-few-abelian-subgroups}, which says that there are a small number of abelian subgroups of the codomain such that every homomorphism in the list maps into one of these subgroups.

In Section~\ref{sec:abelian-enlargements} we introduce a tool called an \emph{abelian enlargement}.
Using this tool, in Section~\ref{section:abelian-comb} we prove Theorem~\ref{thm:codomain-few-abelian-subgroups} (the key result mentioned in the previous paragraph) and infer that abelian groups are universally CombEcon. 
In Section~\ref{sec:abelian-to-anything-algorithm} we 
adapt the algorithm of~\cite{DGKS08, GS14}, to
give an algorithm to locally list-decode these codes.

In Section~\ref{sec:Lambda-for-G-or-H-solvable} we describe $\Lambda_{G, H}$ for these and a few other codes, slightly generalizing a result of Guo~\cite{Guo15}. Our proof of CombEcon only uses that taking a subset of the codomain does not increase $\Lambda$. 

We remark that these codes usually cannot be list-decoded beyond radius 
$1 - (\Lambda_{G, H} + \eps)$ (see Remark~\ref{rmk:blowup}).

\subsection{Abelian enlargements}
\label{sec:abelian-enlargements}

Throughout this section, let $G$ be a finite abelian group, and $H$ a group (finite or infinite).

To prove that abelian groups are universally CombEcon, we will follow the outline given in Section~\ref{sec:bucket-splitting}. That is, we divide the list $\highagr = \highagr(\aHom(G, H), f, \Lambda + \eps$ into buckets $\buckl\psi$, and then subdivide each bucket into sub-buckets $\abkpm$ for a small number of abelian subgroups $M \leq H$.  (Actually, we will use $\Hom$ in place of $\aHom$, but it does not matter by Lemma~\ref{lemma:mean-central-aHom-relationship}.)

To select the subgroups $M$, we introduce an operation that we call an \emph{abelian enlargement}. In Section~\ref{section:abelian-comb}, we will use this operation in our proof that abelian groups are universally CombEcon. For a subset $T$ and a finite abelian subgroup $B$ of a group $H$, the abelian $B$-enlargement of $T$ in $H$ is the group generated by $T$ along with every element of $B$ that commutes with $T$. 

If 

$\phi, \psi \in \Hom(G, H)$ are homomorphisms, then the abelian $\psi(G)$-enlargement of $\phi(G)$ will certainly still include $\phi(G)$. But also, if $\phi$ and $\psi$ have large agreement, then most of $\phi(G)$ is contained in $\psi(G)$, so even if we take a single element of $\phi(G)$, then when we enlarge it by $\psi(G)$ it is likely that the result will already contain all of $\phi(G)$.

This will help us bound the number of subgroups $M \leq H$ we need, and thus the number of sub-buckets.

For any $\phi$ in the bucket $\buckl\psi$, 

its image is likely to be contained 

in the enlargement by $\psi(G)$ of a single random element of $f(G)$.

This is the method by which we choose our subgroups~$M$.

\begin{definition}
	For $H$ a group, $B \leq H$ a subgroup, and $T \subseteq H$ a subset, define the \defn{abelian $B$-enlargement of $T$ in $H$} to be
	\[
	\sig_B(T) = \gen{T, C_H(T) \cap B},
	\]
	where $C_H(T)$ denotes the centralizer of $T$ in $H$.

	For $h \in H$, we may write $\sig_B(h)$ in place of $\sig_B(\{h\})$.

\end{definition}

We note that if $\gen{T_1} = \gen{T_2}$, then $\sig_B(T_1) = \sig_B(T_2)$.

\begin{lemma}
	\label{lem:sig-finite-abelian}
	For $B \leq H$ a finite abelian subgroup and $T \subseteq H$ a set such that $\gen{T}$ is a finite abelian group, $\sig_B(T)$ is a finite abelian group.
\end{lemma}

In fact, we will only be concerned with the case that $B \leq H$ is a finite abelian subgroup, and $\gen{T}$ is abelian.

\begin{proof}
	Every element of $T$ commutes with every element of $C_H(T)$ by the definition of the centralizer. So, $\sig_B(T)$ is the direct product of $\gen{T}$ and $C_H(T) \cap B$. The group $\gen{T}$ is finite abelian by assumption, and $C_H(T) \cap B$ is finite abelian because it is a subgroup of $B$. So, $\sig_B(T)$ is a finite abelian group.

\end{proof}

\begin{lemma}
	\label{lem:sig-expand}
	For $B$ and $T$ as above, and $U \subseteq \sig_B(T)$, we have that $\sig_B(T) = \sig_B(T \cup U)$.
\end{lemma}

\begin{proof}
	First, we show that $\sig_B(T) \leq \sig_B(T \cup U)$. Since $\sig_B(T)$ is abelian, we have that 
	\[
		C_H(T) \cap B \leq \sig_B(T) \leq C_H(\sig_B(T)) \leq C_H(U).
	\]
	So, 
	\[
		C_H(T) \cap B \leq C_H(T) \cap C_H(U) \cap B = C_H(T \cup U) \cap B \leq \sig_B(T \cup U).
	\]
	Since also $T \subseteq \sig_B(T \cup U)$, we have that $\sig_B(T) \leq \sig_B(T \cup U)$.

	Next, we show that $\sig_B(T \cup U) \leq \sig_B(T)$. We have that $T \subseteq \sig_B(T)$, that $U \subseteq \sig_B(T)$, and $C_H(T \cup U) \cap B \leq C_H(T) \cap B \leq \sig_B(T)$. So, $\sig_B(T \cup U) = \gen{T \cup U, C_H(T \cup U) \cap B} \leq \sig_B(T)$.
\end{proof}

\begin{proposition}
	\label{prop:sig-generate}

	Let $\phi, \psi \in \Hom(G, H)$ and $A \subseteq G$ such that $\gen{A, \Eq(\psi, \phi), \ker \phi} = G$.
	
	Then $\sig_{\psi(G)}(\phi(A)) = \sig_{\psi(G)}(\phi(G))$.
\end{proposition}

\begin{proof}
	Since $G$ is finite abelian, so are $\phi(G)$ and $\psi(G)$. Let $B = \psi(G)$. Let $T = \phi(A)$. Let $U = \phi(\Eq(\psi, \phi))$. Since $T, U \subseteq \phi(G)$, and $\phi(G)$ is abelian, $U \leq C_H(T)$. And, since $U = \psi(\Eq(\psi, \phi))$, we have that $U \leq \psi(T) = B$. Thus, $U \leq C_H(T) \cap B \leq \sig_B(T)$. 
	
	Also, $\gen{T \cup U} = \gen{T, U, 1} = \gen{\phi(A), \phi(\Eq(\psi, \phi)), \phi(\ker \phi)} = \phi(\gen{A, \Eq(\psi, \phi), \ker \phi}) = \phi(G)$.
	
	Therefore, by Lemma~\ref{lem:sig-expand},
	\begin{equation}
	\sig_{\psi(G)}(\phi(A)) = \sig_B(T) = \sig_B(T \cup U) = \sig_B(\gen{T \cup U}) = \sig_{\psi(G)}(\phi(G)).
	\end{equation}
\end{proof}

\begin{corollary}
	\label{cor:enlargement-contains-image}
	Let 

	$\phi$, $\psi$, and $A$ be as above. Then $\phi(G) \leq \sig_{\psi(G)}(\phi(A))$.
\end{corollary}

\subsection{Combinatorial list-decodability, finite abelian to anything}
\label{section:abelian-comb}

In this section, we establish that finite abelian groups are universally CombEcon.

Throughout this section, let $G$ be a finite abelian group, and $H$ an arbitrary group (finite or infinite). Let $f \from G \to H$ be a received word. Let $\eps > 0$. Let $\highagr = \highagr(\Hom(G, H), f, \Lambda + \eps)$ be the list (note that in this section we deal with the code of homomorphisms, rather than affine homomorphisms; however, we can convert between the two; see Section~\ref{sec:Hom-versus-aHom}). The list $\highagr$ is divided into buckets $\buckl\psi$ for $\psi \in \Psi$, where $\Psi$ is as in Lemma~\ref{lemma:sphere-packing}.

We will see that there is a small set of abelian subgroups $M \leq H$ such that every $\phi \in \highagr$ has its image in some $M$.

Dinur, Grigorescu, Kopparty, and Sudan \cite{DGKS08} proved that $\aHom(G, H)$ is CombEcon (and in fact, AlgEcon) for all finite abelian groups $G$ and $H$. 

\begin{theorem}[DGKS 2008]
	\label{thm:DGKS-CombEcon}
	The class $\abelgroups \times \abelgroups$ of pairs of abelian groups is CombEcon.
\end{theorem}

The following theorem, combined with the DGKS result, lets us conclude that $\Hom(G, H)$ (and thus $\aHom(G, H)$) is CombEcon.

\begin{theorem}
	\label{thm:codomain-few-abelian-subgroups}

	There exists a set $\abelsubs$ of finite abelian subgroups of $H$ with $\displaystyle{\abs{\abelsubs} \leq \frac{1}{4 (\Lambda + \eps) \eps^2} + \frac{1}{\eps}}$ such that for all $\phi \in \highagr$, 

	there is $M \in \abelsubs$ such that $\phi(G) \leq M$.
\end{theorem}

\begin{corollary}
	\label{cor:abelian-to-anything-combecon}

	Finite abelian groups are universally CombEcon. Specifically, let $C$ be a constant such that 

	$\ell(\aHom(G, H), \Lambda + \eps) \leq (\frac{1}{\eps})^C$ for $G, H$ finite abelian groups. Then 

	$\ell(\aHom(G, H), \Lambda + \eps) \leq O((\frac{1}{\eps})^{C + 4})$ for $G$ a finite abelian group and $H$ and arbitrary group.

\end{corollary}

By \cite{GS14, BGSW}, the constant $C$ is approximately $105$.

\begin{proof}[Proof of Corollary~\ref{cor:abelian-to-anything-combecon}]

	Let $\abelsubs$ be the collection of subgroups of $H$ guaranteed by Theorem~\ref{thm:codomain-few-abelian-subgroups}. Then, 

	$\highagr \subseteq \bigcup_{M \in \abelsubs} \highagr(\Hom(G, M), f, \Lambda + \eps)$ (on the right hand side we let $f$ be redefined arbitrarily at points in its domain that do not map to $M$). So, 
	\[
		\abs{\highagr} \leq \sum_{M \in \abelsubs} \ell(\Hom(G, M), \Lambda + \eps) \leq \left(\frac{1}{4 (\Lambda + \eps) \eps^2} + \frac{1}{\eps}\right) \left(\frac{1}{\eps} \right)^C.
	\]
	We then apply Lemma~\ref{lemma:mean-central-aHom-relationship}.
\end{proof}

In the remainder of this subsection, we prove Theorem~\ref{thm:codomain-few-abelian-subgroups}.

Recall our strategy from Sections~\ref{sec:bucket-splitting} and~\ref{section:generation-graphs} of dividing the list $\highagr$ into buckets $\buckl\psi$, and bounding the size of each bucket by considering the graph $\gph_{\psi, d}$. We will now carry out this strategy. We will use $d = 1$. We next bound the degree of each right vertex.

\begin{lemma}
	\label{lem:abelian-right-vertices-degree}
	Let $\psi \in \Psi$. Each right vertex of $\gph_{\psi, 1}$ is adjacent to at least an $\eps$ fraction of the left vertices.
\end{lemma}
\begin{proof}
	We know that $1/\Lambda$ is the least integer that divides $\ind{G}{N}$, where $N$ is the $(G, H)$-irrelevant kernel. 

	For all $\phi \in \buckl\psi$, we have that $\Eq(\psi, \phi)$ contains the irrelevant kernel and has density greater than $\Lambda^2$, so has depth at most $1$.
	We apply Lemma~\ref{lem:right-degree-big-flex}.
\end{proof}

The next lemma tells us about the connected components of $\gph_{\psi, 1}$. It also helps us bound the degree of a left vertex via the DGKS result.

\begin{lemma}
	\label{lem:adjacent-same-enlargement}
	Let $\psi \in \Psi$. Let $g_1 \in G$ be a left vertex of $\gph_{\psi, 1}$, and $\phi \in \buckl\psi$ a right vertex. If $g_1$ is adjacent to $\phi$ then $\sig_{\psi(G)}(f(g_1)) = \sig_{\psi(G)}(\phi(G))$.
\end{lemma}
\begin{proof}
	We have $g_1 \in \Eq(f, \phi)$ and $\gen{\Eq(\psi, \phi), g_1}$ has density greater than $\Lambda$, so is equal to $G$. So, by Proposition~\ref{prop:sig-generate}, $\sig_{\psi(G)}(f(g_1)) = \sig_{\psi(G)}(\phi(g_1)) = \sig_{\psi(G)}(\phi(G))$.
\end{proof}

Lemma~\ref{lem:adjacent-same-enlargement} allows us to associate an abelian subgroup of~$H$ to each connected component of~$\gph_{\psi, 1}$.

\begin{corollary}
	\label{cor:abelian-left-vertex-degree}
	Each left vertex of of $\gph_{\psi, 1}$ has degree at most $\paren{\frac{1}{\eps}}^C$.
\end{corollary}
\begin{proof}
	All the neighbors of $g_1$ are elements of $\highagr(\Hom(G, M), f, \Lambda + \eps)$ for $M = \sig_{\psi(G)}(f(g_1))$. We apply the DGKS result.
\end{proof}

\begin{remark}
	We can now prove Corollary~\ref{cor:abelian-to-anything-combecon}, bypassing Theorem~\ref{thm:codomain-few-abelian-subgroups}, 
by applying part~(a) of
Lemma~\ref{lemma:mean-comb-bipartite} (double counting).
We can use Lemma~\ref{lem:abelian-right-vertices-degree} to bound the degree of a right vertex, and 

	Corollary~\ref{cor:abelian-left-vertex-degree}
	to bound the degree of a left vertex.
\end{remark}

We need one more lemma before we prove Theorem~\ref{thm:codomain-few-abelian-subgroups}.

\begin{lemma}
	\label{lem:codomain-psi-few-abelian-subgroups}

	Let $\psi \in \Psi$.
	There is a set $\abelsubs_\psi$ of finite abelian subgroups of $H$ with $\abs{\abelsubs_\psi} \leq \frac{1}{\eps}$ such that for all $\phi \in \buckl\psi$, there is $M \in \abelsubs_\psi$ for which $\phi(G) \leq M$.
\end{lemma}
\begin{proof}

	Let $\abelsubs_\psi = \{\sig_{\psi(G)}(\phi(G)) \mid \phi \in \buckl\psi\}$, which is a set of finite abelian subgroups of $H$ by Lemma~\ref{lem:sig-finite-abelian}. By Lemma~\ref{lem:adjacent-same-enlargement}, we can associate an abelian subgroup to each connected component of $\gph_{\psi, 1}$. Each element of $\abelsubs_\psi$ is associated to at least one connected component that contains a right vertex. By Lemma~\ref{lem:abelian-right-vertices-degree}, there are at most $1/\eps$ such components, so 

	$\abs{\abelsubs_\psi} \leq 1/\eps$.
\end{proof}

We are ready to prove the main theorem of this section.

\begin{proof}[Proof of Theorem~\ref{thm:codomain-few-abelian-subgroups}]

	Let $\abelsubs_\psi$ be as in Lemma~\ref{lem:codomain-psi-few-abelian-subgroups}. Let $\abelsubs = \bigcup_{\psi \in \Psi} \abelsubs_\psi$. By Lemma~\ref{lemma:sphere-packing}, $\abs{\Psi} \leq \frac{1}{4 (\Lambda + \eps) \eps} + 1$. By Lemma~\ref{lem:codomain-psi-few-abelian-subgroups}, $\abs{\abelsubs_\psi} \leq \frac{1}{\eps}$. So, $\abs{\abelsubs} \leq \frac{1}{4 (\Lambda + \eps) \eps^2} + \frac{1}{\eps}$.
	
	For $\phi \in \highagr$, we have that $\phi \in \buckl\psi$ for some $\psi \in \Psi$, and so the image of $\phi$ is in $\sig_{\psi(G)}(\phi(G)) \in \abelsubs_\psi \subseteq \abelsubs$.

\end{proof}

\subsection{Algorithm}
\label{sec:abelian-to-anything-algorithm}

For $G$ a finite abelian group given explicitly by a primary decomposition, and $H$ a group with black-box access, we can locally list-decode $\aHom(G, H)$ using essentially the same algorithm as the one by Dinur Grigorescu, Kopparty, and Sudan in Section~5 of \cite{DGKS08}. We make only slight modifications. Thus, such codes are AlgEcon.

\begin{theorem}
	Let $\fD$ be the class of pairs $(G, H)$ where $G$ is an abelian group given explicitly by an primary decomposition, and $H$ is a group with black-box access. Then there is an algorithm to locally list-decode $\fD$ in time $\poly(\log \abs{G} \cdot \frac{1}{\eps})$.
\end{theorem}

We assume black-box access to $H$. We do not assume black-box access to $G$; if only black-box access were assumed, then for $p$ a prime, it would take $p + 1$ queries to a group to determine whether the group were isomorphic to $\Z_p$ or $\Z_p^2$. Like \cite{DGKS08}, we assume that $G$ is given explicitly by an primary decomposition.

We assume that we have an algorithm determining $\Lambda_{G, H}$, although this assumption can be removed.

Next, \cite{DGKS08} reduces to the case where $H = \Z_{p^r}$. We don't make this reduction. We let $p$ be the prime such that $\Lambda = \frac{1}{p}$. Every mention of $\Z_{p^r}$ should be replaced by $H$. As in their algorithm, we take $G = G_1, \dots, G_k$, with each $G_i = \Z_{p_i^{r_i}}$. We order the $G_i$ such that $p_1 = p$. For them, the only important coordinates are the ones where $p_i = p$, but for our purposes, instances of $\Z_{p^{r_i}}$ should be replaced with $\Z_{p_i^{r_i}}$.

In the algorithm $\textsc{Extend}$ of $\cite{DGKS08}$, the statement ``If $c_1 - c_2$ is not divisible by $p$'' should be replaced with ``If $c_1 - c_2$ is not divisible by $p_i$, and if $f(y_1, c_1, s)$ and $f(y_2, c_2, s)$ commute with each other and with $\phi(e_1), \dots, \phi(e_{i - 1})$''. Here $e_j$ denotes a generator of $G_j$. The system of equations that follows should be solved under the assumption that the order of $a$ divides $p_i^{r_i}$.

We note that when solving the system of equations in \textsc{Extend}, we are working in an abelian subgroup of $H$. Actually, even this does not matter; we can solve the system of equations without assuming elements of $H$ commute.

\subsection{$\Lambda_{G, H}$ when $G$ or $H$ is solvable}
\label{sec:Lambda-for-G-or-H-solvable}

We give a combinatorial description of $\Lambda_{G, H}$ when $G$ is a finite abelian group and $H$ is an arbitrary group.

\begin{proposition}
	\label{prop:Lambda-abelian-anything}
	Let $G$ be a finite abelian group and $H$ a group. Then $\Lambda_{G, H} = 1/p$, where $p$ is the smallest prime number such that $p$ divides $\abs{G}$ and $H$ has an element of order $p$. If no such $p$ exists, then $\abs{\Hom(G, H)} = 1$ and $\Lambda_{G, H} = 0$.
\end{proposition}

This proposition is a special case of the following theorem, which describes $\Lambda_{G, H}$ when $G$ or $H$ is a solvable group. This is a slight generalization of a result of Guo \cite[Theorem 1.1]{Guo15}.

\begin{theorem}
	\label{thm:Lambda for G or H solvable}
	Let $G$ be a finite group and $H$ a group, such that at least one of $G$ or $H$ is solvable. Then $\Lambda_{G, H} = 1/p$, where $p$ is the smallest prime number such that $G$ has a normal subgroup of index $p$ and $H$ has an element of order $p$. If no such $p$ exists, then $\abs{\Hom(G, H)} = 1$ and $\Lambda_{G, H} = 0$.
\end{theorem}

We will prove Theorem~\ref{thm:Lambda for G or H solvable} in this subsection. Guo proved Theorem~\ref{thm:Lambda for G or H solvable} in the case where $H$ is finite, and either $G$ is solvable or $H$ is nilpotent.

Our proof relies on the following lemma about $\bigcap_{\phi \in \Hom(G, H)} \ker \phi$.

\begin{lemma}
	\label{lem:prime factors of G/K}
	Let $G$ be a finite group and $H$ a group. Let $K = \bigcap_{\phi \in \Hom(G, H)} \ker \phi$. Then every prime factor of $\ind{G}{K}$ is the order of an element of $H$.
\end{lemma}
\begin{proof}

	Consider any prime factor $p$ of $\ind{G}{K}$. Then there is $g \in G$ such that $gK$ has order $p$ in $G/K$. Since $g \notin K$, there is $\phi \in \Hom(G, H)$ such that $g \notin \ker \phi$. We have $g^p \in K$, so $\phi(g)^p \in \phi(K) = 1$, so $\abs{\phi(g)}$ divides $p$. Since $\phi(g) \neq 1$, we have $\abs{\phi(g)} = p$.
\end{proof}

We also use the following well-known fact and a theorem by Berkovich (see \cite{Isa08}). 
\begin{fact}
	\label{fact:solvable maximal normal subgroup prime index}
	In a solvable group, a normal subgroup is a maximal normal subgroup if and only if it has prime index.
\end{fact}

\begin{theorem}[Berkovich]
	\label{thm:Berkovich subgroup smallest index}
	Let $G$ be a finite solvable group and $K$ a proper subgroup of smallest index. Then $K \trianglelefteq G$.
	
\end{theorem}

Next, we prove Theorem~\ref{thm:Lambda for G or H solvable} in the case when $G$ is solvable. Guo \cite[Theorem 5.5]{Guo15} proved this in the case that also $H$ is finite.

Guo's proof can be modified slightly to also accommodate infinite groups. We include a compact proof for completeness.

\begin{proof}[Proof of Theorem~\ref{thm:Lambda for G or H solvable} in the case where $G$ is solvable]

	Let $K = \bigcap_{\phi \in \Hom(G, H)} \ker \phi$. If $\LGH > 0$, then there is a nontrivial homomorphism $\phi \in \Hom(G, H)$. Then $\ker \phi$ is a proper normal subgroup of $G$, so is contained in a maximal normal subgroup $M$. By Fact~\ref{fact:solvable maximal normal subgroup prime index}, $\ind{G}{M}$ is prime, and since $M \geq \ker \phi \geq K$, we have that $\ind{G}{M}$ divides $\ind{G}{K}$. So, by Lemma~\ref{lem:prime factors of G/K}, we have that $H$ has an element of order $\ind{G}{M}$. So, if $\LGH > 0$, then $p$ exists.

	Henceforth, assume $p$ exists. Let $N$ be a normal subgroup of $G$ of index $p$, and $h$ an element of $H$ of order $p$. We show that $\LGH \geq 1/p$ by exhibiting a pair of homomorphisms that achieve this agreement. Let $\phi_1\colon  G \to H$ be the trivial homomorphism. 
	Since $\abs{G/N} = \abs{\langle h \rangle} = p$ prime, there is a group isomorphism $G/N \to \langle h \rangle$. This lifts to a group homomorphism $\phi_2\colon  G \to \langle h \rangle$. Then $\Eq(\phi_1, \phi_2) = N$, so $\agr(\phi_1, \phi_2) = 1/p$. So, $\Lambda_{G, H} \geq 1/p$.
	
	We next show that $\LGH \leq 1/p$. Since $N = \ker \phi_2 \geq K$, we have that $N/K \trianglelefteq G/K$ and $\ind{G/K}{N/K} = p$. Furthermore, 
	we claim that $N/K$ is a proper normal subgroup of smallest prime index in $G/K$ --- if there were a proper normal $\hat{N}/K$ of prime index $q < p$ (with $K \leq \hat{N} \leq G$), then $\hat{N}$ would be a normal subgroup of $G$ of index $q$, and $H$ would have an element of order $q$ by Lemma~\ref{lem:prime factors of G/K}, which would contradict the definition of $p$.
	By Fact~\ref{fact:solvable maximal normal subgroup prime index}, $N/K$ is in fact a proper normal subgroup of smallest index in $G/K$ (removing ``prime''). By Theorem~\ref{thm:Berkovich subgroup smallest index}, we further have that $N/K$ is a proper subgroup of smallest index in $G/K$ (removing ``normal''). Thus, $N$ is the has the smallest index of any subgroup of $G$ that contains $K$. Any equalizer of two homomorphisms in $\Hom(G, H)$ contains $K$, so no equalizer can have smaller index than $N$. So, $\Lambda \leq \mu(N) = 1/p$.
\end{proof}

To prove Theorem~\ref{thm:Lambda for G or H solvable} in the case where $H$ is solvable, we use the following fact.
\begin{lemma}
	\label{lem:if H solvable then G/K solvable}
	Let $G$ be a group and $H$ a solvable group. Let $K = \bigcap_{\phi \in \Hom(G, H)} \ker \phi$. Then $G/K$ is solvable.
\end{lemma}
\begin{proof}
	
	For each $\phi \in \Hom(G, H)$, we have that $\phi(G)$ is solvable with derived length at most the derived length of $H$. 
	Also, $G/\ker \phi \cong \phi(G)$.
	So, $\prod_{\phi \in \Hom(G, H)} G/\ker \phi$ is a direct product of solvable subgroups with bounded derived length, so is solvable.  
	Let $\psi\colon  G \to \prod_{\phi \in \Hom(G, H)} G/\ker \phi$ be the projection onto each coordinate. Then $\ker \psi = K$. And, $\psi(G)$ is a subgroup of a solvable group, so is solvable. Thus, $G/K = G/\ker \psi \cong \psi(G)$ is solvable.
\end{proof}

We can now prove Theorem~\ref{thm:Lambda for G or H solvable} in the case where $H$ is solvable.

\begin{proof}[Proof of Theorem~\ref{thm:Lambda for G or H solvable} in the case where $H$ is solvable]
	
	Let $K = \bigcap_{\phi \in \Hom(G, H)} \ker \phi$. Let $p$ be the smallest prime divisor of $\abs{G}$ such that $G$ has a normal subgroup of index $p$ and $H$ has an element of order $p$. If $N$ is a normal subgroup of prime index and $H$ contains an element $h$ of order $\ind{G}{N}$, then $K \leq N$, since the isomorphism $G/N \to \langle h \rangle$ lifts to a homomorphism $G \to \langle h \rangle$ with kernel $N$. So, $p$ is the smallest prime index of a normal subgroup of $G$ that contains $K$. So, $p$ is the smallest prime index of a normal subgroup of $G/K$.
	
	We have that $G/K$ is solvable by Lemma~\ref{lem:if H solvable then G/K solvable}. So, the case of Theorem~\ref{thm:Lambda for G or H solvable} in which the domain is solvable, we have that $\Lambda_{G/K, H} = 1/p$. Then, by Lemma~\ref{lemma:irr-normal-subgroup}, we have $\Lambda_{G, H} = \Lambda_{G/K, H} = 1/p$.
\end{proof}

\section{Alternating domain, combinatorial list-decoding}
\label{section:alternating}

In this section, we will find that homomorphism codes with alternating domain are CombEcon. The exact constant is stated in Theorem~\ref{thm:alt-main-with-constant}. We remark that the constant in the $\poly(1/\eps)$-bound on list size can be improved using the SRG methods of Section~\ref{section:SRG}. The proof here utilizes the sphere packing bound, the sub-bucket bound, and a previous result on length of subgroup chains in symmetric groups~\cite{Bab_subgroupchain}, which do most of the heavy lifting. 

First, in Section~\ref{section:alt-background}, we present some background on the structure of alternating groups. In Section~\ref{section:alt-tools-2}, we present a corollary to the sphere packing bound, Lemma~\ref{lemma:sphere-packing}. In Section~\ref{section:alt-proof} we show that $\Lambda_{A_n, H}$ can only take the values $1/n$ and $1/\binom{n}{2}$ and prove the claim that $\altgroups \times \groups$ is CombEcon. Section~\ref{section:alt-blowup} addresses the list-decoding radius of homomorphism codes with alternating domain, by exhibiting a ``blowup'' in list size when agreement is exactly $\Lambda$, or when radius is $(1-\Lambda)$.

\subsection{Background on structure of alternating groups}
\label{section:alt-background}

For a set $\Omega$, let $\Alt(\Omega)$ denote the alternating group on $\Omega$. Similarly, let $\Sym(\Omega)$ denote the symmetric group on $\Omega$. We denote $A_n = \Alt([n])$ and $S_n = \Sym([n])$.  

Let $G \leq \Sym(\Omega)$. For $\pi \in G$ and $x \in \Omega$, we denote by $x^\pi$ the action of $\pi$ on $x$. For $x \in \Omega$, denote by $G_x = \{ \pi \in G \mid x^\pi = x \}$ the point stabilizer of $x$. Let $\Delta \subseteq \Omega$.  Denote by $G_{(\Delta)} = \{ \pi \in G \mid (\forall x \in \Delta)(x^\pi = x)\}$ the pointwise stabilizer of  $\Delta$. Denote by $G_{\{\Delta\}} =  \{ \pi \in G \mid \Delta^\pi = \Delta \}$ the setwise stabilizer of $\Delta$, where $\Delta^\pi := \{ x^\pi: x \in \Delta\}$. 

We present a few useful structural results for alternating and symmetric groups. The following theorem, due to Liebeck (see \cite[Theorem 5.2A]{DM}), describes the large subgroups of $A_n$. 

\begin{theorem}[Jordan-Liebeck]
	\label{thm:JordanLiebeck}
	Let $n \geq 10$ and let $r$ be an integer with $1 \leq r < n/2$. 
	Suppose that $K \leq A_n$ has index $\ind{A_n}{K} < \binom{n}{r}$. Then, for some $\Delta \subseteq [n]$ with $\lvert \Delta \rvert < r$, we have $(A_n)_{(\Delta)} \leq K \leq (A_n)_{\{\Delta\}}$. 
\end{theorem}

We will need the following result from~\cite{Bab_subgroupchain}, which describes the length of subgroup chains. 
\begin{theorem}[Babai]
	\label{thm:babai-subgroup-chain}
	
	The length of any subgroup chain in $S_n$ is at most $2n-3$. 
\end{theorem}

\begin{corollary}
	\label{cor:babai-subgroup-interval}
	The length of every subgroup chain between $A_{n-k}$ and $S_n$ is at most $2k$.
	
\end{corollary}

\subsection{Sphere packing by low-depth subgroups} 
\label{section:alt-tools-2}

We present a consequence (Lemma~\ref{lem:sphere-packing-subgroups}) of the sphere packing bound (Lemma~\ref{lemma:sphere-packing}), which we will use to prove that alternating groups are universally CombEcon. This lemma bounds the list size in terms of a ``starting set'' of subgroups and the subgroup depth of its members. This approach depends very little on the codomain $H$. 

Throughout this section, we let $G$ be a finite group, $H$ a group (finite or infinite), $\Lambda = \Lambda_{G, H}$ the maximum agreement, $f \from G \to H$ a received word, $\eps > 0$ a real, $\highagr = \highagr(\aHom(G, H), f, \Lambda + \eps)$ the list, and $\Psi \subseteq \aHom(G, H)$ as defined in Lemma~\ref{lemma:sphere-packing}. 

Recall our strategy for proving CombEcon using the Sphere Packing Lemma --- We divided the list $\highagr$ into buckets $\buckl\psi$ for $\psi \in \Psi$,

where 

\begin{equation*}
\buckl\psi = \{ \varphi \in \cL(\aHom(G,H), f, \Lambda + \eps) \mid \agr(\psi, \varphi) > \Lambda^2 \}.
\end{equation*}
We then split the bucket $\buckl\psi$ into further \emph{sub-buckets} according to the location of agreement with $\psi$. Each sub-bucket of $\buckl\psi$ is labeled by a subgroup $K$ of $G$ which we call the \emph{label subgroup}. We defined $\sbk_{\psi, K} \subseteq \buckl\psi$ to be the subset of homomorphisms whose equalizer with $\psi$ contain $K$; that is,

\begin{align*}
	\sbk_{\psi, K} = \{ \varphi \in \buckl\psi \mid K \leq \Eq(\varphi, \psi) \}.
\end{align*}

We concern ourselves now with the the set of label subgroups; we call such a set a \emph{starting set}. 

Intuitively, a set $\cS$ of subgroups is a starting set if the upper range of the subgroup lattice of $G$ contains only supergroups of elements in $\cS$.

With an appropriate notion of ``upper range,'' these starting sets form a sufficient set of label subgroups so that the sub-buckets $\sbk_{\psi, K}$ cover the bucket $\buckl\psi$ (see Remark~\ref{rmk:subbucket-subgroups}).

\begin{definition}[$(G,\lambda)$-starting-set]
Let $\cS$ be a set of subgroups of $G$. Let $\lambda \in (0,1)$. We say that $\cS$ is a \defn{$(G,\lambda)$-starting-set} if 
	\begin{equation*}
	(\forall K \leq G)( \mu_G(K) > \lambda \Rightarrow (\exists S \in \cS)(S \leq K)). 
	\end{equation*}
\end{definition} 

\begin{remark}
\label{rmk:subbucket-subgroups}
Suppose that $\cS$ is a $(G, \Lambda^2)$-starting set. Then, for any $f: G \rightarrow H$ and $\psi \in \aHom(G,H)$, 

\begin{align}
	\buckl\psi = \bigcup_{K \in \cS} \sbk_{\psi, K}.
\end{align}
Combining this with the bucket-splitting lemma (Lemma~\ref{lemma:bucket-splitting}),
\begin{align}
	\highagr = \bigcup_{\psi \in \Psi} \bigcup_{K \in \cS} \sbk_{\psi, K}.
\end{align}
\end{remark}

This allows us to use our bound on $\sbk_{\psi, K}$ from Corollary~\ref{cor:bucket-bound} to bound the size of the list.

\begin{lemma}[Sphere packing via low-depth subgroups]
\label{lem:sphere-packing-subgroups}
Let $G$ be a finite group, $H$ a group, and $\eps > 0$. 
Let $\cS$ be a $(G, \Lambda^2)$-starting-set. Then, 
	\begin{equation}
	\ell(\Hom(G,H), \Lambda+\eps) \leq \left( \frac{1}{4(\Lambda+\eps)\eps} + 1 \right) \cdot \sum_{K \in \cS} 1/\eps^{\depth(K)}. 
	\end{equation}
\end{lemma}
\begin{proof}

By Remark~\ref{rmk:subbucket-subgroups}, we have that $\abs{\highagr} \leq \sum_{\phi \in \Phi} \sum_{K \in \cS} \abs{\sbk_{\psi, K}}$. By Lemma~\ref{lemma:sphere-packing}, $\abs{\Psi} \leq \frac{1}{4(\Lambda+\eps)\eps}$, and by Corollary~\ref{cor:bucket-bound}, $\abs{\sbk_{\psi, K}} \leq 1/\eps^{\depth(K)}$.
\end{proof}

We will use Lemma~\ref{lem:sphere-packing-subgroups} in the proof that alternating groups are universally CombEcon.

\subsection{Proof $A_n$ is universally CombEcon}
\label{section:alt-proof}

We prove that $A_n$ is CombEcon by proving Theorem~\ref{thm:alt-main-with-constant} below, which states a constant for the CombEcon claim. (This constant is improved via the methods of SRG groups in Section~\ref{section:SRG}.)

\begin{theorem}
	\label{thm:alt-main-with-constant}
	For every group $H$, integer $n \geq 38$ and $\eps > 0$, we find that 
	\begin{equation*}
	\ell(\Hom(A_n,H), \Lambda_{A_n,H} + \eps) \leq 1/\eps^{16}.
	\end{equation*}
\end{theorem}

\begin{proof}[Proof of Theorem~\ref{thm:alt-main-with-constant}] 
	By Lemma~\ref{lem:lambda-alternating}, we find that $\Lambda^2 \geq 1/{n\choose 2}^2 \geq 1/{n \choose 5}$. We use Lemma~\ref{lem:sphere-packing-subgroups}.
	
We first define a starting set 
	\begin{equation*}
	\cS = \{ (A_n)_{(\Delta)}: \Delta \subseteq [n], \abs{\Delta} = 5 \}. 
	\end{equation*}
That $\cS$ is an $(A_n, \Lambda^2)$-starting-set follows by Jordan-Liebeck, Theorem~\ref{thm:JordanLiebeck}. By Corollary~\ref{cor:babai-subgroup-interval}, we find that $\depth_{A_n}(K) \leq 2\cdot 5 -2  = 8$ for all $K \in \cS$. 
We assume $\eps^2 < 2\Lambda$ by Lemma~\ref{lemma:small-epsilon-vs-lambda}. Since $\abs{\cS} = {n \choose 5} < (\binom{n}{2}/2)^3 = (\Lambda/2)^3 < 1/\eps^6$,  Theorem follows from Lemma~\ref{lem:sphere-packing-subgroups}. 
\end{proof}

\subsection{Upper bound on list-decoding radius}
\label{section:alt-blowup}

We showed in Section~\ref{section:alt-proof} that $\altgroups \times \groups$, and all of its subclasses, have list-decoding radius greater than $1 - (\Lambda + \eps)$ for all $\eps > 0$.

In contrast, $\altgroups \times \groups$ and many of its subclasses have list-decoding radius at most $1 - \Lambda$. In this section, we demonstrate such a subclass. The number of homomorphisms within a closed ball of radius $1 - \Lambda$ of a received word will be exponential in $\log \abs{G}$ and $\log \abs{H}$. We note that $\abs{H} \geq \abs{G}$ unless $\Lambda = 0$.

\begin{proposition}
	For any $n$, and $\lambda \in \{1/n, 1/\binom{n}{2}\}$, there exists a finite group $H_n$ such that $\Lambda_{A_n, H_n} = \lambda$ and
	\begin{align}
	\ell(\Hom(A_n,H_n), \LGH) = 2^{\Omega(n)} \geq 2^{\Omega\left(\sqrt[3]{\log \abs{H}}\right)}.
	\end{align}

	Moreover, for any fixed $n \geq 10$, and any integer $M$, there is a finite group $H$ such that
	\begin{align}
	\ell(\Hom(A_n, H), \LGH) \geq M.
	\end{align}
\end{proposition}
\begin{proof}
We use the same construction for both parts. To prove the first claim, let $k = n$. To prove the second claim, let $k \geq \log_2 M$.

Suppose $\lambda = 1/n$. Let $H_n = A_{n+1}^k$, the direct product of $k$ copies of $A_{n+1}$. Then $\Lambda_{A_n, H_n} = 1/n$. Let $f \from A_n \to H_n$ by $f(g) = (g, \dots, g)$, the diagonal identity map, where $A_n$ is embedded in $A_{n+1}$. For nonempty $S \subseteq [n]$ and $j \in [n]$, let $h = h(S, j) = (h_1, \dots, h_k) \in H_n$, where $h_i$ is the transposition $(j, n+1)$ if $i \in S$ and $1$ otherwise.  For each such $h$, let $\varphi_h \in \Hom(A_n, H_n)$ be given by $\varphi_h(g) = h^{-1} f(g) h$. Each $\varphi_h$ has agreement $\agr(\varphi_h, f) = 1/n = \Lambda$ with $f$. There are $n (2^k - 1)$ such $h$, so $\ell(\Hom(A_n, H_n), \Lambda) \geq n (2^k - 1)$. 

Suppose $\lambda = 1/\binom{n}{2}$. Let $H_n = A_n^k$. Then, $\Lambda_{A_n, H_n} = 1/\binom{n}{2}$. Let $f \from A_n \to H_n$ by $f(g) = (g, \ldots, g)$, the diagonal identity map. For nonempty $S \subseteq [n]$ and $\tau \in S_n$ is a transposition, let $h = h_{S, \tau} = (h_1, \ldots, h_k) \in A_n^k$, where $h_i = \tau$ if  $i \in S$ and $1$ otherwise. For each such $h$, let $\varphi_h \in \Hom(A_n, H_n)$ be given by $\varphi_h(g) = h^{-1} f(g) h$. Each such $\phi_h$ has agreement $\agr(\varphi_h, f) = 1/\binom{n}{2}$. There are $\binom{n}{2} (2^k - 1)$ such $h$, so $\ell(\Hom(A_n, H_n), \Lambda) \geq \binom{n}{2} (2^k - 1)$.

\end{proof}

We remark that $\ell(\Hom(A_n, H), \Lambda_{A_n, H})$ is not bounded as a function of $n$ for a wide variety of classes of $H$.

\section{Shallow random generation}
\label{section:SRG}

In this section, we prove results about SRG groups, defined in Section~\ref{section:intro-SRG}), for which few random elements tend to generate a shallow (low depth) subgroup. 
In Section~\ref{section:SRG An}, we will show that alternating groups are SRG. In 
Section~\ref{section:SRG implies KLC} we will prove that SRG groups are also ``KLC'' groups (another generation property, defined in Section~\ref{section:KLC def}). The consequences of SRG will be proved using the KLC assumption in Section~\ref{section:srg-consequences}. 

Recall that Section~\ref{sec:Hom-versus-aHom} showed 

the code $\Hom(G,H)$ is CombEcon
if and only if $\aHom(G,H)$ is CombEcon, and 
similarly for CertEcon, and AlgEcon under modest
assumptions of the representation of the groups.
All our results about SRG groups (universal CombEcon and CertEcon, see Section~\ref{section:srg-consequences}) will take advantage of this equivalence, as our proofs will argue about $\Hom(G,H)$ instead of $\aHom(G,H)$. To reflect this, concepts in this section are defined in terms of subgroup generation using $\gengroup{ \cdot}$ instead of affine generation using $\genaff{\cdot}$.

Further recall Proposition~\ref{prop:prelim-lambda-Hom-vs-aHom} which states that, if $\Lambda_{G,H} \neq 0$, then `$\aHom$' can be replaced by `$\Hom$' in the definition of $\Lambda_{G,H}$, i.e., 
$$\Lambda_{G,H} = \max_{\substack{\varphi,\psi \in \Hom(G,H) \\ \varphi \ne \psi}} \agr(\varphi, \psi). $$

If we used instead the affine version of our tools and arguments to reason directly about $\aHom(G,H)$ (instead of reasoning about $\Hom(G, H)$ then using the material of Section~\ref{sec:Hom-versus-aHom} to get bounds for $\aHom(G, H)$), the degrees of the polynomials in the $\poly(1/\eps)$ expression for the SRG results would be identical.

\subsection{Alternating groups are SRG}
\label{section:SRG An}

In this subsection, we prove that $\altgroups$ is SRG.

\begin{theorem}
	\label{thm:comb An is SRG}
The class of alternating groups is SRG. In fact, for all $k \geq 2$ there is an integer $n_k$ such that for all $n \geq n_k$, the alternating group $A_n$ is $(k, 4k-2)$-shallow generating. 

The quantity $n_k$ may be large as a function of $k$; however, there is an integer $n_0$ such that for all $k \geq 2$ and all $n \geq \max\{n_0, (3k)^3\}$, the alternating group $A_n$ is $(k, 6k-2)$-shallow generating.
\end{theorem}

\paragraph{Consequences.}  

Before proving Theorem~\ref{thm:comb An is SRG}, we first discuss its consequences. In Section~\ref{section:srg-consequences} we prove facts that hold for all SRG groups. Here are the implications of those facts for alternating groups.

From Theorem~\ref{thm:SRG implies CombEcon}, we find that $\aHom(A_n, H)$ is CombEcon with degree $8$, i.e.,  \linebreak $\ell(\aHom(A_n,H), \Lambda+\eps) < 1/\eps^8$ for all $H \in \groups$. 
We remark that the constant 8 can be improved to 6. This is by improving the $(8,\Lambda, 7)$-generated claim to $(6,\Lambda,5)$-generated by going through the proof with a ``$\depth_\Lambda$'' notion instead of $\depth_G$ as written. This $\depth_\Lambda(K)$ refers to maximal length of a subgroup chain from $K$ to a subgroup of density greater than $\Lambda$.

By Theorem~\ref{thm:SRG implies CertEcon}, $\altgroups$ is CertEcon. More specifically, $\altgroups$ is universally strong certificate-list-decodable using $O(\ln(1/\eps)/\eps^9)$ queries and computation time. Certificates are generated by querying the received word on sets of uniform size. This is formalized in Section~\ref{section:SRG implies CertEcon}.

If $H = S_m$ and $m<2^n/\sqrt{1.6n}$, then calls to the subword extender {\HomExt} can be executed in $\poly(n,m)$-time, by Theorem~\ref{thm:homext}. We combine $\HomExt$ with the certificate-list-decoder to find a list-decoder. One call to $\HomExt$ is made per certificate in the returned certificate-list, so the list-decoder runs in time $\poly(1/\eps, n, m)$ while using $\poly(1/\eps)$-queries. \\  

\paragraph{Presentation of proof.}

To prove Theorem~\ref{thm:comb An is SRG} we first present a useful result \cite[Theorem 1.5]{Babai1989} that two random elements of the symmetric group generate a `large' subgroup with high probability. This event is denoted $E(n,k)$ in the statement.

\begin{theorem}[Babai] 
	\label{thm:babai-generating-alternating}
	Let $\pi, \sigma$ be a pair of independent uniform random elements from $S_n$. For $ 0 \leq k \leq n/3$, let $E(n,k)$ denote the following event: The subgroup $K = \langle \pi, \sigma \rangle$ acts as $S_r$ or $A_r$ on $r$ elements of the permutation domain for some $r \geq n-k$. Then, 
	\begin{equation}
	\label{eq:babai-generating-alternating}
	\Pr(E(n,k)) = 1 - \binom{n}{k+1}^{-1} + O \left( \binom{n}{k+2}^{-1} \right).
	\end{equation}
	The constant implied by the big-$O$ notation is absolute. 
\end{theorem}

\begin{remark}
Suppose that we choose $\pi$ and $\sigma$ from $A_n$ (instead of $S_n$) in Theorem~\ref{thm:babai-generating-alternating}. The same conclusion is still true. However, using only Theorem~\ref{thm:babai-generating-alternating} as justification, the conclusion is slightly weaker -- there will be a coefficient of $4$ in front of $\binom{n}{k+1}^{-1}$. In our application, this coefficient makes no difference to our argument. 
\end{remark}

We now bound the depth of the subgroup generated by two elements, given that $E(n, 2k)$ occurred. 

\begin{claim} 
	\label{claim:An-E-subgroupdepth}
	There exists $n_0$ such that the following holds: 
	Let $E, k, \pi, \sigma$ be defined as in Theorem~\ref{thm:babai-generating-alternating}. If $n \geq n_0$ and $E(n, k)$ occurs, then $\depth_{A_n}(\langle \pi, \sigma \rangle) \leq 2k-2$. 
\end{claim} 
\begin{proof}
	Let $K = \ang{\pi, \sigma}$. If $E(n,k)$ occurs, then  $K$ acts as $S_r$ or $A_r$ on some subset of $r$ elements of $[n]$, for some $r \geq n-k > n/2$. So, $A_r \leq K$. 
	
	By Corollary~\ref{cor:babai-subgroup-interval}, we find that $\depth_{A_n} (K) \leq \depth_{A_n} \left(A_r \right) \leq 2k-2$.

\end{proof}

We can now prove Theorem~\ref{thm:comb An is SRG}.

\begin{proof}[Proof of Theorem~\ref{thm:comb An is SRG}]
	We prove the first part.

	By Theorem~\ref{thm:babai-generating-alternating} and Claim~\ref{claim:An-E-subgroupdepth}, it holds for large $n$ that  
	\begin{equation*}
	\Pr_{\pi, \sigma \in A_n} [ \depth_{A_n}(\langle \pi, \sigma \rangle) > 4k-2] \leq \Pr_{\pi, \sigma \in A_n} [ \neg E(n,2k)] \leq \frac{4}{{n \choose 2k+1}} \leq \frac{1}{\binom{n}{2}^k} = (\Lambda_{A_n}^*)^k. 
	\end{equation*}
	
	It follows that $A_n$ is $(k, 4k-2)$-shallow generating. 
	
	The proof of the second part is similar.
\end{proof}

\subsection{Subset-generation}
\label{section:KLC def}

In this section we define a useful technical ``KLC'' condition on groups, which SRG groups satisfy (shown in the next section). The connection between KLC and universal CombEcon or CertEcon is more direct, and our ``SRG implies universally CombEcon and CertEcon'' results are proved through the KLC property (Section~\ref{section:srg-consequences}).

\begin{definition}[$(k,\lambda,c)$-subset-generated]
	\label{def:klc}
	
	Let $G$ be a finite group, $k$ a nonnegative integer, $0 \leq \lambda < 1$, and $c \geq 0$. We say that $G$ is \defn{$(k, \lambda, c)$-subset-generated} if, for all subsets $S \subseteq G$ with $\mu(S) > \lambda$, we have that
	\begin{align}
	\label{eqn:KLC-def}
	\Pr_{s_1, \dots, s_k \in S} [\mu(\gengroup{ s_1, \dots, s_k } ) > \lambda] \geq \left( 1- \frac{\lambda}{\mu(S)} \right)^c, 
	\end{align}
	where $s_1, \dots, s_k$ are chosen independently and uniformly from $S$.

\end{definition}

Note that, if we define $\eps = \mu(S) - \lambda$, then  $ 1- \frac{\lambda}{\mu(S)} =  \frac{\eps}{\lambda + \eps} $, so Equation~\eqref{eqn:KLC-def} mirrors the expression of Lemma~\ref{lemma:random-el-depth}.

	We say that $G$ is \defn{$(k, \lambda, c)$-affine-generated} if it satisfies Definition~\ref{def:klc} but with $\gengroup{s_1, \dots, s_k}$ replaced by $\genaff{s_1, \dots, s_k}$.

We will make a few remarks on these definitions below, but first we define KLC classes of groups. 

\begin{definition}[KLC]
	Let $\grpclass$ be a class of finite groups. We say that $\grpclass$ is \defn{KLC} if there exists a positive integer $k$ and a constant $c > 0$ such that, for all $G \in \grpclass$ and for all groups $H$, we have that $G$ is $(k, \Lambda_{G, H}, c)$-subset-generated.
\end{definition}

The notion of ``KLC-affine'' can be defined analogously. But, according to Remark~\ref{rmk:KLC def}~\ref{label:rmk coset generated}

below, the two conditions are equivalent on a class of groups.

We make a few remarks on the definitions of $(k, \lambda, c)$-subset-generated groups. 
\begin{remark}
	\label{rmk:KLC def}
	\begin{enumerate}[(a)]
		\item For every $k \geq 1$ and $c \geq 0$, the class $\groups$ of all finite groups is $(k, 0, c)$-subset-generated. 
		\item Classes of $(k, \lambda, c)$-subset-generated groups are monotone in both $k$ and $c$. More specifically, for $k' > k$ and $c' > c$, if $G$ is $(k, \lambda, c)$-subset-generated, then $G$ is also $(k', \lambda, c)$-subset-generated and $(k, \lambda, c')$-subset-generated. 
		\item If $G$ is $(k, \lambda, c)$-affine-generated, then it is $(k, \lambda, c)$-subset-generated. If $G$ is $(k, \lambda, c)$-affine-generated, then it is $(k+1, \lambda, c)$-generated. 

	\label{label:rmk coset generated}
	\end{enumerate}
\end{remark}

\subsection{SRG implies subset-generation} 
\label{section:SRG implies KLC}

We prove that SRG implies KLC, using a straightforward application of Bayes' rule. 
\begin{theorem}[SRG implies KLC]
\label{thm:SRG-suff-condition}
If a class $\grpclass$ of groups is SRG, then $\grpclass$ is KLC.

In particular, let $G$ be a finite group, $k, d \in \NN$, and $\lambda >0$. If $G$ is $(k,d)$-shallow generating, then $G$ is $(k+d, \lambda, 1+d)$-subset-generated for all $\lambda \geq \Lambda_G^*$. 
\end{theorem}

\begin{proof}
All groups are trivially $(1, 0, 1)$-subset-generated, which covers the case where $\lambda = 0$. 

Let $\lambda >0$. 
By assumption, we know that
	\begin{equation}
	\Pr_{g_1, \ldots, g_k \in G}[ \depth(\langle g_1, \ldots, g_k \rangle ) > d] < \lambda^k. 
	\end{equation}
We check the definition of $(k+d, \lambda, 1+d)$-subset-generated.

Let $S \subseteq G$ be such that $\mu(S) >  \lambda$ and let $\eps = \mu(S) - \lambda$. We will pick a $k$-tuple $\boldg = (g_1, \ldots, g_k)$ and a $d$-tuple $\bolds = (s_1, \ldots, s_d)$ from $S$. We write $\langle \boldg \rangle$ to mean $\langle g_1, \ldots, g_k \rangle$. We write $\langle \bolds \rangle$ and $\langle \bolds, \boldg \rangle$ similarly.

Observe that 
\begin{align*}
\Pr_{\boldg \in S^k, \bolds \in S^d} \left[ \mu \left(\langle \boldg, \bolds \rangle \right) > \lambda \right] 
& \geq \Pr_{\bolds \in S^d} \left[ \mu \left( \langle \boldg, \bolds \rangle \right) > \lambda \mid \depth(\langle \boldg \rangle) \geq d \right]
		\cdot  \Pr_{\boldg \in S^k}[ \depth(\langle \boldg \rangle) \geq d].
\end{align*}

We bound the two components of the right hand side separately. When we drop the subscript on $\Pr$, that means the elements are chosen at random from $G$. First, we consider the second component. 

\begin{align*}
 \Pr_{\boldg \in S^k}[ \depth(\langle \boldg \rangle) \geq d] 
 & = \Pr_{\boldg \in G^k} [ \depth(\langle \boldg \rangle) \geq d]  \mid \boldg \in S^k ] \\ 
 & = \frac{\Pr [\boldg \in S^k  \text{ and } \depth(\langle \boldg \rangle) \geq d ] }{\Pr[\boldg \in S^k]} \\ 
 & \geq \frac{\Pr[\boldg \in S^k] + \Pr[\depth(\langle \boldg \rangle) \geq d]  - 1}{\Pr[\boldg \in S^k]} \\ 
 & > \frac{\mu(S)^k - (\lambda)^k}{\mu(S)^k} 
 	= 1  - \left( \frac{\lambda}{\mu(S)} \right)^k \\ 
 & \geq 1 - \frac{\lambda}{ \lambda + \eps} 
 		= \frac{\eps}{\lambda + \eps}. 
\end{align*}

Now, it suffices to show that the first component has probability bounded by $\left(\frac{\eps}{\lambda + \eps} \right)^d$. But, if $\depth(\langle \boldg \rangle) \geq d$, then it follows from Lemma~\ref{lemma:random-el-depth}, with $K = \langle \boldg \rangle$ and $\lambda = \lambda$, that 
\begin{align*}
 \Pr_{\bolds \in S^d} \left[ \mu \left( \langle \boldg, \bolds \rangle \right) > \lambda \mid \depth(\langle \boldg \rangle) \geq d \right] > \left(\frac{\eps}{\lambda + \eps} \right)^d. 
\end{align*}

We conclude that $G$ is $(k+d, \lambda, 1+d)$-subset-generated, or, 
\begin{equation*}
\Pr_{s_1, \ldots, s_{k+d} \in S} \left[ \mu \left(\langle s_1, \ldots, s_{k+d} \rangle \right) > \lambda \right] = \Pr_{\boldg \in S^k, \bolds \in S^d} \left[ \mu \left(\langle \boldg, \bolds \rangle \right) > \lambda \right]  > \left(\frac{\eps}{\lambda + \eps} \right)^{d+1}. 
\end{equation*}

\end{proof}

\section{Consequences of shallow random generation}
\label{section:srg-consequences}

We proved the claimed consequences of shallow random generation, using the KLC property. These consequences include universal CombEcon (Section~\ref{section:SRG implies CombEcon}) and universal CertEcon (Section~\ref{section:SRG implies CertEcon}). 

Universal AlgEcon follows if a subword extender ($\HomExt(G,H)$ oracle) is provided, as discussed in Section~\ref{section:contr-cert-homext-to-alg}. Stronger subword extenders provide better guarantees on the output list and lower bounds on $\Lambda$ (Section~\ref{section:SRG Lambda}).

\subsection{SRG implies CombEcon}
\label{section:SRG implies CombEcon}

KLC implies universally CombEcon.

\begin{theorem}
	\label{thm:SRG implies CombEcon}
	If a class $\grpclass$ is SRG, then $\grpclass$ is universally CombEcon. More precisely, let $k \in \NN$ and $c>0$. If $G$ is a $(k, \Lambda_{G, H}, c)$-subset-generated group, then $\ell(\Hom(G, H), \Lambda_{G, H} + \eps) \leq 1/\eps^{\max\{c,k\}}$ for all $\eps \in (0, 1 - \Lambda)$ and groups $H$. 
	
\end{theorem}

\begin{proof}
	By Theorem~\ref{thm:SRG-suff-condition}, if $\grpclass$ is SRG, then $\grpclass$ is KLC.

	We define a bipartite graph $\gph$. The left vertex set of $\gph$ is $G^k$. The right vertex set is $\highagr = \highagr(\Hom(G,H), f, \Lambda+\eps)$. There is an edge between $(g_1, \dots, g_k) \in G^k$ and $\phi \in \highagr$ if $g_1, \dots, g_k \in \Eq(f, \phi)$, and $\mu(\gen{g_1, \dots, g_k}) > \Lambda$. 
	
	By the definition of $(k, \Lambda_{G, H}, c)$-subset-generated group, each right vertex has degree at least $(\frac{\eps}{\Lambda + \eps})^c \abs{S}^k$, which is at least $(\frac{\eps}{\Lambda+\eps})^c (\Lambda + \eps)^k \abs{G}^k \geq \eps^{\max \{ c, k\}} \abs{G^k}$.

	By Lemma~\ref{lem:unique-homomorphism-K-g} with $K = 1$, each left vertex of $\gph$ has degree at most one.
	
	So, by part~(a) of
Lemma~\ref{lemma:mean-comb-bipartite} (double counting),
$\abs{\highagr} \leq 1 / \eps^{\max \{ c, k\}}$.
\end{proof}

\subsection{SRG implies CertEcon}
\label{section:SRG implies CertEcon}

A KLC class of groups is also universally CertEcon. The argument is conceptually similar to that of CombEcon, but we need to formalize algorithmic issues.

Again, let $\cW_\Lambda$ denote the set of $G\partialto H$ partial maps $\gamma$ satisfying $\mu(\langle \dom \gamma \rangle ) > \Lambda_{G,H}$. 

\begin{theorem}
	\label{thm:SRG implies CertEcon} 
	If $\grpclass$ be an SRG class of groups, then $\grpclass$ is universally strong $\cW_\Lambda$-CertEcon. We assume that all groups in $\grpclass$ are encoded groups, that (nearly) uniform elements of $G$ are provided, and that we have oracle access to the entries of the received word. 
\end{theorem}

If a partial map $\gamma$ can be extended to some homomorphism and satisfies $\mu(\langle \dom(\gamma) \rangle) > \Lambda$, then $\gamma$ is a $\cW$-certificate. However, $\dom(\gamma)$ may fail to generate the entire group, i.e., $\langle \dom(\gamma) \rangle \lneq G$. Following the strategy of Section~\ref{section:terminology-subword}, we may combine a certificate list-decoder and a \textsc{Homomorphism Extension} oracle $\HomExt_\Lambda$ (a $\cW_\Lambda$-subword extender). The next section discusses how to take this strategy a step further using a stronger \textsc{Homomorphism Extension} oracle to lower bound $\Lambda$.

\paragraph{Domain certificates versus certificates.\\}

We develop some terminology for Theorem~\ref{thm:SRG implies CertEcon}, based on the natural idea of generating certificates by querying the received word $f$. ``Domain certificates'' (dependent on $f$) are subsets of the domain that define a certificate when restricting $f$ to that set. 

Let $G$ and $H$ be groups and $f \in H^G$ be a received word in the codespace of $\Hom(G,H)$. Let $S \subseteq G$ be a subset. Denote by $f_S$ the restriction of $f$ to $S$, i.e., the $H \partialto G$ partial map with domain $S$ defined by $f_S(g) = f(g)$ for $g \in S$.

\begin{definition}[Domain certificate] 
When the code $\Hom(G,H)$ and the received word $f$ are understood, we say that a subset $S \subseteq G$ is a \defn{domain certificate} if the $G \partialto H$ partial map $f_S$ is a certificate for $\aHom(G,H)$. 

For a set $\cW$ of $G \partialto H$ partial maps, we say that $S$ is a \defn{domain $\cW$-certificate} if $f_S$ is a $\cW$-certificate for $\Hom(G,H)$. 
\end{definition}

Note that a domain certificate $S \subseteq G$ is a domain $\cW_\Lambda$-certificate if and only if $\mu( \langle S\rangle) > \Lambda_{G,H}$. 

\begin{definition}[Domain-certificate-list]
We say that a list $\Upsilon$ of subsets of $G$ is a \defn{domain-certificate-list} for a subset $\cL \subseteq \aHom(G,H)$ of affine homomorphisms if $\Upsilon$ contains a domain certificate for each codeword in $\cL$. 
We define \defn{domain-$\cW$-certificate-lists} similarly. 
\end{definition}

\paragraph{Domain certificate result.\\}
Now, we can restate the unabridged SRG result (Theorem~\ref{thm:contr-srg-cert-unabr}) in terms of domain certificates.

\begin{theorem}[SRG implies CertEcon, via domain certificates]
\label{thm:SRG-Cert-domaincertificates}
Let $k \in \NN$ and $c >0$. Let $G$ be a $(k, \Lambda_{G,H}, c)$-subset-generated group and $H$ a group.   Let $f: G \to H$, $\eps > 0$ and $\eta > 0$. Let $\Upsilon$ be a list of $\left\lceil\frac{1}{\eps^b} \ln \left(\frac1{\eta\eps^b} \right) \right \rceil$ independently chosen subsets of $G$, each of size $\max\{ c, k \}$. Then, with probability at least $(1-\eta)$, $\Upsilon$ is a domain-$\cW_\Lambda$-certificate-list of $\cL(\aHom(G,H), f, \Lambda+\eps)$. 
\end{theorem}

The proof is delayed to first discuss its implications and access model.

\begin{remark}[Access model] 
To generate the domain-$\cW$-certificate-list, we need access only the domain, only in the ability to generate random elements. No knowledge of $H$ is required. The dependence on $H$ appears only in the $\Lambda_{G,H}$ of the assumption that $G$ is $(k, \Lambda_{G,H}, c)$-subset generated, but the KLC assumption means that $G$ is $(k,\Lambda_{G,H},c)$-subset generated for every $H$. Knowledge of $\Lambda_{G,H}$ is also not required.  
\end{remark}

Theorem~\ref{thm:SRG-Cert-domaincertificates} produces domain $\cW$-certificates. No work is involved other  than  generating these $\poly(1/\eps)$ uniform random elements of $G$. To then produce actual $\cW$-certificates, simply query $f$ on the domain $\cW$-certificates, an additional $\poly(1/\eps)$ queries to $f$. Theorem~\ref{thm:SRG implies CertEcon} follows. 

\begin{remark}[Amount of work] 
Theorem~\ref{thm:SRG implies CertEcon} implies a stronger result than strong CertEcon, as only a $\poly(1/\eps)$ amount of work is required in the unit cost model (no dependency on $\abs{G}$).\footnote{
Two incomparable sufficient conditions for the access model to $G$ are black-box access and polycyclic presentations. In a black-box group, $\eps$-uniform elements can be generated in polynomial time \cite{Bab91BBpolygen}. Given a polycyclic presentation, exactly uniform elements can be generated. }
\end{remark}

\paragraph{Analysis.\\}

To prove Theorem~\ref{thm:SRG-Cert-domaincertificates} we check the definition of domain $\cW_\Lambda$-certificate-list. In other words, with probability $(1-\eta)$, for every $\varphi \in \highagr(\Hom(G,H), f, \Lambda + \eps)$ the list $\Upsilon$ contains a domain $\cW_\Lambda$-certificate $S_\varphi \subseteq G$ for $\varphi$. A sufficient condition to be a domain $\cW_\Lambda$ certificate is given below. 

\begin{observation}
\label{obs:SRG-cert-domWcert-suff-conditions} If the conditions $\mu(\langle {S} \rangle) > \Lambda$  and $S \subset\Eq(\varphi,f)$ are satisfied, then $S$ is a domain $\cW_\Lambda$-certificate for $\varphi$. 
\end{observation}
\begin{proof}[Proof of Theorem~\ref{thm:SRG-Cert-domaincertificates}] 
	Let $\highagr =\highagr(\Hom(G,H),= f, \Lambda + \eps)$. Recall that $G$ is $(k, \Lambda_{G,H}, c)$-subset generated. Let $b = \max\{c, k\}$.	Let $\Upsilon = \{ S_1, \ldots, S_t\}$ be the list of uniformly chosen subsets of $G$ as assumed. Then, $t=\left\lceil\frac{1}{\eps^b} \ln \left(\frac1{\eta\eps^b} \right) \right\rceil$ and each $S_i$ consists of $b$ uniformly and independently chosen elements of $G$. 
	
Fix $\varphi \in \highagr$. Fix $S \in \Upsilon$. We calculate the following. 
\begin{align*}
	\Pr [S \text{ is a domain $\cW_\Lambda$-certificate for }\varphi ]
	& \geq \Pr\left[ S \subseteq \Eq(\psi, f) \cap \mu \left( \langle S \rangle \right) > \Lambda \right]  \\ 
	& = \Pr \left[ \mu \left( \langle S \rangle \right) > \Lambda \bigg\vert S \subseteq \Eq(\psi, f) \right] 
	\cdot \Pr[S \subseteq \Eq(\psi, f) ] \\ 
	& > \left( \frac{\eps}{\Lambda + \eps} \right)^c \cdot (\Lambda+\eps)^k\\ 
	& > \eps^b. 
	\end{align*}
	
The first inequality follows from Observation~\ref{obs:SRG-cert-domWcert-suff-conditions}, and the second inequality follows from the definition of $(k, \Lambda, c)$-subset generated. 

The probability that $\Upsilon$ is not a domain $\cW$-certificate-list for $\cL$, i.e., there is $\varphi \in \cL$ such that $\Upsilon$ contains no domain-$\cW$-certificate for $\varphi$, is bounded by 
\begin{equation*}
\abs{\highagr} \cdot \left( 1 - \eps^b \right)^t \leq \frac{1}{\eps^b} \exp \left(-\eps^{b} \cdot t \right) < \eta,
\end{equation*}
where we have used that $\abs{\highagr} \leq 1/\eps^b$ by the CombEcon result Theorem~\ref{thm:SRG implies CombEcon}.
\end{proof}

\subsection{Improvements on $\Lambda$}
\label{section:SRG Lambda}

In this section we first discuss the role of $\Lambda$ in the relationship between CertEcon, $\HomExt$, and AlgEcon. Then, we will give an algorithm that improves our lower bounds on $\Lambda$ and discuss its benefits. 

\paragraph{Role of $\Lambda$ lower bounds in subword extenders.\\}

As discussed in generality (Section~\ref{section:terminology-subword}), for two sets $\cW_1$ and $\cW_2$ of partial maps $\Omega \to \Sigma$, if they satisfy $\cW_1 \subseteq \cW_2$ then a $\cW_1$-certificate-list-decoder and a $\cW_2$-subword-extender combine to a list-decoder. In our context, we consider the sets $\cW_\lambda$ consisting of $G \partialto H$ partial maps whose domain generate a subgroup of $\lambda$ density, i.e.,
$$\cW_\lambda := \left\{ \gamma: G \partialto H \mid  \mu(\gengroup{\dom \gamma}) > \lambda \right\}.$$
Since we have $\cW_{\Lambda_{G,H}}$-CertEcon results for SRG groups (Theorem~\ref{thm:SRG implies CertEcon}), it suffices to find $\cW_\lambda$-subword extenders, i.e., solve $\HomExt_\lambda(G,H)$, for $\aHom(G,H)$ with $\lambda \leq \Lambda_{G,H}$.

So, a stronger lower bound for $\Lambda_{G,H}$ allows use of a weaker $\HomExt$ oracle. 

\paragraph{Role of subword extenders in $\Lambda$ lower bounds.\\}

Conversely, a stronger $\HomExt(G,H)$ oracle can be used to update lower bounds on $\Lambda_{G,H}$, which allows better pruning of the output list. We will explain both clauses of this sentence below.

First, we discuss finding better lower bounds on $\Lambda_{G,H}$, using a stronger version of $\HomExt$ we call $\HomExt012$. The $\HomExt012$ Problem asks to distinguish between the cases of no extension, unique extension, and multiple extensions. In the case of a unique extension, it asks for the extension. 

\begin{definition}($\HomExt012(G,H)$)\\
\indent \textbf{Instance:} A partial map $\gamma: G \partialto H$. \\
\indent \textbf{Solutions:} The set defined by
$$\HExt(\gamma) : = \{ \varphi \in \Hom(G,H) : \varphi|_{\dom \gamma} = \gamma \}.$$
\indent \textbf{Output:} 
\indent \indent $\begin{cases} 
\text{`none'} & \text{ if }\abs{\HExt(\gamma)} = 0 \\ 
\varphi \in \HExt(\gamma) \, \, & \text{ if } \abs{\HExt(\gamma)} = 1 \\
\text{`multiple'} & \text{ if } \abs{\HExt(\gamma)} \geq 2 \end{cases}. $ 
\end{definition}

The $\HomExt012_\lambda(G,H)$ problem is defined similarly, but requiring only correct answers on the $G \partialto H$ partial maps in $\cW_\lambda$.

\begin{proposition}
\label{prop:SRG-012}
Let $G$ and $H$ be groups to which we are given black-box access. Suppose that we are given an oracle for $\HomExt012_\lambda(G,H)$ and an order oracle for subgroups. Then, a subword extender for $\aHom(G,H)$ can be implemented in $\poly(\enc(G))$-time in the unit-cost model for $H$. Moreover, for any $G \partialto H$ partial map $\gamma$ on which  $\HomExt012(G,H)$ returns `multiple,' the value of $\mu( \langle \dom \gamma \rangle )$ is a lower bound for $\Lambda$. 
\end{proposition}
\begin{proof}
The first conclusion of this proposition is the same as Proposition~\ref{prop:contr-cert-extenders-to-alg}.

The second conclusion is trivial, since multiple extensions of $\gamma$ would be distinct homomorphisms in $\Hom(G,H)$ that agree on $\langle \dom \gamma \rangle$, so $\mu(\langle \dom \gamma \rangle)$ lower bounds $\Lambda_{G,H}$. (Recall Proposition~\ref{prop:prelim-lambda-Hom-vs-aHom} which states that, if $\Lambda_{G,H} \neq 0$, then `$\aHom$' can be replaced by `$\Hom$' in the definition of $\Lambda_{G,H}$.) The order oracle can be used to calculate the value of $\mu(\langle \dom \gamma \rangle)$. 
\end{proof}

Recall that our main algorithmic result (Theorem~\ref{thm:mainalg}) hinges on a solution~\cite{HE} for \linebreak $\HomExt_\lambda(G,H)$  Search in the cases considered ($G = A_n$ and $H = S_m$ for exponentially bounded $m$), for the lower bound  $\lambda = 1/\binom{n}{2}$ of $\Lambda_{G,H}$. In fact~\cite{HE} provides a solution for \linebreak $\HomExt012_\lambda(G,H)$ with the same $\lambda$ (or the version of $\HomExt$ that counts solutions until some threshold). This allows a stronger lower bounding of $\Lambda_{G,H}$ as promised. \\

We discuss the ``better pruning'' consequences of these updated lower bounds on $\Lambda$. 
\begin{enumerate}
\item Better pruning of the final output list: The definition of list-decoder requires only that the output list $\outputlist$ be a superlist of the desired list $\cL = \cL(\aHom(G,H), f, \Lambda + \eps)$. The output can be pruned, using any lower bound $\lambda$ on $\Lambda$, to contain only affine homomorphisms $\varphi \in \aHom(G,H)$ that have high agreement $\agr(\varphi, f) > \lambda + \eps/2$ with $f$. The pruning can be accomplished by sampling agreement. 
\item Faster processing of certificate-lists (output of certificate-list-decoder) into output lists (output of list-decoder): While the value of $\Lambda_{G,H}$ may be unknown, the certificate-list-decoder of Theorem~\ref{thm:SRG implies CertEcon} guarantees certificates in $\cW_{\Lambda_{G,H}}$. A known lower bound $\lambda$ for $\Lambda_{G,H}$ allows pruning of partial maps $\gamma$ that do not satisfy $\mu(\langle \dom \gamma \rangle) > \lambda$. The subword extender need not be called on these maps. 
\end{enumerate}

\nocite{GKZ08, BL15, BBtoAlt}
\bibliographystyle{alpha}
\bibliography{homcodes}
\end{document}